\documentclass[11  pt, oneside,onecolumn]{article} 
\usepackage[margin=1in,lmargin=1in,rmargin=1in,bmargin=1in,tmargin=1in]{geometry}  
\usepackage[latin1]{inputenc}

\usepackage{caption}
\usepackage{subcaption}
\usepackage{cite}
\usepackage{graphicx}
\usepackage[usenames,dvipsnames]{xcolor}
\usepackage{times}
\usepackage{amsmath}
\usepackage{enumitem}

\usepackage{setspace,caption}

\usepackage[usenames,dvipsnames]{xcolor}

\usepackage{amsfonts}
\usepackage{amsmath}
\usepackage{mathrsfs}  
\usepackage{amssymb}
\usepackage{graphicx}
\usepackage{booktabs}
\usepackage{enumitem}
\usepackage{graphicx}
\usepackage{times}
\definecolor{dgreen}{rgb}{0.0, 0.5, 0.0}
\usepackage{url}
\newcommand{\h}{\textbf{h}}
\newcommand{\bs}{\textbf{b}}

\usepackage{mathtools}
\DeclarePairedDelimiter\ceil{\lceil}{\rceil}

\providecommand{\mat}[1]{\boldsymbol{\mathrm{#1}}}%
\renewcommand{\vec}[1]{\boldsymbol{\mathrm{#1}}}

\providecommand{\mB}{\ensuremath{\mat{B}}}

\providecommand{\mD}{\ensuremath{\mat{D}}}
\providecommand{\mE}{\ensuremath{\mat{E}}}

\providecommand{\mI}{\ensuremath{\mat{I}}}

\providecommand{\mR}{\ensuremath{\mat{R}}}

\providecommand{\va}{\ensuremath{\vec{a}}}

\providecommand{\ve}{\ensuremath{\vec{e}}}

\providecommand{\vp}{\ensuremath{\vec{p}}}

\providecommand{\vx}{\ensuremath{\vec{x}}}
\providecommand{\vy}{\ensuremath{\vec{y}}}

\usepackage{xparse}

\usepackage{enumitem}

\usepackage{amsthm}
\newtheorem{theorem}{Theorem}
\newtheorem{proposition}[theorem]{Proposition}
\newtheorem{lemma}[theorem]{Lemma}

\newtheorem*{definition}{Definition}

\newtheorem{retheorem}{Theorem}
\newtheorem{reproposition}[retheorem]{Proposition}
\newtheorem{relemma}[retheorem]{Lemma}

\newtheorem{silemma}{Supplemental Lemma}

\providecommand{\vg}{\bs}
\newcommand{\pr}{\text{Pr}}
\newcommand{\ncdot}{}

\usepackage{url}

\usepackage{soul}

\usepackage{hyperref}
\hypersetup{
	colorlinks=true, 
	linktoc=all,     
	linkcolor=blue,  
	citecolor = red,
}

\author
{Nate Veldt,$^{1}$ Austin R.\ Benson,$^{2}$ Jon Kleinberg$^{2}$\\
	\\
	\normalsize{$^{1}$Department of Computer Science and Engineering, Texas A\&M University}\\
	\normalsize{$^{2}$Department of Computer Science, Cornell University}\\
	\\
}
\date{}

\begin{document}
	
	\title{Combinatorial Characterizations and Impossibilities for Higher-order Homophily} 
	\maketitle

	\begin{abstract}
		Homophily is the seemingly ubiquitous tendency for people to connect and interact with other individuals who are similar to them. This is a well-documented principle and is fundamental for how society organizes.
		Although many social interactions occur in groups, homophily has traditionally been measured using a graph model, which only accounts for pairwise interactions involving two individuals. Here, we develop a framework using hypergraphs to quantify homophily from group interactions. 
		This reveals natural patterns of group homophily that appear with gender in scientific collaboration and political affiliation in legislative bill co-sponsorship, and also reveals distinctive gender distributions in group photographs, all of which cannot be fully captured by pairwise measures. At the same time, we show that seemingly natural ways to define group homophily are combinatorially impossible. This reveals important pitfalls to avoid when defining and interpreting notions of group homophily, as higher-order homophily patterns are governed by combinatorial constraints that are independent of human behavior but are easily overlooked.
	\end{abstract}
	
	\section*{Introduction}
	Homophily is the established sociological principle that individuals tend to associate and form connections with other individuals that are similar to them~\cite{lazarsfeld1954}.
	For example, social ties are strongly correlated with demographic factors such as race, age, and gender~\cite{moody2001race,shrum1988friendship,marsden1988homogeneity}; acquired characteristics such as education, political affiliation, and religion~\cite{marsden1988homogeneity,loomis1946political}; and even psychological factors such as attitudes and aspirations~\cite{almack1922influence,richardson1940community}.
	For many of these factors, homophily persists across a wide range of relationship types, from marriage~\cite{kalmijn1998}, to friendships~\cite{verbrugge1983research}, to ties based simply on whether individuals have been observed together in public~\cite{mayhew1995}. As a consequence, homophily serves as an important concept for understanding human relationships and social connections, and is a key guiding principle for research in sociology and network analysis.
	
	A major motivation in homophily research is to understand how similarity among individuals influences group formation and group interactions~\cite{cohen1977sources,mcpherson2001birds,mcpherson_rotolo,mcpherson1987homophily,mayhew1995,breiger1974duality}.
	This emphasis on group interactions is natural, given how much of life and society is organized around \emph{multiway} relationships and interactions, such as work collaborations, social activities, volunteer groups, and family ties. However, despite the ubiquity of multiway interactions in social settings, existing homophily measures rely on a graph model of social interactions, which encodes only two-way relationships between individuals. In order to measure homophily in group interactions, these approaches typically reduce group participation to pairwise relationships, based on co-participation in groups. While this simplifies the analysis, it discards valuable information about the exact size and make-up of groups in which individuals choose to participate. 
	
	Here, we present a mathematical framework for measuring homophily in \emph{higher-order}, multiway interactions that quantifies the extent to which individuals in a certain class participate in groups with varying numbers of in-class and out-class participants. This relies on a new hypergraph model for representing group interactions, which generalizes the standard graph model. In the graph setting, homophily can be measured by comparing a graph homophily index against a baseline score~\cite{altenburger2018monophily}. We generalize this and show that there are many intuitive ways to quantify tendencies towards same-class interactions in group settings. 
	{One simplistic example of higher-order homophily is a hypergraph where hyperedges only involve nodes from a single class. As an example, this could represent social interactions where every social group involves only men or only women. However, this matches very few real-world examples and fails to capture other intuitive types of same-class mixing patterns in group interactions. For example, an individual may tend to participate mainly in groups where at least the majority of members are from the same class, even if no groups are completely homogeneous with respect to class. Alternatively, one's participation in group interactions may increase in proportion to the number of same-class members in the group. Our framework allows us to quantify many of these generalized notions of group homophily, and identify patterns in group interactions that cannot be captured by existing graph measures. 
		At the same time, we prove fundamental combinatorial limits on the extent to which these generalized notions of same-class group mixing patterns can be exhibited in practice. In particular, we prove that certain seemingly natural approaches for generalizing graph measures of homophily to the hypergraph setting are in fact overly restrictive, and cannot be satisfied by any hypergraph because of combinatorial impossibilities that are independent of human preferences and choices.
		
		Our framework provides a very general way to define and measure higher-order homophily, and our combinatorial impossibilities shed light on empirical observations that would otherwise be hard to explain or easy to misinterpret.
		For example, our framework captures higher-order homophily present in group interactions defined by legislative bill co-sponsorship among US members of congress. 
		One intuitive and unsurprising observation captured by our framework is a higher-than-random tendency for members of congress to co-sponsor bills that are mostly (even if not exclusively) co-sponsored by members of their same political party.
		The deeper and less obvious insight revealed by our framework is that in order for both parties to simultaneously exhibit this behavior, there must be a significant number of members from each party that are willing to co-sponsor bills even when their party is in the minority.} In fact, 
	contrary to what a naive understanding of group homophily might suggest, 
	both political parties must exhibit much higher tendencies to participate in bills where their party is overwhelmingly outnumbered, in comparison with bills where their party actually has a {slight majority}. 
	Our combinatorial impossibility results are key to understanding and interpreting these empirical results. Without them, it would be tempting to conclude that only weak notions of group homophily are satisfied in legislative bill-cosponsorship. However, our theoretical results reveal that group homophily is strongly exhibited in this setting. In fact, outside of unrealistic extreme cases where nearly all group interactions are perfectly homogeneous, stronger notions of group homophily are combinatorially impossible. 
	
	We similarly use our framework to reveal empirical differences in co-authorship patterns between men and women in academic publishing. These should be interpreted and understood in light of our theoretical results, as many of these differences are due to combinatorial constraints that must be satisfied independent of social factors. Finally,  our framework allows us to uncover meaningful patterns in group interactions that cannot be detected by graph homophily measures. As an example, we use our framework to reveal different gender distribution patterns in group pictures, depending on picture context and group size, which are overlooked by graph measures.
	
	\section*{Group Affinities and Hypergraph Homophily}
	\begin{figure*}[t!]
		\centering
		\includegraphics[width=.8\linewidth]{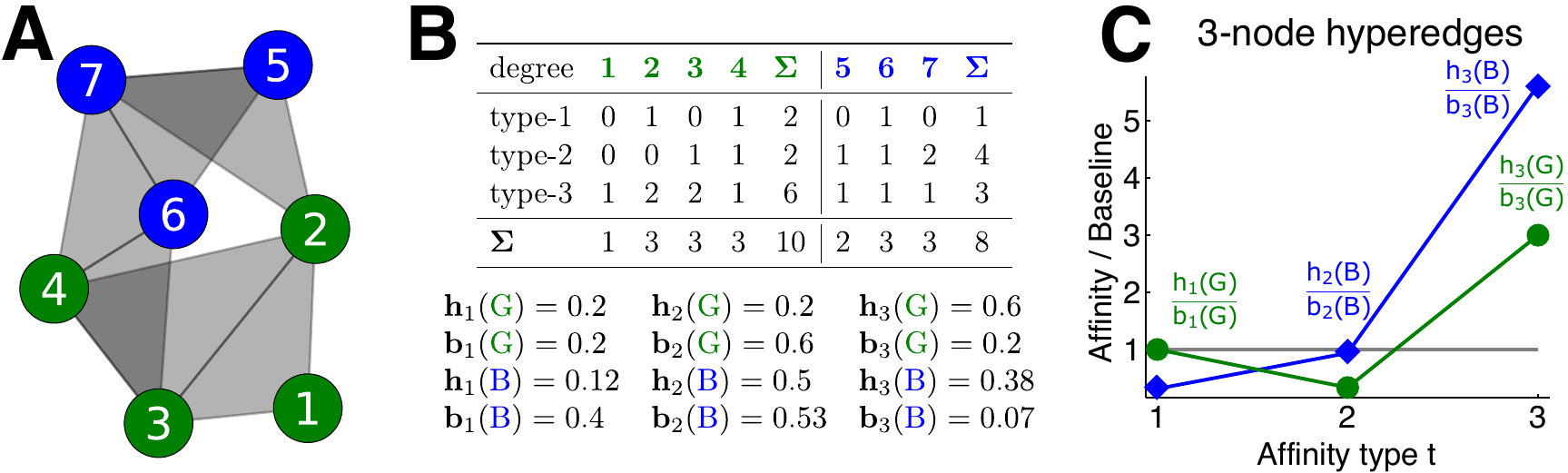} 
		\caption{\small {\textbf{Hypergraph affinity and ratio scores for a small hypergraph.}}
			(\textbf{A}) An example of a set of size-3 group interactions between 2 classes, modeled by a small hypergraph. {Each triangle in the figure indicates a 3-way group interaction, and color indicates node class (green or blue)}. (\textbf{B}) Degrees, affinity scores, and baseline scores for the hypergraph. 
			A node's type-$t$ degree is the number of groups it belongs to in which exactly $t$ nodes are from its class. The type-$t$ degree for an entire class is the sum of type-$t$ degrees across all nodes in the class 
			($\Sigma$ columns in the table). The type-$t$ affinity for a class $X$, denoted by $\h_t(X)$, is the ratio between the class's type-$t$ degree and its total degree (the sum of type-$t$ degrees for all $t$). 
			The type-$t$ baseline score $\bs_t(X)$ is the probability that a node joins a type-$t$ group if other nodes are selected uniformly at random. 
			(\textbf{C}) The ratios of affinity to baseline (ratio scores) summarize the overall group participation rates for both types of nodes. 
		}
		\label{fig:toy} 
	\end{figure*}
	
	To introduce our framework, we consider a hypergraph $H = (V,E)$ (Fig.~\ref{fig:toy}A), where the node set $V$ represents individuals in a population (e.g., students in a school, researchers in academia, employees at a company). Each hyperedge $e \in E$ represents a group interaction among members of the population. The set $X \subseteq V$ indicates a set of individuals with the same class label (e.g., gender, political affiliation). We would like to quantify the extent to which individuals in this class tend to interact and connect with one another in group settings. 	{For our mathematical formulation, we can treat $H = (V,E)$ as a $k$-uniform hypergraph, meaning that each hyperedge represents a multiway relationship among exactly $k$ nodes. In practice, the same measures can be applied to each group size $k$ separately, in order to analyze homophily patterns that are exhibited in different ways across different groups sizes.}
	
	\subsection*{Measuring group affinities} In order to measure how class label affects group interactions
	of a fixed size $k$, we define for each positive integer $t \in [k] = \{1, 2, \hdots , k\}$ a type-$t$ \emph{affinity score}, summarizing the extent to which individuals in class $X$ participate in groups where exactly $t$ group members are in class $X$ (Fig.~\ref{fig:toy}B). To define the affinity score, we first define the total degree of a node $v$, denoted by $d(v)$, to be the number of groups it participates in. Its type-$t$ degree, denoted by $d_t(v)$, is the number of these groups with exactly $t$ members from $v$'s class, {including $v$ itself}. The \emph{class} degree $D(X)$ is the sum of total degrees across all nodes in $X$, and the class type-$t$ degree $D_t(X)$ is the sum of individual type-$t$ degrees. The ratio of these values defines the type-$t$ affinity score:
	\begin{equation}
		\label{affinities}
		\h_t(X) = \frac{D_t(X)}{D(X)} = \frac{\sum_{v \in X} d_t(v)}{\sum_{v \in X} d(v)}.
	\end{equation}
	When $k = t = 2$, this ratio is the well-studied homophily index of a graph~\cite{altenburger2018monophily}, the fraction of same-class friendships for class $X$. This index can be statistically interpreted as the maximum likelihood estimate for a certain homophily parameter when a logistic-binomial model is applied to the degree data. An analogous result is also true for our more general hypergraph affinity score (see the appendix).
	
	\subsection*{Baseline scores for group affinities}
	In order to determine whether an affinity score $\h_t(X)$ is meaningfully high or low, we compare it against a baseline score $\bs_t(X)$ representing a null probability for type-$t$ interactions. If $\h_t(X) > \bs_t(X)$, this will indicate that type-$t$ group interactions are \emph{overexpressed} for class $X$. This is analogous to the notion of \emph{inbreeding homophily} in traditional social network analysis---the tendency for individuals to connect with other similar individuals more than would be expected by chance~\cite{currarini2009economic,mcpherson2001birds}. Given a set of baseline scores, one way to summarize the group participation for a class $X$ is to  plot a sequence of ratio scores $\frac{\h_t(X)}{\bs_t(X)}$ for $t \leq k$ (Fig.~\ref{fig:dblp}C). Ratio scores near one indicate that the class distribution in group interactions is roughly what would be expected at random.
	
	The standard baseline score we consider (and which we use in Fig.~\ref{fig:dblp}C) is the probability that a class-$X$ node joins a group where $t$ members are from class $X$, if $k-1$ other nodes are selected uniformly at random. Formally this is given by
	\begin{equation}
		\label{base1}
		\hat{\bs}_t(X) = \frac{{|X|-1 \choose t-1} {n-|X| \choose k-t}} {{n-1\choose k-1}},
	\end{equation}
	where $n = |V|$ is the number of nodes in the hypergraph. In the graph setting ($k = 2$), the homophily index $\h_2(X)$ is typically compared against $\alpha_X = |X|/n$, the proportion of nodes in class $X$~\cite{altenburger2018monophily}. Observe that this is simply an asymptotic version of the standard baseline score $\hat{\bs}_2(X) = (|X|-1)/(n-1)$.
	In practice it is often useful to use asymptotic baseline scores for general $k$ and $t$ by fixing the class proportion $\alpha_X = |X|/n$ and computing $\lim_{n\rightarrow\infty} \hat{\bs}_t(X).$
	
	We provide two additional intuitive interpretations for this standard baseline score. The first is that baseline scores correspond to affinity scores for a complete $k$-uniform hypergraph. The second is that generating random hyperedges without regard for node class produces a hypergraph whose ratio scores asymptotically converge to 1. We prove these results in the appendix.
	\begin{proposition}
		\label{prop:base1}
		Let $H_{k,n}^* = (V,E)$ be the complete $k$-uniform hypergraph on $n$ nodes.
		The type-$t$ affinity score for class $X  \subseteq V$ equal the type-$t$ baseline score in~\eqref{base1}.
	\end{proposition}
	\begin{proposition}
		\label{prop:base2}
		Fix any $p \in (0,1)$ and a positive integer $k$, and let $H = (V,E)$ be a random hypergraph on $n$ nodes that is formed by turning each $k$-tuple of nodes in $V$ into a hyperedge with probability $p$. As $n \rightarrow \infty$, the ratio scores for a class $X \subseteq V$ with $|X| = \Theta(n)$ converge in probability to 1.
	\end{proposition}
	While the baseline scores in~\eqref{base1} are a natural choice for a number reasons, it is also possible to define and consider baselines corresponding to other null probabilities for affinity scores. The combinatorial impossibility results we prove later will in fact apply to a more general class of \emph{realizable scores}. A set of baseline scores $\{\bs_t(X) \colon  t \in [k]\}$ is \emph{realizable} if they are all positive and there exists a hypergraph whose affinity scores are exactly these baselines scores. Proposition~\ref{prop:base1} indicates that the scores in~\eqref{base1} are realizable.
	
	\subsection*{Defining higher-order homophily}
	Informally, homophily is the tendency for individuals in one class to disproportionately interact and connect with other individuals in the same class. In a graph setting, this can be measured by checking whether nodes tend to form links with other same-class nodes more often than random. We consider three natural ways to extend this to the hypergraph setting. 
	
	One simplistic way to check for group homophily is to see whether a class has a higher-than-baseline affinity for group interactions that \emph{only} involve members of their class. Formally, this means that $\h_t(X) > \bs_t(X)$ for $t = k$. We refer to this as \emph{simple} homophily. This captures one valid notion of hypergraph homophily, but is very restrictive and fails to capture other intuitive types of same-class mixing patterns. This includes high affinities for groups where at least a majority of members are from the same class, or group participation levels that increase in proportion to the number of same-class members in a group.
	
	
	{ In order to capture more nuanced notions of group homophily, we say that class $X$ exhibits \emph{order-$j$ majority homophily} if the top $j$ affinity scores for this class are higher than baseline, i.e., $\h_{k-j+1}(X) > \bs_{k-j+1}(X), \h_{k-j+2}(X) > \bs_{k-j+2}(X) , \hdots, \h_{k}(X) >\bs_{k}(X)$. Simple homophily is the special case of order-$1$ majority homophily, while larger values of $j$ capture generalized and stronger notions of homophily. We say that a class has \emph{order-$j$ monotonic homophily} if each of the top $j$ ratio scores is larger than the ratio score that comes before it. Formally, this means $\h_t(X) / \bs_t(X) > \h_{t-1}(X) / \bs_{t-1}(X)$ for $t \geq k-j+1$. In other words, ratio scores from $\h_{k-j}(X)/ \bs_{k-j}(X)$ to $\h_{k}(X)/ \bs_{k}(X)$ are strictly increasing. Finally, we say that the \emph{majority homophily index} (MaHI) of a class $X$ is the maximum $j$ such that order-$j$ majority homophily holds, and the \emph{monotonic homophily index} (MoHI) is the maximum $j$ such that order-$j$ monotonic homophily holds.}

	\subsection*{Strict notions of higher-order homophily}
	We also consider two basic measures for majority and monotonic group participation that seem intuitive at first, but which we will prove are much more restrictive than they appear. We say that class $X$ exhibits \emph{strict majority homophily} if the class has higher than random affinities for groups where they constitute a majority, i.e., $\h_t(X) > \bs_t(X)$ when $t > k-t$. {This is equivalent to order-$(\ceil{k/2})$ majority homophily}. Similarly, class $X$ exhibits \emph{strict monotonic homophily} if, for groups where $X$ constitutes a majority, its ratio scores are strictly increasing as the number of class-$X$ members grows, i.e., 
	$\h_t(X) / \bs_t(X) > \h_{t-1}(X) / \bs_{t-1}(X) $ when $t > k-t$. This is equivalent to
	{order-$(\ceil{k/2})$ monotonic homophily}. For groups of size two, the notions of simple homophily, strict majority homophily, and strict monotonic homophily all reduce to checking whether the homophily index of a graph is higher than the relative class size. This is a standard way of checking for graph homophily~\cite{altenburger2018monophily}. In graph-based analysis, it is typical for multiple node classes in a social network to exhibit homophily at the same time. A natural question is whether this observation generalizes to the hypergraph setting.
	
	\subsection*{Gender affinities in co-authorship data}
	\begin{figure*}[t!]
		\centering
		\includegraphics[width=.75\linewidth]{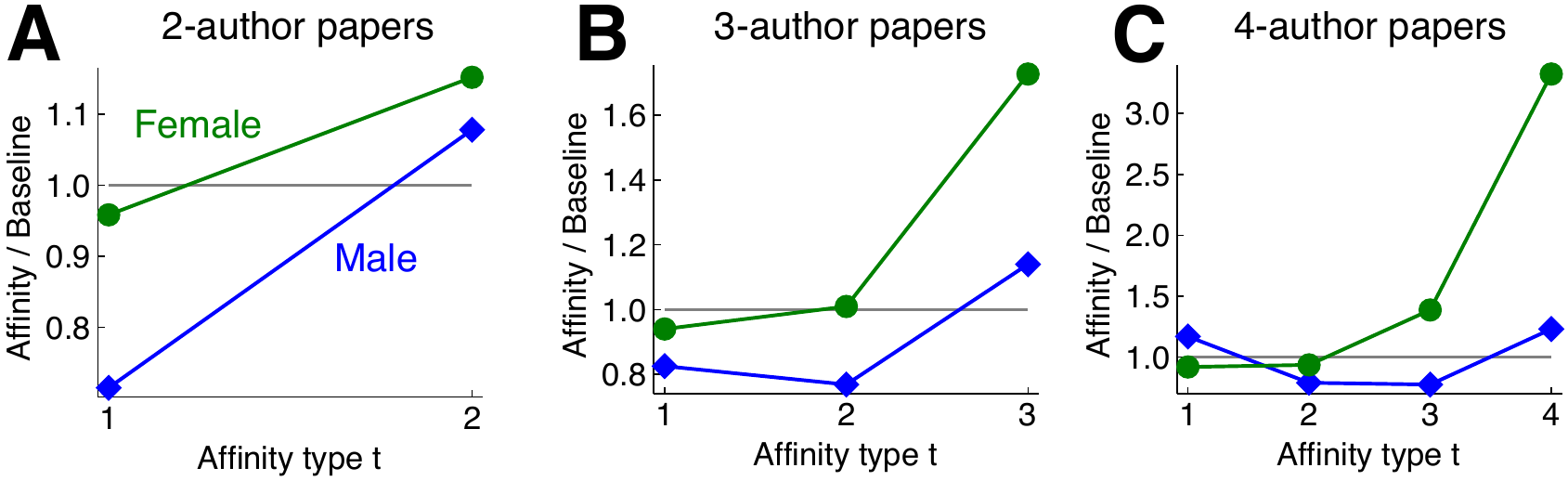} 
		\caption{\small \textbf{
				Ratio scores with respect to gender for groups defined by co-authorship in computer science publications}. (\textbf{A}) Collaborating with same-gender co-authors on 2-person papers is more likely than expected by chance, as seen by ratio scores higher than 1. (\textbf{B,C}) For 3- and 4-author papers, the ratio score curves are substantially different for men and women. Female authors exhibit monotonically increasing scores, whereas male authors do not. Our theoretical results show that many of these differences are due simply to combinatorial constraints on hypergraph affinity scores. If we reduce the set of 2-4 author papers to pairwise co-authorships and apply graph-based measures, women and men have graph homophily indices of 0.261 and 0.828 respectively. These scores are higher than group proportions of 0.215 and 0.785, and therefore reveal some level of gender homophily. Similar graph homophily indices are obtained if we also include papers with more authors. However, this provides less information than knowing the full range of hypergraph affinity scores, and fails to uncover the nuanced differences in co-authorship patterns between men and women.
		}
		\label{fig:dblp} 
	\end{figure*}
	As a first example, we measure hypergraph affinity scores with respect to gender in academic collaborations, where nodes represent researchers and each hyperedge indicates co-authorship on a paper published at a computer science conference.  Our framework reveals differences in co-author patterns for men and women (Fig.~\ref{fig:dblp}). Both men and women have overexpressed tendencies for being authors on papers that only involve authors of their same gender. In other words, for collaborations of size two to four, both genders exhibit simple homophily. 
	For two-author papers, both genders exhibit strict monotonic homophily and strict majority homophily, as these coincide with simple homophily.
	For three- and four-author papers, women exhibit both strict majority and strict monotonic homophily, but men do not exhibit either. 
	{ These definitions seem to capture intuitive higher-order notions of homophily, and if we restrict them to the graph setting, we recover existing notions of homophily that are often satisfied by multiple classes at once. How, then, can we explain the differences between men and women in this dataset? It is tempting to wonder whether these differences are purely due to social factors. In other contexts, can we expect both men and women to exhibit strict monotonic and strict majority homophily? Our main theoretical results will show, perhaps surprisingly, that this is in fact impossible for \emph{any dataset}, and that many of the differences we see between men and women in our co-authorship results must exist simply because of combinatorial inevitabilities. It is not immediately clear why this should be, given that men and women have separate affinity scores as well as separate baseline scores. These combinatorial limits provide a deeper understanding of how higher-order homophily can be manifested in practice. This also highlights pitfalls and misunderstandings to avoid when drawing conclusions about the presence or level of group homophily in different contexts.}
	
	
	\section*{Impossibility Results for Strict Hypergraph Homophily}
	{Our main theoretical results highlight combinatorial constraints that govern higher-order mixing patterns in hypergraphs. These are easily overlooked, but are crucial for properly defining and understanding higher-order homophily. This also reveals a fundamental difference between measuring homophily in group settings and measuring homophily in graphs, as these impossibilities do not apply to the graph setting. Although strict monotonic and strict majority homophily seem to capture intuitive notions of same-class mixing patterns, we show that it is combinatorially impossible for two classes to \emph{simultaneously} exhibit either of these types of homophily in the hypergraph setting (subject to a small additional constraint if groups have an even number of members). In other words, even if all individuals \emph{preferred} to participate in group interactions that are monotonically or majority homophilous with respect to their class, this cannot be accomplished.} 
	
	We formalize our results as a set of combinatorial impossibilities for two-class, $k$-uniform hypergraphs. 
	\begin{definition}
		A hypergraph $H = (V,E)$ is a two-class, $k$-uniform hypergraph if $|e| = k$ for every $e \in E$, and there exist two node classes $\{A,B\}$ such that $V = A \cup B$ and $A\cap B = \emptyset$. 
	\end{definition}
	Although our theorems focus on this family of hypergraphs, our framework and results have important implications for understanding homophily in general group settings. When groups vary in size, our results can be applied to each group size separately, to understand which behaviors are possible or impossible in each case. We take this approach when measuring affinity scores on real datasets involving group interactions of varying size $k$. For a hypergraph with more than two node labels, our results imply combinatorial impossibilities for an arbitrary class $A = X \subseteq V$ relative to the collective behavior of all other classes, joined by a single ``out class'' label $B = V \backslash X$. 
	
	\subsection*{Equivalent characterization of affinity scores}
	In proving our impossibility results for a $k$-uniform hypergraph with classes $A$ and $B$, it will be convenient to categorize hyperedges based on the number of nodes from each class. Although type-$t$ degrees and type-$t$ affinity scores are defined relative to a given class $X$, we define hyperedge types in an absolute sense: for $j \in \{0, 1, 2, \hdots k\}$, a hyperedge $e \in E$ is of type-$j$ if it contains exactly $j$ nodes \emph{specifically from class $A$}. We denote the number of type-$j$ hyperedges by $m_j$. This allows us to write type-$t$ affinity scores in terms of absolute hyperedge counts $\{m_0, m_1, \hdots, m_k\}$.
	The type-$t$ affinities for $A$ and $B$ are then
	\begin{equation}
		\label{eq:edgeaffinities}
		\h_t(A) = \frac{tm_t}{\sum_{i = 1}^k i m_i}, \text{ and } \h_t(B) = \frac{tm_{k-t}}{\sum_{i = 1}^k i m_{k-i}}.
	\end{equation}
	Observe that in the numerator of $\h_t(A)$, we scale $m_t$ by $t$ to account for the fact that each type-$t$ hyperedge involves $t$ nodes from class $A$, and therefore contributes to the degree of $t$ different nodes in class $A$. Meanwhile, type-$(k-t)$ hyperedges involve $t$ nodes from class $B$, leading to the expression for $\h_t(B)$.

	\subsection*{Impossibility results for strict monotonic homophily}
	We begin with an impossibility result for monotonic homophily. This has a comparatively simple proof that relies on considering two contradictory inequalities that result from assuming two classes exhibit homophily. We separate our results based on whether $k$ is odd or even.
	\begin{theorem}
		\label{thm:oddkmonotonic}
		Let $H$ be a two-class, $k$-uniform hypergraph and $\{\vg_i(X)\colon i \in [k], X \in \{A,B\} \}$ be realizable baseline scores. For odd $k$, it is impossible for both classes to exhibit strict monotonic homophily.
	\end{theorem}
	\begin{proof}
		If both classes $A$ and $B$ exhibit strict monotonic homophily, then the following two sequences of inequalities hold:
		\begin{align}
			\label{classa}
			\text{Class A: } &\,\, \frac{\h_k(A)}{\vg_k(A)} > \frac{\h_{k-1}(A) }{\vg_{k-1}(A)} > \cdots > {\frac{\h_{r}(A)}{\vg_{r}(A)} > \frac{\h_{r-1}(A)}{\vg_{r-1}(A)}  },\\
			\label{classb}
			\text{Class B: } &\,\, \frac{\h_k(B)}{\vg_k(B)} > \frac{\h_{k-1}(B) }{\vg_{k-1}(B)} > \cdots > {\frac{\h_{r}(B)}{\vg_{r}(B)} > \frac{\h_{r-1}(B)}{\vg_{r-1}(B)}  },
		\end{align}
		where $r = (k+1)/2$. 
		Using the characterization of affinity scores given in~\eqref{eq:edgeaffinities}, the last inequality in each of~\eqref{classa} and~\eqref{classb} can be rearranged as follows:
		\begin{align}
			\label{firstA}
			{\frac{\h_{r}(A)}{\vg_{r}(A)} > \frac{\h_{r-1}(A)}{\vg_{r-1}(A)}  } &\iff \frac{rm_r}{\vg_{r}(A)} > \frac{(r-1)m_{r-1}}{\vg_{r-1}(A)},\\
			\label{secondB}
			{\frac{\h_{r}(B)}{\vg_{r}(B)} > \frac{\h_{r-1}(B)}{\vg_{r-1}(B)}  } &\iff \frac{rm_{r-1}}{\vg_{r}(B)} > \frac{(r-1)m_{r}}{\vg_{r-1}(B)}.
		\end{align}
		Above, we have used the observation that $m_{k-r} = m_{r-1}$ and $m_{k-(r-1)} = m_{r}$, in order to write both inequalities in terms of $m_r$ and $m_{r-1}$.
		Since the baseline scores are realizable, there exists some two-class $k$-uniform hypergraph $G$ whose affinity scores equal the baseline scores. Letting $M_t$ denote the number of type-$t$ hyperedges in $G$, we can write the baseline scores as
		\begin{align*}
			\vg_r(A)&= \frac{r M_r}{D_A}, \;  \vg_{r-1}(A)= \frac{(r-1) M_{r-1}}{D_A}, \\
			\vg_r(B) &= \frac{r M_{r-1}}{D_B}, \;   \vg_{r-1}(B) = \frac{(r-1) M_{r}}{D_B},
		\end{align*}
		where $D_A = \sum_{i = 1}^k i M_i$ and $D_B = \sum_{i = 1}^k i M_{k-i}$.
		Applying a few steps of algebra to the inequality on the right of~\eqref{firstA} shows that if $A$ exhibits monotonic homophily, then
		\begin{align*}
			\frac{rm_r}{\vg_{r}(A)} > \frac{(r-1)m_{r-1}}{\vg_{r-1}(A)}  \iff \frac{m_r}{M_r } > \frac{m_{r-1}}{M_{r-1}}.
		\end{align*}
		Meanwhile, the inequality in~\eqref{secondB} implies the exact opposite:
		\begin{align*}
			\frac{rm_{r-1}}{ \vg_r(B) } > \frac{(r-1)m_{r}}{ \vg_{r-1}(B)} \iff \frac{m_{r-1}}{M_{r-1}}  >\frac{m_r}{M_r }. 
		\end{align*}
		Thus, assuming both classes exhibit strict monotonic homophily leads to a contradiction.
	\end{proof}
	
	Strict monotonic homophily is in fact possible for two classes at once if $k$ is even. This can happen, for example, by starting with a complete $k$-uniform hypergraph and deleting all type-$k/2$ hyperedges, if we are specifically considering standard baseline scores from~\eqref{base1}. However, an analogous impossibility result holds if we add one extra assumption. The proof follows the same strategy as the proof of Theorem~\ref{thm:oddkmonotonic}, after adding an extra inequality for one class.
	\begin{theorem}
		\label{thm:evenkmonotonic}
		Let $H$ be a two-class, $k$-uniform hypergraph and $\{\vg_i(X)\colon i \in [k], X \in \{A,B\} \}$ be realizable baseline scores. If $k$ is even, then it is impossible for both classes to satisfy strict monotonic homophily if additionally $\frac{\h_\ell(X)}{\vg_\ell(X)} > \frac{\h_{\ell-1}(X)}{\vg_{\ell-1}(X)} $ for one class $X \in \{A,B\}$, where $\ell = k/2$.
	\end{theorem}
	{Theorems~\ref{thm:oddkmonotonic} and~\ref{thm:evenkmonotonic} lead to other impossibility results as direct corollaries. Our definition of strict monotonic homophily for a class $X$ is defined specifically for groups where $X$ is in the majority. If we remove this restriction and consider a strict notion of monotonic homophily that requires $\h_t(X)/ \vg_t(X) > \h_{t-1}(X)/\vg_{t-1}(X)$ for \emph{all} $t \leq k$, then for every $k > 2$ this is impossible for two classes simultaneously whether or not $k$ is odd. We also see from the proof of Theorem~\ref{thm:oddkmonotonic} that a contradiction results from assuming both classes have increasing ratio scores when going from type-$(r-1)$ to type-$r$ affinities. Thus, any notion of homophily involving this assumption is impossible for two classes at once, regardless of what happens with other ratio scores.}
	
	\subsection*{Impossibility results for strict majority homophily}
	Next we turn to extremal results for majority homophily.
	\begin{theorem}
		\label{thm:majorityimpossible}
		Let $H = (V,E)$ be a two-class $k$-uniform hypergraph and $\{\bs_i(X)\colon  i\in [k], X \in \{A,B\} \}$ be realizable baseline scores.
		\begin{itemize}[itemsep = 0pt]
			\item If $k$ is odd, it is impossible for both classes to simultaneously exhibit strict majority homophily.
			\item If $k$ is even, it is impossible for both classes to exhibit strict majority homophily if additionally $\h_{k/2}(X) > \bs_{k/2}(X)$ for one of the classes $X \in \{A,B\}$.
		\end{itemize} 
	\end{theorem}
	Although our results for strict majority homophily closely mirror our results for strict monotonic homophily, Theorem~\ref{thm:majorityimpossible} is significantly more challenging to show and requires a different and more in-depth proof technique. A full proof is provided in the appendix. We provide a detailed proof sketch of the result for odd values of $k$. The same overall strategy yields the result for even $k$. 
	
	For odd $k$ and $r = (k+1)/2$, assuming both classes exhibit strict majority homophily is equivalent to satisfying two sets of inequalities:
	\begin{align}
		\label{majinequalitiesA}
		\h_t(A) &> \bs_t(A) \text{ for $t > k-t$}\\
		\label{majinequalitiesB}
		\h_t(B) &> \bs_t(B) \text{ for $t > k-t$}.
	\end{align}
	Using the characterization of affinity scores given in~\eqref{eq:edgeaffinities}, we can rearrange these inequalities to yield equivalent inequalities in terms of typed hyperedge counts $\{m_0, m_1, \hdots, m_k\}$:
	\begin{align}
		\label{eq:mt}
		m_t &> \frac{\vg_t(A)\sum_{i \neq t} i m_i }{t \cdot (1-\vg_t(A))}, &\text{for } t = r, r+1, \hdots, k \\
		\label{eq:ms}
		m_s &> \frac{\vg_{k-s}(B) \sum_{i \neq k-s} i m_{k-i}}{(k-s) \cdot (1-\vg_{k-s}(B))}, &\text{for } s = 0, 1, \hdots, r-1.
	\end{align}
	{Our aim is to show that all of the above inequalities cannot hold simultaneously.
		
		It is important to note that although every hyperedge type is bounded below by one of the inequalities in~\eqref{eq:mt} and~\eqref{eq:ms}, this does not immediately imply any contradiction. 
		This is because half of the hyperedge types are bounded below in terms of baseline scores for class $A$, while the other half are bounded below in terms of baseline scores for class $B$.} Baseline scores for $A$ and $B$ can differ significantly if these classes differ in size, even if we consider the very special case of standard scores given by~\eqref{base1}. {We ultimately show that these inequalities cannot be satisfied simultaneously for \emph{any} set of realizable baseline scores, by analytically finding solutions to a linear program encoding the maximum amount of homophily that can be satisfied by two classes at once.} Because of the complexity of this proof, we begin by considering other approaches that seem simple and natural at first but ultimately fail to prove the main result.

	We first of all note that it is simple to show that hypergraph affinity scores for a \emph{single} class $X$ cannot all be simultaneously above baseline, i.e., $\h_t(X) > \vg_t(X)$ for all $t \in [k]$. Since hyperedge affinity scores sum to one, and baseline scores do as well, summing both sides of the inequality for all $t$ leads to an immediate contradiction. This argument does not apply if we assume two classes exhibit strict majority homophily, since in this case only some types of interactions are above the baseline for class $A$, while other types of interactions are above baseline for class $B$. A next approach for trying to prove Theorem~\ref{thm:majorityimpossible} is to sum up the left and right hand sides of the hyperedge inequalities in~\eqref{eq:mt} and~\eqref{eq:ms}, and see if this leads to a contradiction. However, this also does not work. For example, let $k = 3$ and assume we use standard baseline scores from~\eqref{base1}. If classes are equal in size, then $\hat{\bs}_t = \hat{\bs}_t(A) = \hat{\bs}_t(B)$ for $t \in \{1,2,3\}$. Summing both sides of inequalities~\eqref{eq:mt} and~\eqref{eq:ms} leads to a new inequality
	\begin{equation*}
		3(m_0 + m_3) + 2(m_1 + m_2) > 3(m_0 + m_1 + m_2 + m_3)(\hat{\bs}_2 + \hat{\bs}_3).
	\end{equation*}
	This is satisfied by any hypergraph where $m_1 = m_2 = 0$, so it does not provide the contradiction we are looking for.

	The proof of Theorem~\ref{thm:oddkmonotonic} shows that a {subset} of the strict monotonic homophily inequalities (in fact, just two of them) leads to a contradiction. Therefore, another natural strategy for trying to prove Theorem~\ref{thm:majorityimpossible} is to see whether a subset of the inequalities given by~\eqref{eq:mt} and~\eqref{eq:ms} contradict each other. In the appendix, we prove the following result, which rules out this possibility.
	\begin{proposition}
		\label{prop:needall}
		If any inequality from~\eqref{eq:mt} and~\eqref{eq:ms} is discarded, it is possible to construct a two-class $k$-uniform hypergraph satisfying the remaining inequalities.
	\end{proposition}
	This means that any strategy similar to the proof for Theorem~\ref{thm:oddkmonotonic} will fail for Theorem~\ref{thm:majorityimpossible}. Instead, any proof for Theorem~\ref{thm:majorityimpossible} will need to incorporate every one of the inequalities in~\eqref{eq:mt} and~\eqref{eq:ms} if we are to show that strict majority homophily is impossible. 
	
	
	\paragraph{Capturing extremal limits of homophily via linear programming}
	Having ruled out simpler strategies for proving Theorem~\ref{thm:majorityimpossible}, we now outline a linear programming framework for checking the maximum amount of homophily that can be exhibited in a set of group interactions, subject to different constraints on higher-order affinity scores. This first of all provides a general framework for numerically checking whether different extremal notions of higher-order homophily can be satisfied or not. We will also show how to use analytical solutions and linear programming duality to fully prove Theorem~\ref{thm:majorityimpossible}.
	
	We specifically consider a linear program (LP) that encodes the maximum amount of majority homophily that can be satisfied by two classes $A$ and $B$ simultaneously in a $k$-uniform two-class hypergraph. { This LP is given by
		\begin{equation}
			\label{lpmain}
			\begin{array}{llr}
				\text{max} & \gamma \\
				\text{s.t.}& \sum_{i = 0}^k x_i = 1 \\
				& t\cdot x_t - \vg_t(A) \cdot \sum_{i = 0}^k i \cdot x_i \geq \gamma \hspace{1cm}\text{ for $t \in \{r, \hdots, k\}$}   \\
				&t\cdot  x_{k-t} - \vg_t(B) \cdot \sum_{i = 0}^k i \cdot x_{k-i}\geq \gamma \hspace{.3cm} \text{ for $t \in \{r, \hdots, k\}$}    \\
				& x_i \geq 0 \hspace{4.32cm}\text{ for $i \in \{0\} \cup [k]$}. \\
			\end{array}
	\end{equation}}
	In this LP, there is a variable $x_i \geq 0$ for each type of hyperedge in some hypergraph. The constraint $\sum_{i = 0}^k x_i = 1$ encodes the fact that $x_i$ in fact represents the proportion of hyperedges that are of type-$i$. The constraint
	\begin{equation*}
		t\cdot x_t - \vg_t(A) \cdot \sum_{i = 1}^k i \cdot x_i \geq \gamma
	\end{equation*}
	can be rearranged into the inequality:
	\begin{equation*}
		\frac{t \cdot x_t}{\sum_{i = 1}^k i \cdot x_i} \geq \vg_t(A) + { \frac{\gamma}{\sum_{i = 1}^k i \cdot x_i}.}
	\end{equation*}
	This constrains the type-$t$ affinity score for class $A$ to be larger than its baseline score by at least an additive term {$\gamma /\sum_{i = 1}^k i \cdot x_i$}, which will be positive if and only if $\gamma$ is positive. The second set of constraints encodes similar bounds for the affinity scores of class $B$. A feasible solution with $\gamma = 0$ can always be achieved if the $x_i$ variables represent hyperedge counts for a hypergraph whose affinity scores are equal to the realizable baseline scores. We prove the following results in the  appendix.
	\begin{lemma}
		\label{iff}
		Let $\gamma^*$ be the optimal solution to the linear program in~\eqref{lpmain}. There exists a two-class $k$-uniform hypergraph where both classes exhibit strict majority homophily if and only if $\gamma^* > 0$. 
	\end{lemma}
	Given this result, we can check numerically whether strict majority homophily can hold for two classes, as long as we are given a fixed set of baseline scores and a fixed $k$. However, numerical solutions do not provide a full proof of our result for general baseline scores and arbitrary $k$. { In order to prove our theorem, we consider the dual of the linear program in~\eqref{lpmain}, which is given by
		\begin{equation}
			\label{lpdual}
			\begin{array}{llr}
				\text{min} & \alpha \\
				\text{s.t.}& \sum_{t = r}^k y_{A,t} + y_{B,t} \geq 1 \\
				&-iy_{B,i} +  (k-i)\sum_{j = r}^k y_{A,j}\vg_j(A)+ \;i\sum_{j = r}^k y_{B,j}\vg_j(B)  + \alpha \geq 1 \hspace{.5cm}\text{ for $i \in \{0\} \cup [k]$} \\
				& y_{A,t} \geq 0 \hspace{4.1cm}\text{ for $t \in \{r, \hdots, k\}$} \\
				& y_{B,t} \geq 0 \hspace{4.1cm}\text{ for $t \in \{r, \hdots, k\}$}.
			\end{array}
		\end{equation}
		We prove a key result regarding a set of feasible variables for the dual LP.}
	\begin{lemma}
		\label{dualvariables}
		For an odd integer $k$ and $r = (k+1)/2$, define $\delta = 2k\sum_{t = r}^k \frac{1}{t}$, and consider the following set of dual variables:
		\begin{align*}
			\alpha &= 0 \\
			y_{B,k} &= \frac{2}{\delta} \cdot \frac{\sum_{i = r}^k (\frac{k}{i} - 1)\vg_i(B)} {1 - \sum_{i = r}^k (2 - \frac{k}{i}) \vg_i(B)} \\
			y_{B,t} &= \frac{2}{\delta}\left( \frac{k}{t} - 1\right) + \left(2 - \frac{k}{t}\right) y_{B,k} \,\,\text{ for $t \in \{r, \cdots, k-1\}$} \\
			y_{A,t} &= \frac{2k}{\delta t} - y_{B,t} \,\,\text{ for $t \in \{r, \cdots, k\}$.}
		\end{align*}
		If $Y = \sum_{t = r}^k y_{A,t} + y_{B,t}$, then the set of normalized dual variables defined by $\tilde{y}_{X,t} = y_{X,t}/Y$ for $ X\in \{A,B\}$ and $t \in \{r, \hdots , k\}$ is feasible for the dual LP in~\eqref{lpdual}.
	\end{lemma}
	The full proof of this lemma, provided in the appendix, is quite involved and relies on the fact that the baseline scores are realizable. Once the result is proven, it immediately implies our impossibility result for odd $k$.  By linear programming duality, any feasible solution for the dual LP provides an upper bound on the solution to the primal LP. Since the dual variables we provide in Lemma~\ref{dualvariables} come with an objective score of $\alpha = 0$, we know the optimal solution to the primal LP is also $0$. By Lemma~\ref{iff}, strict majority homophily must be impossible to satisfy for both classes $A$ and $B$ at once if $k$ is odd. For even $k$, the impossibility result in Theorem~\ref{thm:majorityimpossible} can be shown by adding one more constraint to the primal linear program and providing an analytical solution to the new dual linear program, similar to Lemma~\ref{dualvariables}.
	
	While our linear programming framework does not constitute the only way to prove Theorem~\ref{thm:majorityimpossible}, it is a useful approach for capturing extremal limits of higher-order homophily beyond this specific result and its proof. The linear constraints encoding bounds on different affinity scores can be easily altered to quickly check the feasibility of other notions of homophily. For example, one could quickly check whether the top $i$ ratio scores can all be above a certain fixed threshold for two node classes at once, for different values of $i$. This LP framework can also be used more broadly as a proof technique for other theoretical results. Our proof of Proposition~\ref{prop:needall} in the appendix in fact makes use of the LP formulation in~\eqref{lpmain} and Lemma~\ref{iff}. We can also use an alternative linear program and LP duality proof to prove Theorems~\ref{thm:oddkmonotonic} and~\ref{thm:evenkmonotonic}. In this case, and unlike Lemma~\ref{dualvariables}, most of the optimal dual variables for a linear program encoding strict \emph{monotonic} homophily end up being  zero, except for the dual variables associated with the contradictory constraints in inequalities~\eqref{firstA} and~\eqref{secondB}. This simplifies the proof for monotonic homophily, and again indicates that the result for monotonic homophily is simpler to show than the corresponding result for majority homophily.

	\subsection*{Alternative affinity scores and normalizations} 
	{ Several slightly different measures of graph homophily have been considered in previous research~\cite{altenburger2018monophily,currarini2009economic,coleman1958relational,newman2003mixing},
		and in the same way there is more than one way to quantify homophily in the hypergraph setting. }
	We additionally consider the following alternative hypergraph affinity scores:
	\begin{equation}
		\label{eq:edgeaffinities2}
		\tilde{\mathbf{h}}_t(A) = \frac{m_t}{\sum_{i = 1}^k m_i}, \text{ and } \tilde{\mathbf{h}}_t(B) = \frac{m_{k-t}}{\sum_{i = 1}^k m_{k-i}}.
	\end{equation}
	Unlike the affinity scores in~\eqref{eq:edgeaffinities}, which are equivalent to our original definition in~\eqref{affinities}, these scores directly depend on the proportion of different hyperedge types. This is another natural approach for quantifying an entire class's group interaction patterns. In the appendix, we derive matching combinatorial impossibilities for these alternative scores, showing that our main results persist across various notions of group affinities. { We primarily focus on the affinity score in~\eqref{affinities}, defined by ratios of typed degrees, as this directly generalizes an existing notion of a graph homophily index~\cite{altenburger2018monophily}. 
		This focus on node degrees is also shared by other closely related measures of graph homophily~\cite{coleman1958relational,currarini2009economic}, and provides a way to capture the average experience or behavior of an individual in a certain node class.
		
		There is also more than one approach to measuring how much a graph homophily index deviates from a null model. One useful normalization in the graph setting is to consider how much a graph homophily index deviates from its baseline, relative to the \emph{maximum} amount that it could deviate from baseline~\cite{coleman1958relational,currarini2009economic}. We can incorporate this notion into our hypergraph framework by defining the following type-$t$ \emph{normalized bias} 
		score 
		\begin{equation}
			\label{eq:normbias}
			\textbf{f}_t(X) = \begin{cases}
				\frac{\h_t(X) - \bs_t(X)}{1 - \bs_t(X)} & \text{ if $\h_t(X) \geq \bs_t(X)$ } \\ \\
				\frac{\h_t(X) - \bs_t(X)}{\bs_t(X)} & \text{ if $\h_t(X) < \bs_t(X)$. } \\
			\end{cases}
		\end{equation}
		The value $\h_t(X) - \bs_t(X)$ is the \emph{bias} that class $X$ has for type-$t$ interactions, and $\textbf{f}_t(X)$ normalizes this by the maximum possible bias. If this affinity is overexpressed ($\h_t(X) > \bs_t(X)$), then it has a positive bias and the maximum bias is achieved when $\h_t(X) = 1$. When the affinity is underexpressed, the maximum deviation from baseline is when $\h_t(X) = 0$ and we therefore normalize by $\bs_t(X)$. The normalized bias score conveys useful information both in terms of its sign and magnitude. The sign indicates whether an affinity has a positive or negative bias, and the magnitude is always a value between 0 and 1 that indicates how close it comes to its maximum bias. While the magnitude of a ratio score $\h_t(X)/\bs_t(X)$ may depend on the hypergraph size, the fact that $\textbf{f}_t(X)$ is always between $-1$ and $1$ makes it a particularly useful score to use when comparing notions of homophily across different datasets. In our empirical results, we will often consider normalized bias scores in addition to raw affinity scores and ratio scores. Finally, in the appendix we show that Theorems~\ref{thm:oddkmonotonic},~\ref{thm:evenkmonotonic}, and~\ref{thm:majorityimpossible} immediately lead to analogous impossibility results for normalized bias scores as well. We see therefore that our main impossibility results persist across a wide range of different notions of higher-order homophily, and that our framework easily accommodates different approaches for measuring deviation from a null model.}

	\section*{Empirical Results}
	Our theoretical results reveal that seemingly natural notions of group homophily are in fact overly strict and cannot be exhibited, simply because of combinatorial impossibilities. However, this by no means implies that higher-order homophily cannot be meaningfully measured or exhibited in practice. To the contrary, establishing limits on what \emph{is} combinatorially possible in group homophily allows us to better interpret and appreciate relaxed notions of majority and monotonic group homophily that \emph{do} hold in practice despite being very close to the combinatorial limits of group homophily. We specifically apply our framework to study group homophily in legislative bill-cosponsorship, online hotel reviewing, shopping trip data, and group picture data.
	Our framework reveals new insights into the way same-class group mixing patterns are exhibited in these settings, and allows us to uncover structure and patterns that are missed when applying graph homophily measures, which only account only for pairwise interactions.
	\subsection*{Political homophily in legislative bill co-sponsorship}
	\begin{figure}[t]
		\centering
		\includegraphics[width=.95\linewidth]{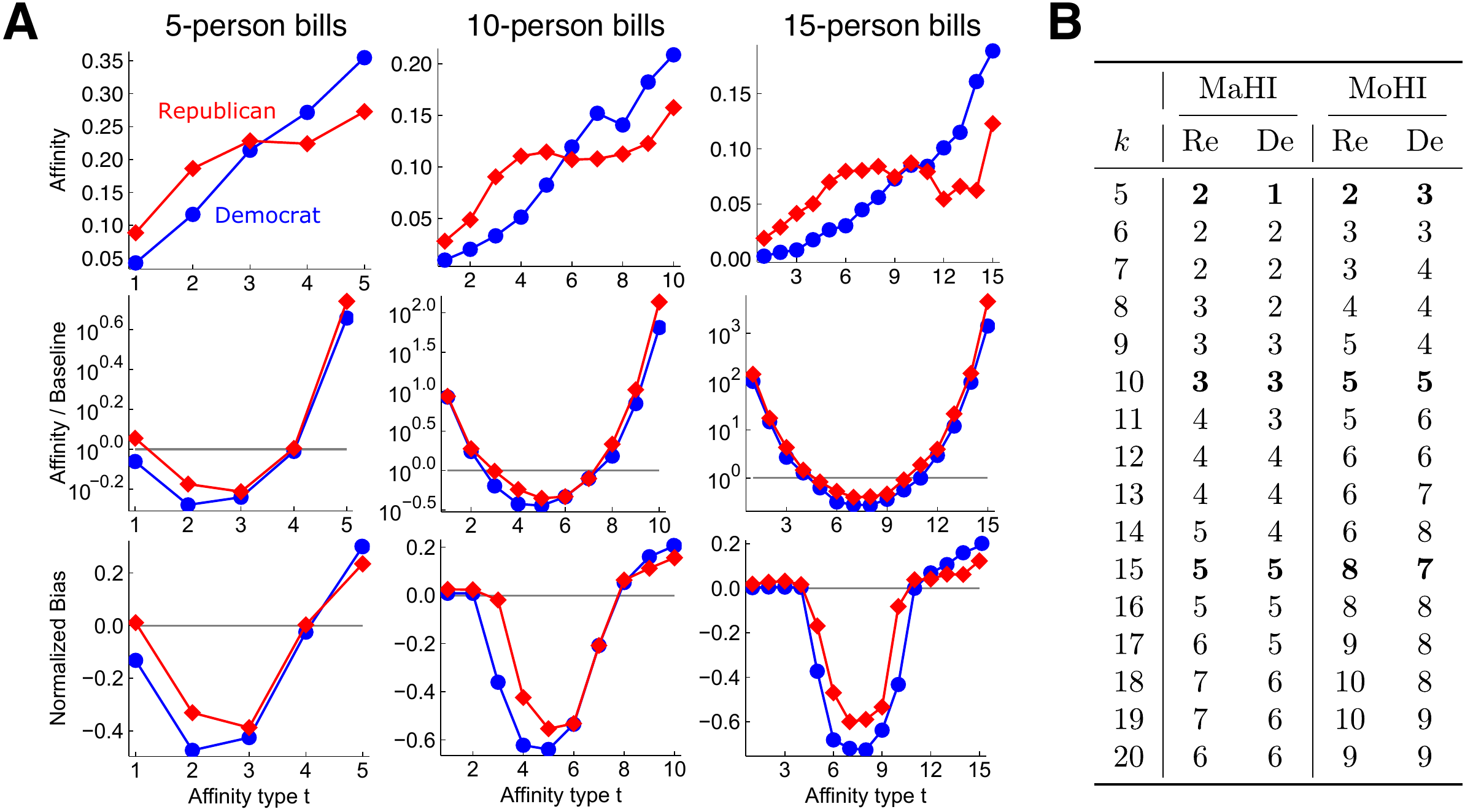} 
		\caption{\small
			\textbf{US Members of congress (nodes) co-sponsoring bills (hyperedges) exhibit certain notions of same-class homophily in terms of political party}. (\textbf{A}) Affinity scores (top row) increase for Democrats, but are relatively flat for Republicans. However, after dividing by baselines (middle row), both classes exhibit bowl-shaped ratio scores that nearly satisfy majority homophily (higher-than-baseline affinities for groups where one's class is the majority) and monotonic homophily (strictly increasing ratio scores for groups where one's class is in the majority) without ever violating our theoretical impossibility results. For example, for bills with 5 co-sponsors, Democrats exhibit monotonic homophily, and Republicans almost do as well, except for a slight decrease in scores from $t = 2$ to $t = 3$ (middle row, leftmost panel). Similar observations hold for bills of other co-sponsorship sizes. 
			(\textbf{B}) Both Republicans (Re) and Democrats (De) almost always exhibit the highest possible monotonic homophily index (MoHI) without violating combinatorial limits (e.g., MoHI of $k/2$ for both parties when $k$ is even). Neither class exhibits majority homophily. However, as group size increases, the majority homophily index (MaHI, the number of top affinity scores above baseline) increases. Bold rows correspond to plots for bills with 5,  10, and 15 co-sponsors. Our results are also robust to perturbations in the data; we obtain nearly identical plots when averaging scores obtained by repeatedly subsampling hyperedges from the dataset (see the appendix).}
		\label{fig:bills} 
	\end{figure}
	
	We quantify political homophily~\cite{loomis1946political,gentzkow2019measuring,adamic2005polblogs,poole1985spatial,berelson1954voting} 
	with respect to political party for US members of congress (nodes in a hypergraph), based on groups of congress members formed by co-sponsorship of legislative bills (hyperedges)~\cite{Fowler-2006-cosponsorship,Benson-2018-simplicial,Fowler-2006-connecting}.
	{	Affinity scores for Democrats strictly increase for most group sizes, though scores are flatter for Republicans (Fig.~\ref{fig:bills}A, top row). However, both classes exhibit similar bowl-shaped ratio curves and normalized bias curves (Fig.~\ref{fig:bills}A, bottom two rows), demonstrating that highly imbalanced groups are highly overexpressed.}

	Our framework reveals very strong notions of group homophily, at the most extreme limits of what is combinatorially possible. Recall that for even-sized groups, strict monotonic homophily is the most extreme example of what is combinatorially possible for two classes simultaneously, while for odd-sized groups it lies just beyond the combinatorially feasible boundary. Ratio scores for the two political parties almost perfectly match these extreme combinatorial limits. Both parties satisfy strict monotonic homophily for most even-sized groups, which is reflected in monotonic homophily indices (MoHI) of $k/2$ when $k$ is even (Fig.~\ref{fig:bills}B).  When group size $k$ is odd, we typically see one party exhibit strict monotonic homophily (MoHI of $(k+1)/2$), while the other party just barely fails to satisfy strict majority homophily (MoHI of $(k-1)/2$). { In other words, for nearly every group size, we observe the maximum level of monotonic homophily that can be exhibited by two classes simultaneously.}
	Across all bill sizes, neither political party exhibits strict majority homophily, but {both} classes exhibit higher-than-random affinity scores when their political party makes up a large enough majority of bill co-sponsors. This is formally captured by majority homophily indices (MaHI) that steadily increase for each political party as the group size $k$ increases.
	
	The bowl-shaped ratio curves { and normalized bias curves} for Democrats and Republicans also illustrate a point that at first seems counterintuitive, but is easily understood in light of our framework. In order for both parties to exhibit overexpressed affinities for groups where they possess a large majority, a substantial number of individuals from each party must be willing to participate in groups where their party is in the minority. Overall, both parties have a higher tendency to participate in co-sponsorship groups where they are significantly outnumbered, compared with participating in groups where they constitute just a slight majority. { The normalized bias scores (third row of Fig.~\ref{fig:bills}A) for these slight-majority groups are in fact close to the minimum score of $-1$, indicating that affinities are almost as far below baseline as they could be.}
	{This at first seems to contradict the notion of homophily in group interactions, but our theoretical results explain why this must hold in order for both parties to satisfy strong notions of group homophily.}
	
	The major difference between affinity scores $\h_t(X)$ and ratio scores $\h_t(X)/\bs_t(X)$ arises because selecting group members at random from two balanced classes would naturally tend to produce class-balanced groups. In other words, the baseline scores $(\bs_1(X), \bs_2(X), \hdots, \bs_k(X))$ will be imbalanced as group size $k$ grows, with $\bs_t(X)$ decreasing as $t$ approaches extreme values.
	Flat or increasing affinity scores for both political parties (Fig.~\ref{fig:bills}A, top row) indicate that the social processes driving bill co-sponsorship have overcome the tendency towards class-balanced group interactions that would be expected at random. This reveals another major difference between measuring group homophily and measuring homophily in graphs. For class-balanced graphs, roughly equal affinity scores $\h_1(X) \approx \h_2(X) \approx 0.5$ indicate that there is no clear tendency towards in-class or out-class links, i.e, no clear tendency towards homophily.
	
	\subsection*{Location homophily in trip reviews}
	\begin{figure}[t]
		\centering
		\includegraphics[width=.95\linewidth]{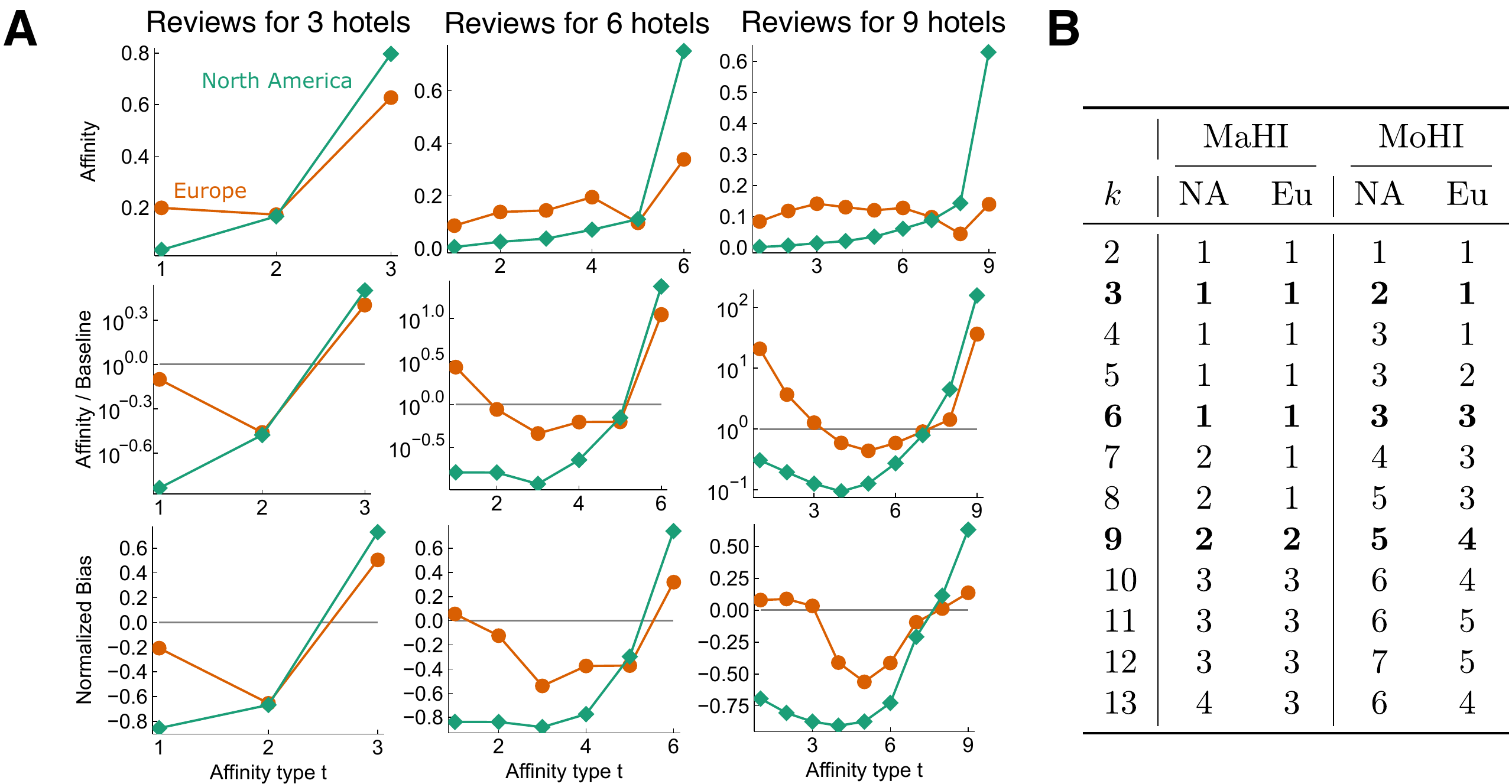} 
		\caption{\small \textbf{Measures of homophily with respect to location for groups of co-reviewed vacation rentals.} (\textbf{A}) Affinity, ratio, and normalized bias scores on a hypergraph where nodes are hotels, separated into to location classes (North America and Europe), and hyperedges are sets of hotels reviewed by the same user account on Trip Advisor~\cite{wang2011latent}. Results are comparable to our findings on political homophily in legislative bill co-sponsorship. In each case, both classes have monotonic or \emph{nearly} monotonic ratio scores (second row of plots) for groups where their class is in the majority. (\textbf{B}) {Monotonic homophily indices (MoHI) for both North America (NA) and Europe (Eu) are usually at the most extreme limits of what is combinatorially possible. The increasing majority homophily index (MaHI) shows that classes do tend to exhibit high affinities for groups where their class has a substantial majority, but this is not the case for groups where their class has only a slight majority. }
		}
		\label{fig:ta} 
	\end{figure}
	Our framework also applies to hypergraphs that do not encode social interactions. We compute affinity scores for a hypergraph where nodes are hotels on \verb+tripadvisor.com+ from two location classes, North America and Europe, and each hyperedge indicates a group of hotels reviewed by the same user account (Fig.~\ref{fig:ta}A). 
	Affinity scores demonstrate intuitive location homophily in review information: users tend to review hotels that are in the same location. {Monotonic homophily indices for both location classes are at the extreme combinatorial limits of what is possible for two classes simultaneously, across different numbers of reviews.
		Furthermore, affinities are higher than random for reviewing sets of hotels as long as a substantial majority of the hotels are from the same location (Fig.~\ref{fig:ta}B).}
	
	\subsection*{Product homophily in shopping baskets}
	\begin{figure}[t]
		\centering
		\includegraphics[width=.95\linewidth]{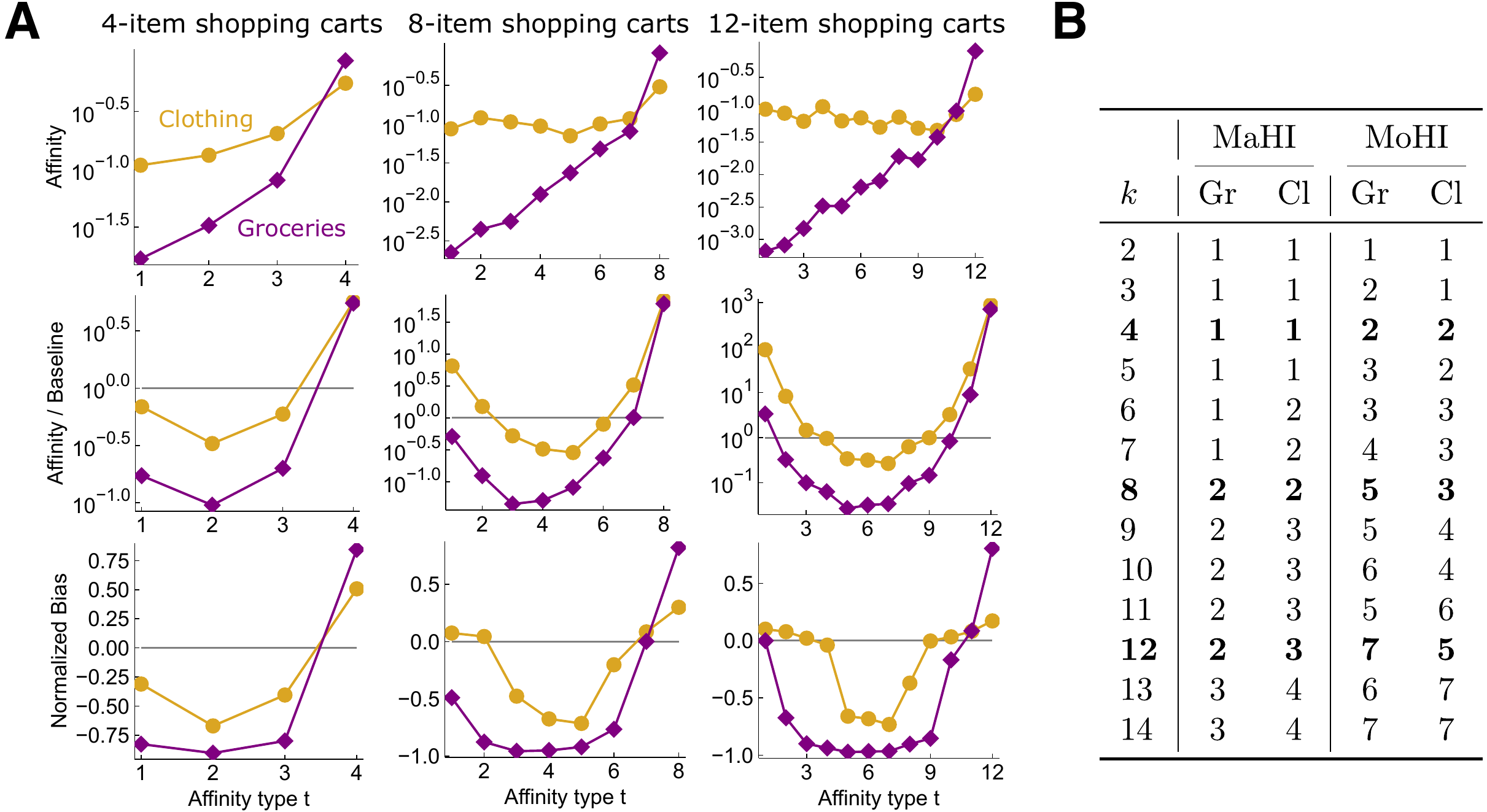} 
		\caption{\small
			\textbf{ Measures of homophily with respect to product type for groups of co-purchased retail products.}
			(\textbf{A}) Affinity, ratio, and normalized bias scores for a hypergraph where nodes indicate clothes (Cl) and grocery (Gr) products at Walmart, and hyperedges indicate sets of co-purchased items during a shopping trip. The bowl-shaped ratio { and normalized bias} curves for both products indicate that it is typical for shopping trips to primarily focus on one type of product or the other. The fact that affinity scores for grocery products are mostly increasing, while affinity scores for clothing are mostly flat, also matches basic intuition about shopping trips. For example, the relatively small gap between type-$k$ and type-$(k-1)$ affinities for clothing is indicative of the fact that when going clothes shopping, it is not uncommon to pick up a needed grocery item while at the store. In contrast, when grocery shopping, it is much less common to additionally pick up a small number of clothing items at the store. 
			This is reflected in the larger gaps between affinity scores and normalized bias scores for groceries. 
			(\textbf{B}) {The majority homophily indices grow as shopping basket size increases. Monotonic homophily indices are also high for both classes.}
		}
		\label{fig:walmart}
	\end{figure} 
	We also compute affinity scores for a hypergraph derived from a dataset on Walmart shopping trips~\cite{Amburg-2020-categorical}. Each node is a product, and hyperedges represent sets of co-purchased products (i.e., shopping baskets). Each product comes with a store department label. The labels ``Food, Household \& Pets'' and ``Clothing, Shoes \& Accessories'' make up over half of the original dataset, indicating that a large proportion of shopping purchases can broadly be categorized as clothes purchases or grocery purchases. We consider the two-class hypergraph obtained by restricting to products with these two labels, and compute affinity, ratio, { and normalized bias scores}. Similar to our empirical results on congress bills and hotel reviews, we observe that flat or slightly increasing affinity scores translate to bowl-shaped ratio scores and normalized bias scores (Fig.~\ref{fig:walmart}A). In other words, shopping trips where a significant majority of purchases are from one product category are much more common than expected by chance. This matches the intuition that many shopping trips can be categorized as grocery runs or clothes shopping trips. Our results also highlight an intuitive difference between these types of trips: it is more common to pick up a small number of grocery items while on a clothes shopping trip, than to purchase a small number of clothes items while on a grocery run. { This is reflected in the larger gaps between affinity scores for groceries, and can also been seen in the normalized bias scores. In particular, the type-$k$ normalized bias score $\textbf{f}_k$ for groceries is very close to 1 for all values of $k$, indicating that simple homophily for groceries is almost as high as it possibly could be. Meanwhile, normalized biased scores $\textbf{f}_t$ are very close to $-1$ if $t$ is just a few values less than $k$ (e.g., when $t =  9$ for the $k = 12$ plot), indicating that grocery trips that include a few clothing items are almost as far below baseline scores as they could be. }
	
	\subsection*{Gender homophily in pictures}
	\begin{figure}
		\centering
		\includegraphics[width=.95\linewidth]{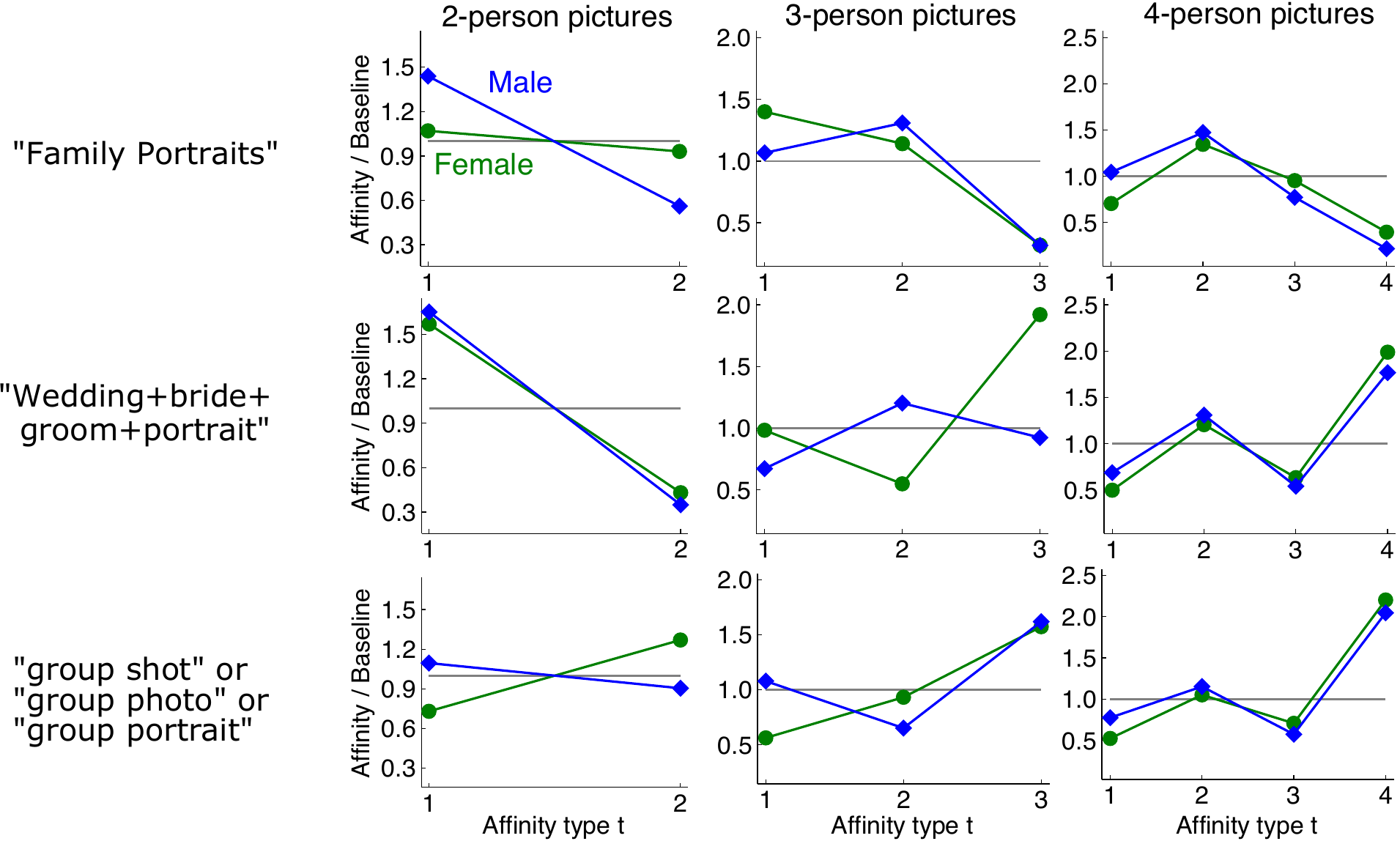} 
		\caption{\small \textbf{Ratio scores with respect to gender for three collections of group pictures, obtained via three image search queries on Flickr}. { Normalized bias scores capture the same trends}. Hypergraph measures provide richer information than graph homophily indices obtained by collapsing pictures into pairwise relationships based on co-appearance. 
			For family pictures (top row), ratio scores capture the intuition that all-male and all-female family pictures are statistically uncommon, as shown by low ratio scores for 3- and 4-person pictures of all men or all women. Meanwhile, the graph homophily indices for men and women when collapsing all family pictures into pairwise relationships are 0.43 and 0.41 respectively, just below the baseline of 0.5 for balanced classes.
			Ratio scores for 4-person wedding pictures (middle row) or general group pictures (bottom row) indicate a high frequency of social gatherings of all men or all women. The slightly higher-than-random affinities for gatherings with two men and two women are possibly due to a high number of pictures of two opposite-gender couples. Two-person wedding photos are likely to be of a bride and groom, which is reflected in low type-2 ratio scores (first column, middle row). However, pictures with three or four people are often gender homogeneous. This information is lost when collapsing all pictures into pairwise co-appearances, in which case the resulting graph homophily indices of 0.57 and 0.55 are both slightly above baseline (0.5). Results for the dataset are robust to perturbations; we see similar patterns from average affinity scores obtained from different subsamples of each dataset (see the appendix). }
		\label{fig:photos}
	\end{figure} 
	Our framework can also be used to study homophily in groups even when group members are not uniquely associated with nodes in a hypergraph. We apply our framework to analyze gender homophily in group pictures, comprised of three subsets of pictures capturing family portraits, wedding portraits, and general group pictures~\cite{gallagher2009}. In order to compute affinity scores, it suffices to know the size and gender composition of each group in a picture, even without unique identifiers for each individual. Hypergraph affinity scores reveal that the gender distribution in group pictures depends on group size. This information is completely lost if we reduce group pictures to pairwise co-appearances and compute a graph homophily index. Our framework also reveals several salient differences between gender distributions across different picture types (Fig.~\ref{fig:photos}A). 
	
	In wedding pictures, affinity scores for same-gender two-person pictures are far smaller than expected at random, reflecting the fact that many of these photos are of a bride and groom. However, there is more gender homophily in three- and four-person pictures at weddings, as shown by higher-than-random affinity scores for gender homogeneous groups (i.e., simple homophily). Reducing all wedding pictures to pairwise relationships suppresses these subtle differences, and just produces a graph where both genders have slightly higher-than-random graph homophily indices. Wedding pictures and general group pictures with exactly four people show similar patterns. Pictures with two men and two women are slightly more common than expected by chance, and pictures of all men or all women are much more common than expected by chance. The former may be due to a high volume of pictures of two heterosexual couples, while the latter indicates an overall tendency for friends to gather in groups that are completely homogeneous with respect to gender. 
	
	Our framework reveals several interesting differences in the context of family pictures. In family pictures with three or four people, pictures with only men or only women are far less common than expected by chance. This matches the intuition that statistically, family photos are less likely to be of all men or all women. In contrast, graph homophily indices for men and women on a reduced graph (defined by co-appearances) are only slightly lower than expected by chance. Another interesting observation is that for 4-person pictures, type-2 ratio scores are the highest for both men and women. This can be explained by the fact that it is statistically very common for four-person families to consist of a mother and father with two children. { In this case, assuming that the children have an equally likely chance of being male or female, there is a 50\% chance that there will be one boy and one girl, a 25\% chance that both children will be boys, and a 25\% chance that both children will be girls. This explains why type-2 ratio scores are much higher than type-3 ratio scores. Again, this type of nuanced information is lost when reducing group pictures to co-appearances and using graph measures of homophily. }
	
	\section*{Discussion}
	Understanding group formation and interactions has long been a goal of homophily research, but previous methods have focused almost exclusively on pairwise approaches. Our framework for hypergraph homophily quantifies the tendency of individuals to participate in multiway interactions that differ in size and class balance. Our results show that group interactions among different classes of individuals must obey certain combinatorial constraints, which render seemingly intuitive notions of group homophily impossible. At the same time, these combinatorial impossibilities do not imply that group interactions happen indiscriminately of class labels, and in practice we do see many examples of class homogeneity in group interactions. We find that in many group settings, homophily can be characterized by bowl-shaped ratio score curves. These scores indicate that different classes of individuals exhibit increasing and higher-than-random affinities for group interactions when a large enough majority of group members are from the same class. 
	
	{ Our empirical results illustrate the utility of defining and computing a different affinity score for each group size and group type separately. This is most clearly illustrated in our results on group pictures, where homophily patterns are significantly affected both by picture context (e.g., wedding vs.\ family picture) and group size. At the same time, it can also be useful to capture aggregate information about group homophily that persists across group types and sizes. Our measures of MoHI and MaHI provide one simplified aggregate score; determining other aggregate scores that summarize the tendency towards homophily across multiple group types and sizes at once is an interesting direction for future research. Another direction to consider is how our framework can be applied and generalized to hypergraphs with multiple class labels. Our definition of affinity and baseline scores can already be applied to an arbitrary number of class labels, but an interesting open question to explore is how combinatorial limits change in this setting. 
		Finally, it would be worthwhile to consider how our hypergraph framework can be used to explore higher-order generalizations of other mixing patterns, such as \emph{monophily}~\cite{altenburger2018monophily}, that may be present in group interactions even when homophily is not.}

	\section*{Materials and Methods}
	The appendix provides details for the original datasets and the construction of each hypergraph from Figures~\ref{fig:dblp},~\ref{fig:bills},~\ref{fig:ta},~\ref{fig:walmart}, and~\ref{fig:photos}. 
	
	\paragraph{Asymptotic baseline scores.} For numerical experiments on all hypergraphs except for the small hypergraph used in Figure~\ref{fig:walmart}, we used asymptotic variant of the standard baseline scores in~\eqref{base1}. To compute asymptotic baselines, we consider a two-class hypergraph where the class proportions are given by $\alpha = \frac{|A|}{n}$ and $(1-\alpha) = \frac{|B|}{n}$.
	Treating $\alpha$ as a constant and letting $n \rightarrow \infty$, the standard baseline scores converge to the following asymptotic baselines:
	\begin{align}
		\bs_t(A) &= \alpha^{t-1} (1-\alpha)^{k-t} { k-1 \choose t-1} \\\bs_t(B) &= (1-\alpha)^{t-1} \alpha^{k-t} { k-1 \choose t-1},
	\end{align}
	for $t \in [k]$.
	These scores correspond to probability mass functions for binomial random variables $\text{Bin}(\alpha, k)$ and $\text{Bin}(1-\alpha, k)$ respectively. The number of nodes $n$ is sufficiently large for our datasets that these are virtually indistinguishable from standard baseline scores, and are also more convenient to compute and use in practice. For standard baselines scores, computing binomial coefficients for large $n$ and $k$ can lead to overflow issues in practice; asymptotic baselines provide one way to avoid this issue. This also mirrors the standard practice in the graph case, since typically the graph homophily index for a class $X$ is compared against the asymptotic baseline score $|X|/n$ rather than $(|X|-1)/(n-1)$.
	
	\paragraph{Graph homophily index}
	In some cases, we compare against the graph homophily index obtained by reducing hyperedges in the hypergraph $H = (V,E)$ to pairwise relationships. Formally, we define a graph $G$ where nodes $u,v \in V$ share an edge if $\{u,v\}\subseteq e$ for some $e \in E$. The graph $G$ can be described as a two-uniform hypergraph, and the graph homophily index for a class $X \subseteq V$ is exactly the type-$2$ affinity $\h_2(X)$, computed using~\eqref{affinities}.
	
	
		\section*{Acknowledgments}
	We thank Johan Ugander and Kristen Altenburger for helpful conversations. 
	A.R.B.\ was supported by NSF (DMS-1830274), ARO (W911NF19-1-0057), ARO MURI, and JPMorgan Chase \& Co.
	J.K.\ was supported by a Simons Investigator Award, a Vannevar Bush Faculty Fellowship, ARO MURI, and AFOSR.
	Author contributions: N.V.\ contributed to the conceptualization, methodology, software, formal analysis, data curation, visualization, and wrote the original draft; A.R.B.\ contributed to the conceptualization, methodology, data curation, visualization, writing - review \& editing, supervision, and funding acquisition; and J.K.\ contributed to the conceptualization, methodology, supervision, visualization, writing - review \& editing, and funding acquisition.
	Authors declare no competing interests. { All software and data for reproducing the experimental results are available publicly on Zenodo (\url{https://doi.org/10.5281/zenodo.7086798})}. Additional details are provided in the appendix.

	\bibliographystyle{unsrt}
	\bibliography{vbk-arxiv}

	\appendix

\section{Main Theoretical Results for Hypergraph Affinities}
In this section we provide full details for our theoretical results regarding hypergraph affinity scores. We begin by reviewing and covering additional necessary terminology and notation in Section~\ref{sec:notation}. We then show how to interpret affinity scores as maximum likelihood estimates for a certain affinity parameter of a binomial model for degree data (Section~\ref{sec:mle}), cover additional background on baseline scores (Section~\ref{sec:baseline}), and then prove our main theoretical impossibility results for hypergraph homophily (Section~\ref{sec:impossibility}). In Section~\ref{sec:alternative}, we show how these results can be adapted to an alternative notion of affinity scores.

\subsection{Notation and Terminology}
\label{sec:notation}
Consider a hypergraph $H = (V,E)$ where $V$ is a set of $n = |V|$ nodes and $E$ is a set of $m = |E|$ hyperedges. We assume throughout that $H$ is $k$-uniform (where $k$ is constant) and non-degenerate, meaning that all hyperedges are of a fixed size $k$, and a node can appear at most once in a hyperedge. We also assume that nodes are organized into one of two classes $A \subseteq V$ and $B \subseteq V$ where $A \cup B = V$ and $A \cap B = \emptyset$. 

For class $X \in \{A,B\}$ and integer $t \in [k] = \{1, 2, \hdots k\}$, we say that a hyperedge $e \in E$ is of type-$(X,t)$ if exactly $t$ nodes in $e$ are from class $X$. Let $m_t(X)$ denote the number of type-$(X,t)$ hyperedges in $E$. Since there are exactly two node classes, $m_t(A) = m_{k-t}(B)$ and $m_t(B) = m_{k-t}(A)$. It is often convenient to refer to hyperedge types in an absolute sense, without specifying class. We say that a hyperedge is of absolute type-$t$ if exactly $t$ of its nodes are from class $A$, and denote the number of such edges by
\begin{equation}
	m_t = m_{t}(A) = \text{number of hyperedges of absolute type-$t$ in $E$.}
\end{equation}

The degree of a node $v \in V$, denoted $d(v)$, is the number of hyperedges it participates in. For an integer $t \in [k]$, let $d_t(v)$ denote the number of hyperedges $v$ participates in where exactly $t$ nodes are from $v$'s class, including $v$ itself. We refer to this as the type-$t$ degree of $v$. Summing typed-degrees produces the degree of a node:
\begin{equation*}
	d(v) = \sum_{t = 1}^k d_t(v).
\end{equation*}
The type-$t$ affinity score for a class $X \in \{A,B\}$ measures the propensity for nodes in this class to participate in type-$(X,t)$ hyperedges. This score can be expressed in terms of sums of node degrees:
\begin{equation}
	\label{eq:degaffinity}
	h_t(X) = \frac{\sum_{v \in X} d_t(v)}{\sum_{v \in X} d(v)}.
\end{equation}
This directly generalizes the homophily index of a graph~\cite{altenburger2018monophily}, which is defined similarly in terms of typed degrees, and corresponds to the case where $k = t = 2$. 

Affinity scores can also be expressed in terms of hyperedge counts. The sum of type-$t$ degrees for a class $X$ satisfies
\begin{equation*}
	\sum_{v \in X} d_t(v) = t m_{t}(X)
\end{equation*}
The value $m_t(X)$ is scaled by a factor $t$ to account for the fact that each type-$(X,t)$ hyperedge affects the degree of $t$ different nodes from class $X$. We can express type-$t$ affinity scores for both classes in terms of absolute hyperedge types as follows:
\begin{align}
	\label{eq:edgeaffinities-app}
	\h_t(A) &= \frac{tm_t}{\sum_{i = 1}^k i m_i}\\
	\h_t(B) &= \frac{tm_{k-t}}{\sum_{i = 1}^k i m_{k-i}}.
\end{align}

\subsection{Cardinality-Based Hypergraph Stochastic Block Model}
\label{sec:cardhsbm}
In order to provide a statistical interpretation for hypergraph affinity scores, we define a simple new generative model for hypergraphs. For this model, consider a set of nodes $V$ separated into two classes $A$ and $B$. We say a tuple of $k$ distinct nodes in $V$ is a type-$t$ $k$-tuple if exactly $t$ of nodes in the tuple are from class $A$. For each $t \in \{0,1, \hdots , k\}$, define a probability $p_t \in [0,1]$. We emphasize the fact that these probabilities are defined with respect to a fixed class $A$; since it is possible to have hyperedges where all nodes are from class $B$, this also includes a probability $p_0$. We define the \emph{cardinality-based hypergraph stochastic block model} (cardinality-based HSBM) as follows: for each type-$t$ $k$-tuple of nodes $T = (v_1, v_2, \hdots , v_k)$, we generate a hyperedge on $T$ with probability $p_t$. We denote the distribution of cardinality-based hypergraph stochastic block models with size $k$ hyperedges by $\mathcal{H}(n,k,A,B,\textbf{p})$ where $\textbf{p} = [p_0, p_1, \hdots , p_k]$ is a vector of hyperedge probabilities. This a special case of the general $k$-uniform hypergraph stochastic block model~\cite{ghoshdastidar2014consistency}, which may involve more than two ground truth clusters or classes.

\subsection{Affinity Scores as Maximum Likelihood Estimates}
\label{sec:mle}
We now show how affinity scores for class $A$ can be derived as maximum likelihood estimates for an affinity parameter of a certain binomial distribution. The same approach could also be applied to class $B$. 

\subsubsection{Type-degree Random Variables}
Assume we are given a $k$-uniform hypergraph from the cardinality-based HSBM $\mathcal{H}(n,k,A,B,\textbf{p})$, where the probability vector $\textbf{p} = [p_0, p_1, \hdots, p_k]$ is given up front and fixed. For a node $a \in A$, let $T_j$ be the total number of type-$j$ $k$-tuples of nodes that $a$ belongs to:
\begin{equation}
	\label{tat}
	T_j = {|A| -1 \choose j-1}{|B| \choose k-j}.
\end{equation}
This value is the same regardless of which $a \in A$ we consider, and represents the maximum  number of type-$j$ hyperedges that $a$ could belong to in an $n$-node hypergraph with node classes $A$ and $B$. The type-$t$ degree of each node $a \in A$ conditioned on probability $p_t$ will be a binomial random variable
\begin{equation}
	\label{Dt}
	D_t(a) \sim \text{Binom}\left(T_t, p_{t}\right).
\end{equation}
We also define a random variable $D(a)$ for the total degree of $a \in A$ by
\begin{equation}
	\label{D}
	D(a)  = \sum_{j = 1}^k D_j(a)\,.
\end{equation}
Finally, define another random variable for measuring the contribution to $D(a)$ made by  hyperedges that are \emph{not} of type-$t$:
\begin{equation}
	\label{Dnt}
	\tilde{D}_t(a)  = \sum_{j = 1, j \neq t}^k D_j(a)\,.
\end{equation}
For any fixed $t$, $D_t(a) + \tilde{D}_t(a) = D(a)$. The degree random variables $D(a)$ will not be independent for different $a \in V$, since the degrees must define a valid degree sequence for a $k$-uniform hypergraph. However, we can prove that they will be approximately independent by adapting existing techniques for graphs~\cite{van2016random}.
\begin{silemma}
	\label{lem:deg}
	Let $H \sim \mathcal{H}(n,k,A,B,\textbf{p})$ be a cardinality-based HSBM with hyperedge parameters satisfying $\hat{p} = \max_{i} p_i = O\left(\frac{1}{n^{k-1}}\right)$. If $\ell > k$ is a fixed constant, the degree random variables for any set of $\ell$ nodes, $D(1), D(2), \hdots , D(\ell)$, are asymptotically independent.
\end{silemma}
\begin{proof}
	Let $\mathcal{K}$ denote the set of $k$-tuples of nodes in $H$, and let $L$ denote the set of $\ell$ nodes we are considering, which we denote by $\{1, 2, \hdots, \ell\}$ without loss of generality. Each $e \in \mathcal{K}$ is associated with a Bernoulli random variable $X_e$ 
	such that $X_e = 1$ if an edge is placed at $k$-tuple $e$. Each random variable $D(i)$ for $i \in [\ell]$ is a sum of Bernoulli random variables:
	\begin{equation*}
		D(i) = \sum_{e \colon i \in e} X_e.
	\end{equation*}
	For $i, j \in [\ell]$, $D(i)$ and $D(j)$ are not independent, as $X_e$ appears in both of their sums whenever $i \in e$ and $j \in e$. If $e$ is a $k$-tuple of nodes from $L$, then $X_e$ contributes to the sum of all of the random variables $(D_i)_{i \in L}$. In general for an arbitrary $e \in \mathcal{K}$, the variable $X_e$ shows up in $|e \cap L|$ of these degree variables. Note that for each $s \in \{ 2, 3, \hdots k\}$, there are ${\ell \choose s} {n - \ell \choose k - s}$ distinct $k$-tuples involving $s$ nodes from $L$ and $k-s$ nodes from $V-L$.
	
	
	In order to prove that random variables $(D_i)_{i \in L}$ are approximately independent, for each $i \in L$ we construct a new random variable $\hat{D}(i)$ in such a way that $D(i)$ and $\hat{D}(i)$ have the same distribution, and so that the $\hat{D}(i)$ variables are mutually independent. In order to establish asymptotic independence of the original $D(i)$ variables, we then prove that
	\begin{equation*}
		\pr\left[ (D(i))_{i \in L} \neq (\hat{D}(i))_{i \in L} \right] = o(1).
	\end{equation*}
	In order to accomplish this, for each $X_e$ that shows up in more than one variable from $(D_i)_{i \in L}$, we construct $|e \cap L|-1$ other \emph{independent} copies of this random variable $X_e$, denoted by $X^{(2)}_e, X^{(3)}_e, \hdots, X^{(s)}_e$ where $s = |e \cap L|$. Define $X^{(1)}_e = X_e$ for notational convenience. We then define a new variable $\hat{D}(i)$ for each $i \in L$, which is the same as $D(i)$, except we carefully replace the $X_e$ variables with the independent copies of $X_e$, in order to ensure the $\hat{D}(i)$ variables are independent. We begin by defining
	\begin{equation*}
		\hat{D}(1) = \sum_{e \colon 1 \in e}  X_e^{(1)} = {D}(1).
	\end{equation*}
	Then, to define $\hat{D}(j)$ for $j > 1$, we start with the same sum of random variables that defines $D(j)$, but then we replace each $X_e$ in this sum by the copy $X_{e}^{(r+1)}$, where $r$ is the number of times that $X_e$ appeared in a sum $D(i)$ for $i < j$. Thus, if $X_e$ shows up $s$ times in the summations defining $(D_i)_{i \in L}$, we have constructed $s$ independent copies of $X_e$ and used these in defining $(\hat{D}_i)_{i \in L}$. As a result, $\hat{D}(i)$ and $D(i)$ have the same distribution, but the variables $(\hat{D}_i)_{i \in L}$ are independent. 
	
	What remains is to prove that the probability that $(\hat{D}_i)_{i \in L}$ and $({D}_i)_{i \in L}$ are not equal goes to zero. These random variables will be the same if for every $k$-tuple $e$ with $s = |e \cap L | \geq 2$, the variables $X_e^{(1)}, \hdots , X_e^{(s)}$ are all the same. Each of these is a Bernoulli random variable with some probability
	$p_t \leq \hat{p} = O\left(  \frac{1}{n^{k-1}}\right)$. The probability that $X_e^{(1)}, \hdots , X_e^{(s)}$ all coincide is the probability that they all equal one or all equal zero, so:
	\begin{equation}
		\label{coincide}
		\pr(X_e^{(1)}, \hdots , X_e^{(s)} \text{ do not coincide}) = 1 - p_t^s - (1-p_t)^s \leq 1 - (1-\hat{p})^s.
	\end{equation}
	For $s \in \{2, 3, \hdots, k\}$, let $\mathcal{K}_s$ denote the set of $k$-tuples in with exactly $s$ nodes from $L$. We can use Boole's inequality to bound the probability that $(\hat{D}_i)_{i \in L}$ and $({D}_i)_{i \in L}$ are not equal:
	\begin{align*}
		\pr\left[ (\hat{D}_i)_{i \in L}\neq ({D}_i)_{i \in L} \right] &\leq \sum_{s = 2}^k \sum_{e \in \mathcal{K}_s}   \pr(X_e^{(1)}, \hdots , X_e^{(s)} \text{ do not coincide}) \\
		&\leq \sum_{s = 2}^k \sum_{e \in \mathcal{K}_s}  1 - (1-\hat{p})^s \\
		&\leq \sum_{s = 2}^k {\ell \choose s} {n - \ell \choose k - s} \left( 1 - (1-\hat{p})^s\right).
	\end{align*}
	
	Finally, if $\hat{p} = O(n^{-\alpha})$ for some integer $\alpha$, we know that $\hat{p} \leq \frac{C}{n^{\alpha}}$ for some constant $C$, so we have the following asymptotic result:
	\begin{align*}
		{n - \ell \choose k - s }(1 - (1- \hat{p})^s) &\leq {n - \ell \choose k - s }\left(1 - \left(1- \frac{C}{n^{\alpha}}\right)^s\right) \\
		&= {n - \ell \choose k - s } \left( \frac{n^{\alpha s} - (n^{\alpha} - C)^s}{n^{\alpha s}}\right) \\
		&= O \left( \frac{n^{k-s}\cdot n^{\alpha(s-1)}}{n^{\alpha s}} \right) \\
		&= O \left( \frac{n^{k-s}}{n^{\alpha}} \right) .
	\end{align*} 
	where we have used the fact that $(n^{\alpha} - C)^s = n^{\alpha s} - C n^{\alpha(s-1)} + f(n)$ where $f(n)$ represents lower order terms in $n$. Thus, as long as $\alpha \geq k-1 > k - s$, we see that this entire expression goes to zero for every $s \in \{2,3,\hdots , k\}$, and so we have our desired asymptotic result:
	\begin{equation*}
		\lim_{n \rightarrow \infty} \pr\left[ (\hat{D}_i)_{i \in L}\neq ({D}_i)_{i \in L} \right] \leq \lim_{n \rightarrow \infty} \sum_{s = 2}^k   {\ell \choose s} {n - \ell \choose k - s} \left( 1 - (1-\hat{p})^s\right) = 0.
	\end{equation*}
	

\end{proof}

\subsubsection{Conditional Distribution of Type-$t$ Degrees}
Assume we observe a two-class $k$-uniform hypergraph for which the degree of node $a \in A$ is given by $d(a)$. Given our fixed hyperedge-probabilities $p_{j}$ for each edge type, we can prove that for a fixed $t$, the conditional random variable $(D_t(a) \;\lvert\; D_a = d(a))$ is asymptotically binomially distributed 
\begin{equation}
	\label{datbin}
	D_t(a) \;\lvert\;d(a) \sim \text{Binom}\left(d(a), f_t\right),
\end{equation} 
where $f_t$ is the affinity parameter:
\begin{equation}
	f_t = \frac{ p_t \cdot  T_t}{ \displaystyle{\sum_{j =1}^k} p_j \cdot T_j}\,.
\end{equation}

This holds specifically for the parameter regime where $T_j$ is large but $p_j \cdot T_j$ is constant for all $j \in \{1, 1, \hdots k\}$, and each $d(a)$ is a constant. Recall that $D_j(a)$ is binomially distributed for all $j \in \{1, \hdots, k\}$. Under the assumed conditions on $p_j$ and $T_j$, these binomial distributions are asymptotically equivalent to Poisson distributions:
\begin{equation}
	\label{approxpoisson}
	\lim_{n\rightarrow \infty} \text{Pr}( D_j(a) = c) =\frac{(p_j \cdot T_j)^c}{c!} e^{-p_j \cdot T_j}.
\end{equation}
A sum of binomials is also asymptotically equivalent to a Poisson distribution with the same mean~\cite{serfling1978some}, so we also have that
\begin{equation}
	\label{approxpoisson2}
	\lim_{n\rightarrow \infty} \text{Pr} \left(\sum_{j = 1}^k D_j(a) = c\right) = \frac{(\sum_{j = 1}^k p_j \cdot T_j)^c}{c!} e^{-(\sum_{j = 1}^k p_j \cdot T_j)}.
\end{equation}
Similarly,
\begin{equation}
	\label{approxpoisson3}
	\lim_{n\rightarrow \infty} \text{Pr} \left(\sum_{j = 1}^k \tilde{D}_j(a) =d(a)- c\right) = \frac{(\sum_{j\neq t} p_j \cdot T_j)^{d(a)-c}}{(d(a)-c)!} e^{-\sum_{j\neq t} p_j \cdot T_j} 
\end{equation}
By our assumption that $p_jT_j = O(1)$, the right hand sides of~\eqref{approxpoisson},~\eqref{approxpoisson2}, and~\eqref{approxpoisson3} are all constants. Therefore, asymptotically the distribution of $(D_t(a) \lvert D_a = d(a))$ is a binomial with parameter $f_t$, since
\begin{align*}
	\pr(D_t(a) = c\lvert D_a = d(a)) &= \pr\left(D_t(a) = c \;\Bigg\lvert\; \sum_{j = 1}^k D_j(a) = d(a)\right)\\
	&= \frac{\pr(D_t(a) = c, \tilde{D}_t(a) = d(a)- c)}{ \pr(\sum_{j = 1}^k D_j(a) = d(a))} \\
	&= \frac{\pr(D_t(a) = c) \cdot \pr(\tilde{D}_t(a) = d(a)- c)}{ \pr(\sum_{j = 1}^k D_j(a) = d(a))} 
\end{align*}
and therefore
\begin{align*}
	\lim_{n\rightarrow \infty} \pr(D_t(a) = c\;\lvert\; D_a = d(a)) &=  \frac {  \left[\frac{(p_t \cdot T_t)^c}{c!} e^{-p_t \cdot T_t} \right]  
		\left[\frac{(\sum_{j\neq t} p_j \cdot T_j)^{d(a)-c}}{(d(a)-c)!} e^{-\sum_{j\neq t} p_j \cdot T_j} \right]   } 
	{ \left[\frac{(\sum_{j} p_j \cdot T_j)^{d(a)}}{(d(a))!} e^{-\sum_{j} p_j \cdot T_j} \right]  }\\
	&= {d(a) \choose c} \left( \frac{p_t \cdot T_t}{\sum_{j} p_j \cdot T_j} \right)^c \left(1 - \frac{p_t \cdot T_t}{\sum_{j} p_j \cdot T_j} \right)^{d(a)-c}. \\
\end{align*}

\subsubsection{Affinity Index as Maximum Likelihood Estimate}
Given observed degree data $(d(a), d_t(a) |  a \in A)$ for a two-class, $k$-uniform hypergraph, we can model the type-$t$ degree for an arbitrary node in $A$ using the binomial distribution given in~\eqref{datbin}. For this data, the type-$t$ affinity score will equal the maximum likelihood estimate for the affinity parameter $f_t$. To show this result, define $\tilde{d}_t(a) = d(a) - d_t(a)$, i.e., the part of the degree that does not come from type-$t$ hyperedges. For simplicity and without loss of generality, denote the nodes in $A$ by the indices $1, 2, \hdots, |A|$. Treating degrees as independent, the likelihood function for observing the given degrees for a given parameter $f_t$ is
\begin{align*}
	L(f_t) &= \pr(D_t(1)= d_t(1), D_t(2)= d_t(2), \hdots , D_t(|A|)= d_t(|A|)) \\
	&= \prod_{a \in A} {d(a) \choose d_t(a)} \cdot {f_t}^{d_t(a)} \cdot (1-f_t)^{\tilde{d}_t(a)}.
\end{align*}
The log-likelihood function is
\begin{align*}
	\ell(f_t) = \log L(f_t) &= \sum_{a \in A} \log {d(a) \choose d_t(a)} + \log ( {f_t}^{d_t(a)}) + \log ((1-f_t)^{\tilde{d}_t(a)}) \\
	&\propto \sum_{a \in A}  d_t(a) \log (f_t) + \tilde{d}_t(a) \log (1-f_t) .
\end{align*}
Taking the derivative with respect to the parameter $f_t$ and setting it equal to zero, we find that the log-likelihood is maximized then $f_t$ is exactly equal to the type-$t$ affinity index.
\begin{align*}
	\frac{\partial \ell}{\partial f_t} = 0 \Rightarrow \frac{\sum_{a \in A} d_t(a)}{f_t} = \frac{\sum_{a \in A} \tilde{d}_t(a)}{1-f_t} \Rightarrow f_t = \frac{\sum_{a \in A} d_t(a)}{\sum_{a\in A} d(a)}.
\end{align*}

\subsection{Baseline Scores and Proofs of Propositions~\ref{prop:base1} and~\ref{prop:base2}}
\label{sec:baseline}
In order to determine the meaningfulness of a type-$t$ affinity score, we compare it against a baseline score representing a null probability for participation in type-$t$ hyperedges. For class $X$, the standard type-$t$ baseline score $\bs_t(X)$ measures the probability that a node $v \in X$ forms a type-$(X,t)$ hyperedge (i.e., there are exactly $t$ nodes from $v$'s class in the hyperedge) if it selects $k - 1$ other nodes from $V$ uniformly at random. Formally, 
\begin{align}
	\label{base1-app}
	\bs_t(X) = \frac{{|X|-1 \choose t-1} {n-|X| \choose k-t}}{{n-1\choose k-1}}.
\end{align}
Comparing the type-t affinity $\h_t(X)$ against $\bs_t(X)$ generalizes a standard approach for checking for homophily in graphs. When $k = t = 2$, the hypergraph affinity score $\h_t(X)$ equals the homophily index of a graph, and the type-2 baseline score is
\begin{equation}
	\bs_{2}(X) = \frac{|X| - 1}{n - 1}.
\end{equation}
As $n \rightarrow \infty$, this converges to the class proportion $|X|/n$, which is typically used as the standard baseline for the homophily index of a graph. In addition to generalizing the baseline for a graph homophily index, our hypergraph baseline scores satisfy the following intuitive interpretation, as given in the main text.
\begin{reproposition}
	\label{lem:base-app}
	Let $H_{k,n}^* = (V,E)$ be the complete $k$-uniform hypergraph on $n$ nodes.
	The type-$t$ affinity score for class $X  \subseteq V$ equal the type-$t$ baseline score in~\eqref{base1-app}.
\end{reproposition}
\begin{proof}
	The number of type-$(X,t)$ hyperedges in $H$ is $m_t(X) = {|X| \choose t} {n-|X| \choose k-t}$, the total number of ways to choose $t$ nodes from $X$ and $k-t$ nodes that are not in $X$. The type-$t$ affinity score for class $X$ is therefore
	\begin{align*}
		\label{2equiv1}
		\h_t(X) &= \frac{t m_t(X)}{\sum_{i = 1}^k i m_i(X)} = \frac{t {|X| \choose t} {n-|X|\choose k-t}}{ \sum_{i = 1}^k i {|X| \choose i} {n-|X| \choose k-i} } \\ \\
		&= \frac{|X| {|X|-1 \choose t-1} {n-|X| \choose k-t}}{|X| \sum_{i = 1}^k  {|X|-1 \choose i-1} {n-|X|\choose k-i} } = \frac{{|X|-1 \choose t-1} {n-|X| \choose k-t}}{ \sum_{i = 1}^k   {|X|-1 \choose i-1} {n-|X| \choose k-i} },
	\end{align*}
	where we have used the fact that ${|X| \choose i} = \frac{|X|}{i} {|X| - 1 \choose i-1}$. This is the same as $\bs_t(X)$ --- the numerator is identical, and the denominator is an alternative way to list all the possible ways to select $k-1$ nodes from a set of $n-1$ nodes, by separately counting how many of each type of hyperedge could be formed.
\end{proof}
To provide further intuition for the baseline scores, we observe that the type-$t$ baseline score is also related the type-$t$ hypergraph affinity for a hypergraph obtained by generating hyperedges at random without regard to node class. Consider a cardinality-based HSBM where for some $p \in (0,1)$, $p_i = p$ for all $i \in \{0,1, 2, \hdots , k\}$. If $M_t$ is a random variable representing the number of type-$t$ hyperedges, then 
\begin{equation*}
	\mathbb{E}[M_t] = p \cdot {|A| \choose t} {n-|A| \choose k-t}
\end{equation*}
If we replace $m_t$ with this expected value $\mathbb{E}[M_t]$ in the definition for $\h_t(A)$ given in equation~\eqref{eq:edgeaffinities-app}, then we exactly recover the type-$t$ baseline score for class $A$:
\begin{align}
	\label{eq:emt}
	\frac{t\mathbb{E}[M_t] }{\sum_{i = 1}^k i \cdot \mathbb{E}[M_i]} &= \frac{p \cdot t \cdot {|A| \choose t} {n-|A| \choose k-t} }{p \cdot \sum_{i = 1}^k i \cdot {|A| \choose i} {n-|A| \choose k-i}} = \frac{|A| {|A|-1 \choose t-1} {n-|A| \choose k-t}}{|A| \sum_{i = 1}^k  {|A|-1 \choose i-1} {n-|A|\choose k-i} } = \bs_t(A).
\end{align}
An analogous result also holds for baselines scores of class $B$. 

Proposition~\ref{prop:base2} from the main text is an even stronger result, showing that when hyperedges are generated at random without regard for node labels, the ratio scores of the resulting hypergraph converge to one.
\begin{reproposition}
	Fix any $p \in (0,1)$ and a positive integer $k$, and let $H = (V,E)$ be a random hypergraph on $n$ nodes that is formed by turning each $k$-tuple of nodes in $V$ into a hyperedge with probability $p$. As $n \rightarrow \infty$, the ratio scores for a class $X \subseteq V$ with $|X| = \Theta(n)$ converge in probability to 1.
\end{reproposition}
\begin{proof}
	Our goal is to show that for every $\varepsilon > 0$ and  $\delta > 0$, there exists some $n_0 \in \mathbb{N}$ such that for all $n > n_0$ and for all $t \in [k]$, with probability at least $1-\delta$, 
	\begin{equation*}
		\left \lvert  \frac{\h_t(X)}{\bs_t(X)} - 1 \right \rvert < \varepsilon.
	\end{equation*}

	For $i \in [k]$, let $N_i$ be the number of $k$-tuples with exactly $i$ nodes from class $X$, and $M_i$ be the expected number of hyperedges of type-$(X,i)$ (exactly $i$ nodes from class $X$) in $H$. The random variable $M_i$ is binomially distributed, $M_i \sim \text{Bin}(p, N_i)$, and has the following expected value and variance:
	\begin{align*}
		m_i &= \mathbb{E}[M_i] = p N_i \\
		\sigma_i^2 &= N_i p(1-p).
	\end{align*}
	As long as $M_i > 0$ for some $i \in [k]$, the type-$t$ affinity score for $H$ is well defined and equals
	\begin{equation}
		\label{eq:htx}
		\h_t(X) = \frac{t   M_t }{\sum_{i = 1}^k i   M_i}.
	\end{equation}
	From~\eqref{eq:emt} we know that replacing $M_i$ with $m_i$ in~\eqref{eq:htx} yields the type-$t$ baseline score for $H$:
	\begin{equation*}
		\bs_t(X) = \frac{t   m_t }{\sum_{i = 1}^k i   m_i}.
	\end{equation*}
	To prove that $\h_t(X)/\bs_t(X)$ converges to one, we first prove several facts about the limiting behavior of $M_i/m_i$.
	By Chebyshev's inequality, we know that for any $i \in [k]$ and any $\varepsilon > 0$,
	\begin{equation}
		\label{cheby}
		\Pr \left( \left \lvert  \frac{M_i}{m_i} - 1 \right \rvert \geq \varepsilon\right ) \leq \frac{1}{N_i p(1-p) \varepsilon^2}.
	\end{equation}
	Since $N_i = O(n^k)$ for every $i \in [k]$, this establishes that as long as $n$ is large enough, $M_i/m_i$ can be made arbitrarily close to one with high probability. With this in mind, fix $\varepsilon > 0$ and $\delta > 0$, and choose $n_0$ so that for all $n > n_0$, with probability at least $1-\delta$
	\begin{align}
		\label{e1}
		\left \lvert \frac{M_i}{m_i} - 1\right \vert < \hat{\varepsilon},
	\end{align}
	for all $i \in [k]$,  where $\hat{\varepsilon} < \min \left\{\frac{1}{2}, \frac{\varepsilon}{4}\right \}$. This implies that $M_i > 0$ and $(1-\hat{\varepsilon}) < \frac{M_i}{m_i} < (1+\hat{\varepsilon})$, and therefore we also know that
	\begin{equation}
		\label{otheratio}
		\left \lvert \frac{m_i}{M_i} - 1 \right \rvert = \frac{m_i}{M_i} \left \lvert  \frac{M_i}{m_i} - 1 \right \rvert < \frac{\hat{\varepsilon}}{1 - \hat{\varepsilon}}.
	\end{equation}
	For our final step of the proof we will make use of the following useful inequality, that holds for two sets of positive numbers $\{a_1, a_2, \hdots , a_\ell\}$ and $\{b_1, b_2, \hdots , b_\ell\}$ and an arbitrary integer $\ell$:
	\begin{equation}
		\label{useful}
		\frac{ \sum_{i=1}^\ell a_i}{ \sum_{i=1}^\ell b_i} \leq \max_{i \in [\ell]} \frac{a_i}{b_i}.
	\end{equation}
	Using~\eqref{e1}, ~\eqref{otheratio}, and~\eqref{useful}, we know that with probability at least $1-\delta$,
	\begin{align*}
		\left \lvert \frac{\h_t(A)}{\bs_t(A)} - 1 \right \rvert &= \left \lvert \frac{t   M_t }{\sum_{i = 1}^k i   M_i} \frac{\sum_{i = 1}^k i   m_i}{t  m_t } - 1\right \rvert
		=\left \lvert \frac{(M_t \sum_{i = 1}^k i  m_i)  - (m_t \sum_{i = 1}^k i   M_i )}{m_t \sum_{i = 1}^k i  M_i}  \right \rvert\\
		&=\left \lvert \frac{(M_t \sum_{i = 1}^k i m_i) - (m_t \sum_{i = 1}^k i m_i)  + (m_t\sum_{i = 1}^k i m_i) - (m_t \sum_{i = 1}^k i   M_i) }{m_t \sum_{i = 1}^k i  M_i}  \right \rvert\\
		& \leq \left \lvert \frac{(M_t \sum_{i = 1}^k i m_i) - (m_t \sum_{i = 1}^k i m_i)}{m_t \sum_{i = 1}^k i  M_i}  \right \rvert +\left \lvert \frac{(m_t\sum_{i = 1}^k i m_i) - (m_t \sum_{i = 1}^k i   M_i) }{m_t \sum_{i = 1}^k i  M_i}  \right \rvert\\ 
		& \leq \frac{ |M_t - m_t|}{m_t} \frac{\sum_{i = 1}^k i m_i}{\sum_{i = 1}^k i M_i} + \frac{ \sum_{i = 1}^k i | m_i - M_i|}{\sum_{i = 1}^k iM_i }\\
		& \leq \left \lvert \frac{M_t}{m_t} - 1 \right \rvert \max_i \frac{m_i}{M_i} + \max_i \frac{|m_i - M_i|}{M_i}  \hspace{1cm} (\text{applying~\eqref{useful}}) \\
		&= \left \lvert \frac{M_t}{m_t} - 1 \right \rvert \frac{m_j}{M_j} + \left \lvert \frac{m_\ell}{M_\ell} -1 \right \rvert \hspace{1cm} (\text{for some integers $j$ and $\ell$}) \\
		& < \frac{\hat{\varepsilon}}{1-\hat{\varepsilon}} + \frac{\hat{\varepsilon}}{1-\hat{\varepsilon}} < 4 \hat{\varepsilon} < \varepsilon.
	\end{align*}
	
\end{proof}

\section{Hypergraph Homophily Impossibility Results}
\label{sec:impossibility}
We restate our main results for hypergraph homophily, exactly as given in the main text.
\begin{retheorem}
	Let $H$ be a two-class, $k$-uniform hypergraph and $\{\vg_i(X)\colon i \in [k], X \in \{A,B\} \}$ be realizable baseline scores. For odd $k$, it is impossible for both classes to exhibit monotonic homophily.
\end{retheorem}
\begin{retheorem}
	Let $H$ be a two-class, $k$-uniform hypergraph and $\{\vg_i(X)\colon i \in [k], X \in \{A,B\} \}$ be realizable baseline scores. If $k$ is even, then it is imposssible for both classes satisfy monotonic homophily if additionally $\frac{\h_\ell(X)}{\vg_\ell(X)} > \frac{\h_{\ell-1}(X)}{\vg_{\ell-1}(X)} $ for one class $X \in \{A,B\}$, where $\ell = k/2$.
\end{retheorem}
\begin{retheorem}
	Let $H = (V,E)$ be a two-class $k$-uniform hypergraph and $\{\bs_i(X)\colon  i\in [k], X \in \{A,B\} \}$ be realizable baseline scores.
	\begin{itemize}[itemsep = 0pt]
		\item If $k$ is odd, it is impossible for both classes to simultaneously exhibit majority homophily
		\item If $k$ is even, it is impossible for both classes to exhibit majority homophily if additionally $\h_{k/2}(X) > \bs_{k/2}(X)$ for one of the classes $X \in \{A,B\}$.
	\end{itemize} 
\end{retheorem}

In the main text we include a full proof of Theorem~\ref{thm:oddkmonotonic}. Theorem~\ref{thm:evenkmonotonic} is nearly identical and relies on simply adding one extra constraint and repeating the same basic steps. Here we provide full details for our majority homophily result for odd $k$, and show how they can be altered to yield the result for even $k$. For clarity and ease of presentation, we include many of the same steps as given in the main text, in some cases with expanded explanations. 

Throughout the section, $H = (V,E)$ denotes a hypergraph with two node classes $\{A,B\}$ and hyperedges of a fixed size $k$. For $t \in \{0,1, 2, \hdots, k\}$, $m_t$ represents the number of hyperedges in $H$ where exactly $t$ out of the $k$ nodes in the hyperedge come from class $A$. As before, $\h_t(A)$ and $\h_t(B)$ denote the type-$t$ affinity scores for classes $A$ and $B$ respectively, and can both be expressed in terms of absolute hyperedge counts:
\begin{align}
	\label{eq:eaffA}
	\h_t(A) &= \frac{tm_t}{\sum_{i = 1}^k i \cdot m_i}\\
	\label{eq:eaffB}
	\h_t(B) &= \frac{tm_{k-t}}{\sum_{i = 1}^k i \cdot m_{k-i}}.
\end{align}
Our results apply to a generalized notion of baseline scores.
\begin{definition}
	We will refer to the set of baseline scores $\{\vg_t(X) \colon t \in [k], X \in \{A,B\}\}$ as \emph{realizable} or \emph{generalized} baseline scores if they satisfy the following two assumptions:
	\begin{itemize}
		\item For $t \in \{1, 2, \hdots ,k\}$, $\vg_t(A) > 0 $ and $\vg_t(B) > 0$.
		\item There exists some two-class, $k$-uniform hypergraph $G$ such that for each $t \in \{1, 2, \hdots ,k\}$, $\vg_t(A)$ and $\vg_t(B)$ are the type-$t$ affinity scores for classes $A$ and $B$ in $G$.
	\end{itemize}
\end{definition} 
As long as $\min \{|A|,|B|\} \geq k$, the standard baseline scores satisfy the above definition, by Proposition~\ref{prop:base1}.
We now recall the definition of majority homophily presented in the main text. 
\begin{definition}
	Class $X \in \{A,B\}$ exhibits majority homophily if for all $t > k-t$, $\h_t(X) > \vg_t(X)$.
\end{definition}

\subsection{Impossibility Result for Odd $k$}
\label{sec:majorityodd}
Recall from the main text that when $k$ is odd, requiring both classes to exhibit majority homophily induces a constraint on each hyperedge count $m_t$ for $t \in \{0,1,2, \hdots k\}$. This is due to the fact that type-$(A,t)$ hyperedges ($t$ nodes from class $A$) are also type-$(B,k-t)$ hyperedges ($k-t$ nodes from class $B$). Applying a few steps of algebra, we can show that $\h_t(A) > \vg_t(A)$ for $t > k-t$ implies that
\begin{align}
	\label{eq:mt-app}
	m_t > \frac{\vg_t(A)}{t \cdot (1-\vg_t(A))} \sum_{i \neq t} i m_i \hspace{1cm} \text{ for $t = k, k-1, k - 2, \hdots, (k+1)/2$} 
\end{align}
while $\h_t(B) > \vg_t(B)$  for $t > k-t$ implies that 
\begin{align}
	m_{k-t} &> \frac{\vg_t(B)}{t \cdot (1-\vg_t(B))} \sum_{i \neq t} i m_{k-i} \hspace{1cm}\text{ for $t = k, k-1, k - 2, \hdots, (k+1)/2$} \\
	\label{eq:ms-app}
	\implies m_s &> \frac{\vg_{k-s}(B)}{(k-s) \cdot (1-\vg_{k-s}(B))} \sum_{i \neq k-s} i m_{k-i} \hspace{1cm}\text{ for $s = 0, 1, 2, \hdots, (k-1)/2$}.
\end{align}
If there existed a type $j$ such that $m_j$ were not bounded below, we could set $m_j = 0$ and instead make all other hyperedge counts higher than would be expected at random, and in doing so make most hyperedge types overexpressed relative to the baseline. We will show that this is not possible when $m_t$ is lower bounded for all $t \in \{0,1,2, \hdots, k\}$ as shown above. However, it is not immediately clear why lower bounds for each hyperedge type cannot all be satisfied simultaneously, especially given that half of the constraints depend on baseline scores for class $A$, while the other bounds depend on baseline scores for class $B$, which might be very different. 

Proposition~\ref{prop:needall} in the main text highlights one key challenge in proving that majority homophily cannot be satisfied by two classes at once. 
\begin{reproposition}
	\label{prop:needallsi}
	If any inequality from~\eqref{eq:mt-app} and~\eqref{eq:ms-app} is discarded, it is possible to construct a two-class $k$-uniform hypergraph satisfying the remaining inequalities.
\end{reproposition}
To prove this result, it will be convenient to consider the linear program (LP) that we will use to prove our impossibility results for majority homophily. A proof of Proposition~\ref{prop:needallsi} will follow by considering what happens in the case of the LP obtained by removing one constraint.

\paragraph{Linear program for measuring homophily}
The following linear program encodes the maximum amount of homophily that can be satisfied by classes $A$ and $B$ simultaneously.
\begin{equation}
	\label{lpmain-app}
	\begin{array}{lll}
		\text{maximize} & \gamma \\
		\text{subject to}&t\cdot x_t - \vg_t(A) \cdot \sum_{i = 1}^k i\cdot x_i \geq \gamma &\text{  for $t \in [k], t > k-t$} \\
		&t\cdot  x_{k-t} - \vg_t(B) \cdot \sum_{i = 1}^k i\cdot x_{k-i} \geq \gamma  &\text{  for $t \in [k], t > k-t$} \\
		& \sum_{i = 0}^k x_i = 1 &\\
		& x_i \geq 0 &\text{ for all $i \in \{0\} \cup [k]$}
	\end{array}
\end{equation}
Recall from the main text that there is a variable $x_t \geq 0$ for each type of hyperedge in some $k$-uniform hypergraph. More specifically, the constraint $\sum_{i = 1}^k x_i = 1$ encodes the fact that $x_t$ represents the proportion of hyperedges that are of type-$t$ (i.e., $t$ out of $k$ nodes are from class $A$). The constraint
\begin{equation*}
	t\cdot x_t - \vg_t(A) \cdot \sum_{i = 1}^k i \cdot x_i \geq \gamma
\end{equation*}
can be rearranged into the inequality:
\begin{equation*}
	\frac{t \cdot x_t}{\sum_{i = 1}^k i \cdot x_i} \geq \vg_t(A) + \frac{\gamma} {\sum_{i = 1}^k i \cdot x_i}.
\end{equation*}
This constrains the hypergraph type-$t$ affinity score for class $A$ to be larger than its baseline score by at least an additive term $\gamma/ \sum_{i = 1}^k i \cdot x_i$, which will be positive if and only if $\gamma$ is positive. The second set of constraints encodes similar bounds for the affinity scores of class $B$. A feasible solution with $\gamma = 0$ can always be achieved if the $x_i$ variables represent hyperedge counts for a hypergraph whose affinity scores are equal to the generalized baseline scores. The following lemma shows that if the constraints are satisfied for some $\gamma > 0$, this means there exists a two-class $k$-uniform hypergraph where both classes exhibit majority homophily.
\begin{relemma}
	\label{iff-app}
	Let $\gamma^*$ be the optimal solution to the linear program in~\eqref{lpmain-app}. There exists a two-class $k$-uniform hypergraph $\mathcal{H}$ with both classes exhibiting majority homophily if and only if $\gamma^* > 0$. 
\end{relemma}
\begin{proof}
	Given a hypergraph where both classes satisfy majority homophily, let $x_i = m_i/M$, where $M$ is the total number of hyperedges and $m_i$ is the number of type-$i$ hyperedges. The type-$t$ affinity for class $A$ is given by
	\begin{equation}
		\frac{t \cdot m_t}{\sum_{i = 1}^k i \cdot m_i}=	\frac{t \cdot x_t}{\sum_{i = 1}^k i \cdot x_i}.
	\end{equation}
	A similar expression in terms of $x_i$ variables can be shown for affinity scores for class $B$.
	Choose the maximum value of $\gamma$ so that all constraints are still satisfied. Since both classes are assumed to satisfy majority homophily, this $\gamma$ will be strictly greater than zero. 
	
	If on the other hand we assume that $\gamma^* > 0$, the variable $x_i$ will represent the proportion of type-$i$ hyperedges in some two-class hypergraph where both classes exhibit majority homophily. All coefficients in the LP are rational, so its solution will be rational as well. We can therefore scale the $x_i$ variables by a common denominator $C$ so that the value $M_i = Cx_i$ is an integer. To construct the appropriate hypergraph, generate $M_i$ hyperedges of type-$i$ for each $i \in \{0,1,2,\hdots , k\}$. This can always be done by generating hyperedges that are completely disjoint. If one desires a specific balance between the number of nodes in classes $A$ and $B$, isolated nodes from either class can be added. The resulting hypergraph provides the desired example. 
\end{proof}

\paragraph{Proof of Proposition~\ref{prop:needallsi}}
Observe that Proposition~\ref{prop:needallsi} is equivalent to stating that if we alter the above LP by removing any one of the constraints of the form
\begin{align*}
	tx_t - \vg_t(A)  \sum_{i = 1}^k i x_i \geq \gamma, 
\end{align*}
or any constraint of the form
\begin{align*}
	t x_{k-t} - \vg_t(B) \sum_{i = 1}^k i x_{k-i} \geq \gamma ,
\end{align*}
then the optimal solution to the resulting LP will be strictly greater than zero. We now show why this is true.


Let $r = (k+1)/2$. If $\gamma = 0$, these constraints are all satisfied tightly by setting $\tilde{x}_i = m_i/M$, where $m_i$ is the number of type-$i$ hyperedges and $M$ is the total number of hyperedges, in the hypergraph whose affinity scores equal the given generalized baseline scores. Define 
\begin{align*}
	D_A &= \sum_{i = 1}^k i \ncdot \tilde{x}_i\\
	D_B &= \sum_{i = 1}^k i \tilde{x}_{k-i} = \sum_{i = 0}^{k-1} (k-i) \ncdot \tilde{x}_i. 
\end{align*}
We then have
\begin{align}
	\label{gA}
	t\tilde{x}_t &= \vg_t(A) \ncdot D_A \;\;\;\;\text{ for $t \in \{k, k-1, \hdots , r\}$} \\
	\label{gB}
	(k-t)\tilde{x}_t &= \vg_{k-t}(B) \ncdot D_B \;\;\text{ for $t \in \{r-1, r-2, \hdots , 2, 1\}$}.
\end{align} 
Satisfying the LP constraints for some $\gamma > 0$ is equivalent to satisfying the following set of strict inequalities, one for each of the $x_i$ variables:
\begin{align*}
	k \ncdot x_k &> \vg_k(A) \ncdot \sum_{i = 1}^k i\ncdot x_i \\
	(k-1) \ncdot x_{k-1}  &> \vg_{k-1}(A) \ncdot \sum_{i = 1}^k i\ncdot x_i \\
	& \vdots\\
	r \ncdot x_r &>\vg_{r}(A) \ncdot \sum_{i = 1}^k i\ncdot x_i \\
	r \ncdot x_{r-1} &> \vg_{r}(B) \ncdot \sum_{i = 0}^{k-1} (k-i) \ncdot x_i \\
	\vdots & \\
	(k-1)\ncdot  x_1 &> \vg_{k-1}(B) \ncdot \sum_{i = 0}^{k-1} (k-i) \ncdot x_i \\
	k\ncdot x_0 &> \vg_{k}(B) \ncdot \sum_{i = 0}^{k-1} (k-i) \ncdot x_i .
\end{align*}
If we remove the constraint associated with variable $x_t$ for some $t \in \{1, 2, \hdots, k-1\}$, we can satisfy the remaining inequalities strictly by setting $x_t = 0$ and keeping all other variables the same: $x_i = \tilde{x}_i = m_i/M$ for $i \neq t$. In this case, the right hand side of the equalities in~\eqref{gA} and~\eqref{gB} will strictly decrease, but the left hand side will not change. This set of variables must afterwards be normalized to sum to one, to ensure feasibility for the LP, but this does not change the fact that all inequalities (except the one we discarded) are satisfied strictly.

If we remove the constraint associated with $x_0$ or $x_k$, the proof is similar, but is slightly more involved since the first $r$ inequalities do not involve $x_0$ and the second set of $r$ inequalities do not involve $x_k$. We will prove the result when discarding the inequality that lower bounds $x_0$. By symmetry, the result holds in the same way if we removed the inequality for $x_k$. 

After removing the inequality for $x_0$, consider $\varepsilon > 0$ and the following set of new variables:
\begin{align*}
	x_t = 
	\begin{cases}
		\tilde{x}_t + \frac{\varepsilon}{t}& \text{ if $t \in \{ k, k-1, \hdots , r\}$} \\
		\tilde{x}_t - \frac{\varepsilon}{t} &\text{ if $t \in \{r-1, r-2, \hdots, 2, 1\}$} \\
		0 &\text{if $t = 0$.}
	\end{cases}
\end{align*}
For this new set of variables, we have
\begin{align*}
	\sum_{i = 1}^k i \ncdot x_i = \sum_{i = 1}^k i \ncdot \tilde{x}_i + \sum_{i = r}^k i \ncdot \frac{\varepsilon}{i} - \sum_{i = 1}^{r-1} i \ncdot \frac{\varepsilon}{i} = D_A + \varepsilon.
\end{align*}
Also,
\begin{align*}
	\sum_{i = 0}^{k-1} (k-i) \ncdot x_i = \sum_{i = 1}^{k-1} (k-i) \ncdot \tilde{x}_i + \sum_{i = r}^{k-1} (k-i)\ncdot \frac{\varepsilon}{i} - \sum_{i = 1}^{r-1} (k-i) \ncdot \frac{\varepsilon}{i} - k\ncdot \tilde{x}_0 = D_B - c \varepsilon - k\ncdot \tilde{x}_0,
\end{align*}
where $c =  \sum_{i = 1}^{r-1} \frac{k-i}{i} -\sum_{i = r}^{k-1} \frac{k-i}{i} $ is a positive constant. Observe that the first set of $r$ constraints, which are associated with class $A$ and variables $x_k$ to $x_r$, are satisfied strictly with the new set of variables. For  $t \in \{k, k-1, \hdots , r\}$, we have:
\begin{align*}
	t\ncdot x_t = t \ncdot \tilde{x}_t + \varepsilon > t \ncdot \tilde{x}_t + \vg_t(A) \ncdot \varepsilon = \vg_t(A)\ncdot  D_A + \vg_t(A) \ncdot \varepsilon = \vg_t(A) \ncdot \sum_{i = 1}^k i \ncdot x_i,
\end{align*}
where we have used the fact that $\vg_t(A) < 1$.
Finally, we must simply choose any $\varepsilon > 0$ small enough that the set of inequalities for the variables $x_{r-1}, x_{r-2}, \hdots, x_2, x_1$ are also satisfied strictly, which is still possible since we set $x_0 = 0$. In particular, for $t \in \{r-1, r-2, \hdots , 2, 1\}$, we must satisfy
\begin{align*}
	(k-t) \ncdot x_t > \vg_{k-t}(B)\sum_{i = 0}^{k-1} (k-i) \ncdot x_i.
\end{align*}
Substituting in the definition of $\tilde{x}_i$, this is equivalent to
\begin{align*}
	(k-t) \ncdot \tilde{x}_t - \frac{k-t}{t} \varepsilon > \vg_{k-t}(B) (D_B - c \varepsilon - k\ncdot \tilde{x}_0).
\end{align*}
Using a few steps of algebra, we can see that this is true as long as 
\begin{align*}
	k \ncdot \tilde{x}_0 \ncdot \vg_{k-t}(B) > \varepsilon \left( \frac{k-t}{t} - \vg_{k-t}(B)c \right).
\end{align*}
The left side of the above expression is always positive. If the right side is negative for the given set of generalized baseline scores, the inequality is trivial. If the right side is positive, there exists some $\varepsilon > 0$ that is small enough to ensure the inequality holds strictly. Therefore, if we remove the LP constraint that lower bounds $x_0$, there exists a set of variables whose objective score is strictly positive. This in turn implies that we can remove one inequality from~\eqref{eq:eaffA} and~\eqref{eq:eaffB} and find a hypergraph that satisfies all of the others. This therefore proves Proposition~\ref{prop:needallsi}.

\paragraph{Primal-Dual LP Formulation.}
We now turn our attention back to the original linear program in~\eqref{lpmain-app} that includes all constraints. In order to prove results about optimal solutions this LP, we first re-write it in a general form using matrix notation and compute its dual. Let $\vx = \begin{bmatrix} x_0 & x_1 & \cdots & x_k \end{bmatrix}$ and $\ve$ be the all ones vector, so the constraint $\sum_{i = 0}^k x_i = 1$ is encoded by
\begin{align*}
	\begin{bmatrix}
		\textbf{e}^T & 0
	\end{bmatrix}
	\begin{bmatrix}
		\textbf{x} \\ \gamma
	\end{bmatrix}
	= 	1.
\end{align*}
For odd $k$, let $r = (k+1)/2$. Later we will show how to make adjustments to the LP when proving impossibility results for even $k$. We construct a matrix $\mB$ so that the set of constraints
\begin{align}
	\label{constraintsb}
	t\cdot  x_{k-t} - \vg_t(B) \cdot \sum_{i = 1}^k i\cdot x_{k-i} &\geq \gamma  \text{   for $t = k, k-1, k-2, \hdots, r$} \\
	\label{constraintsa}
	t\cdot x_t - \vg_t(A) \cdot \sum_{i = 1}^k i\cdot x_i &\geq \gamma \text{  for $t = r, r+1, \hdots, k$} 
\end{align}
is encoded by
\begin{align*}
	\begin{bmatrix}
		-\textbf{B} &\textbf{e} \\
	\end{bmatrix}
	\begin{bmatrix}
		\textbf{x} \\ \gamma
	\end{bmatrix}
	\leq 0.
\end{align*}
In order to write the constraints in this way, we carefully order the constraints in~\eqref{constraintsb} and~\eqref{constraintsa} based on our ordering of $x_i$ variables in $\vx$. The first $r$ rows of $\mB$ correspond to constraints in~\eqref{constraintsb}, starting with $t = k$ and decreasing $t$ until $t = r$. The second set of $r$ rows in $\mB$ corresponds to constraints~\eqref{constraintsa}, starting with $t = r$ and increasing until $t = k$. 

Applying standard techniques for computing the dual of a linear program, the LP from~\eqref{lpmain-app} and its dual linear program are then given by

\begin{equation*}
	\begin{minipage}{.4\textwidth}
		\centering
		\textbf{Primal Linear Program}
		\begin{equation}
			\label{primal}
			\begin{array}{ll}
				\max \,\, & \gamma \\
				\text{s.t. }  	&\begin{bmatrix}
					-\textbf{B} &\textbf{e} \\
				\end{bmatrix}
				\begin{bmatrix}
					\textbf{x} \\ \gamma
				\end{bmatrix}
				\leq 0\\
				&\begin{bmatrix}
					\textbf{e}^T & 0
				\end{bmatrix}
				\begin{bmatrix}
					\textbf{x} \\ \gamma
				\end{bmatrix}
				= 	1 \\
				&\textbf{x} \geq 0, \gamma \geq 0
			\end{array}
		\end{equation}
	\end{minipage}
	\begin{minipage}{.6\textwidth}
		\centering
		\textbf{Dual Linear Program}
		\begin{equation}
			\label{dual}
			\begin{array}{ll}
				\min \,\, &\alpha \\
				\text{s.t. }  	&\begin{bmatrix}
					-\textbf{B}^T &\textbf{e} \\
					\textbf{e} & 0
				\end{bmatrix}
				\begin{bmatrix}
					\textbf{y} \\ \alpha
				\end{bmatrix}
				\geq 
				\begin{bmatrix}
					0 \\ 1
				\end{bmatrix}\\
				& \textbf{y} \geq 0\\
				& \alpha \text{ unrestricted}.
			\end{array}
		\end{equation}
	\end{minipage}
\end{equation*}

When considering optimal variables for the dual LP, it will be convenient to work with a decomposed form of the matrix $\mB$. As an example, when $k = 3$, the matrix $\mB$ is given by
\begin{equation}
	\label{B3}
	\mB = 	\begin{bmatrix}
		3(1-\vg_3(B)) & -2 \vg_3(B) & -\vg_3(B) & 0  \\	
		-3\vg_2(B)& 2(1-\vg_2(B)) & -\vg_2(B) & 0   \\	
		0 & -\vg_2(A) & 2(1-\vg_2(A)) & -3 \vg_2(A) \\	
		0 & -\vg_3(A) & -2\vg_3(A) & 3(1-\vg_3(A))   \\
	\end{bmatrix}
\end{equation}
This matrix can be decomposed as follows:
\begin{equation}
	\label{B3decomposed}
	\mB = \begin{bmatrix}
		3 & &&&  \\	
		&2 & &&  \\	
		& & 2& \\	
		& & & 3  
	\end{bmatrix} - \begin{bmatrix}
		\vg_3(B) & &&&  \\	
		&\vg_2(B)  & &&  \\	
		& & \vg_2(A) & \\	
		& & & \vg_3(A) 
	\end{bmatrix} 
	\begin{bmatrix}
		3 & 2 & 1& 0  \\	
		3 & 2 & 1& 0  \\	
		0 & 1 & 2& 3\\	
		0 & 1 & 2& 3
	\end{bmatrix}.
\end{equation}
In general for odd $k$, we can decompose the matrix $\mB$ in the following way:
\begin{equation}
	\label{Bdecomp}
	\mB = \textbf{D}_k - \textbf{D}_\vg \textbf{R}, \\
\end{equation}
where $\textbf{D}_k$ is a diagonal matrix with diagonal entries $[k, k-1, \cdots, r, r+1, r+1, r, \cdots, k-1, k]$, and $\textbf{D}_\vg$ is a diagonal matrix with diagonal entries $[\vg_k(B), \vg_{k-1}(B), \cdots \vg_{r}(B), \vg_r(A), \cdots , \vg_{k-1}(A), \vg_k(A)]$. The first $r$ rows of matrix $\textbf{R}$ are $\begin{bmatrix} k & k-1 & \cdots &1 & 0\end{bmatrix}$, and the next $r$ rows are $\begin{bmatrix} 0 & 1 & \cdots & k-1 & k\end{bmatrix}$:
\begin{align*}
	\textbf{D}_k &=  
	\begin{bmatrix}
		k & &&& & \\	
		& \ddots &&&& \\
		& & r & & & \\
		& & & r & & \\
		& & & & \ddots & \\
		& & & & & k
	\end{bmatrix}\\
	\textbf{D}_\vg &=
	\begin{bmatrix}
		\vg_k(B) & &&& & \\	
		& \ddots &&&& \\
		& & \vg_r(B)  & & & \\
		& & & \vg_r(A)  & & \\
		& & & & \ddots & \\
		& & & & & \vg_k(A) 
	\end{bmatrix} \\
	\textbf{R} &=
	\begin{bmatrix}
		k & \cdots & 1& 0  \\	
		\vdots & \vdots & \vdots & \vdots \\
		k & \cdots & 1& 0  \\	
		0 & 1 &\cdots & k \\
		\vdots & \vdots & \vdots & \vdots \\
		0 & 1 &\cdots & k
	\end{bmatrix}.
\end{align*}

We can quickly obtain a solution with objective score of $\alpha = 0$ for the primal linear program. Recall that by assumption, the baseline scores correspond to affinity scores for some $k$-uniform hypergraph $G$. Let $m_t$ be the number of hypergraphs of type $t$ in $G$, and $M$ be the total number of hyperedges. Then define a set of primal solutions $\vx$ by setting $x_t = m_t/M$ for $t \in \{0, 1, 2, \hdots k\}$, and set $\gamma = 0$. The fact that the affinity score in this hypergraph equals the baseline score means that this set of primal variables is feasible for the primal LP. The following lemma, which proves our majority homophily result for odd $k$ in Theorem~\ref{thm:majorityimpossible}, shows that these are in fact optimal primal solutions. 

\begin{relemma}
	\label{dualvariables-app}
	For an odd integer $k$ and $r = (k+1)/2$, define $\delta = 2k\sum_{t = r}^k \frac{1}{t}$, and consider the following set of dual variables:
	\begin{align}
		\alpha &= 0 \\
		\label{ybk}
		y_{B,k} &= \frac{2}{\delta} \cdot \frac{\sum_{i = r}^k (\frac{k}{i} - 1)\vg_i(B)} {1 - \sum_{i = r}^k (2 - \frac{k}{i}) \vg_i(B)} \\
		\label{ybt}
		y_{B,t} &= \frac{2}{\delta}\left( \frac{k}{t} - 1\right) + \left(2 - \frac{k}{t}\right) y_{B,k} \,\,\text{ for $t \in \{r, \cdots, k-1\}$} \\
		\label{yat}
		y_{A,t} &= \frac{2k}{\delta t} - y_{B,t} \,\,\text{ for $t \in \{r, \cdots, k\}$.}
	\end{align}
	If $Y = \sum_{t = r}^k y_{A,t} + y_{B,t}$, then the set of normalized dual variables defined by $\tilde{y}_{X,t} = y_{X,t}/Y$ for $ X\in \{A,B\}$ and $t \in \{r, \hdots , k\}$ is feasible for the dual LP for majority homophily.
\end{relemma}

\begin{proof}
	When considering variables for the dual LP~\eqref{dual}, first recall that the matrix $\mB$ can be decomposed into the form $\mB = \textbf{D}_k - \textbf{D}_\vg \textbf{R}$, where $\textbf{D}_\vg$ is a diagonal matrix with diagonal entries 
	\[[\vg_k(B), \vg_{k-1}(B), \cdots, \vg_{r}(B), \vg_r(A), \cdots , \vg_{k-1}(A), \vg_k(A)].\]
	In this way, each row and column of $\mB$ can be mapped to a pair $(X,t)$ where $t \in \{r, r+1, \hdots , k\}$ represents a hyperedge type and $X \in \{A,B\}$ is a class. Therefore, each dual variable is also associated with an $(X,t)$ pair, which is why we doubly-index dual variables in $\vy$ as follows:
	\begin{equation*}
		\vy^T  = \begin{bmatrix} y_{B,k} & y_{B,k-1} & \hdots & y_{B,r} & y_{A,r} & \hdots & y_{A,k-1} & y_{A,k} \end{bmatrix}.
	\end{equation*}
	
	In the remainder of the proof, we will show that the unnormalized dual variables given in the lemma statement are nonnegative and satisfy $\mB^T \vy = 0$. In their current form, these variables do not satisfy $\ve^T \vy = 1$, but this can easily be fixed by dividing the entries of $\vy$ by their sum to produce a vector $\hat{\vy}$ whose entries sum to 1. At this point, the vector $\hat{\vy}$ along with $\alpha = 0$ provides a feasible solution with objective score of zero for the dual LP, which will conclude the proof. Thus, in the remainder of the proof we prove that the variables in~\eqref{ybk},~\eqref{ybt}, and~\eqref{yat} are nonnegative and satisfy $\mB^T \vy = 0$.
	
	\emph{Nonnegativity of dual variables.}
	The nonnegativity of dual variables follows from the nonnegativity of baseline scores. Note in particular that 
	\begin{align*}
		\sum_{i = r}^k (2- k/i) \vg_i(B) \leq \sum_{i = r}^k \vg_i(B)  < 1
	\end{align*}
	which shows that the denominator of $y_{B,k}$ is positive. The numerator is also positive by inspection. Since $y_{B,k} > 0$, we can see that all three terms in $y_{B,t}$ are positive. Finally,
	\begin{align*}
		y_{A,t} &= \frac{2}{\delta} - \left(2 - \frac{k}{t}\right) y_{B,k} \geq \frac{2}{\delta} - y_{B,k}\\
		&= \frac{2}{\delta}\left(1 -  \frac{\sum_{i = r}^k (\frac{k}{i} - 1)\vg_i(B) } {1 - \sum_{i = r}^k (2 - \frac{k}{i}) \vg_i(B) }\right) \\
		&= \frac{2}{\delta}\left(\frac{ 1 - \sum_{i = r}^k (2 - \frac{k}{i})\vg_i(B)  - \sum_{i = r}^k (\frac{k}{i} - 1)\vg_i(B) } {1 - \sum_{i = r}^k (2 - \frac{k}{i}) \vg_i(B) }\right) \\
		&= \frac{2}{\delta}\left(\frac{ 1 - \sum_{i = r}^k \vg_i(B) } {1 - \sum_{i = r}^k (2 - \frac{k}{i}) \vg_i(B) }\right).
	\end{align*}
	As before, the numerator and denominator are both positive, so $y_{A,t} > 0$.
	
	\textit{Proving $\mB^T \vy = 0$.}
	Given the decomposition $\mB = \textbf{D}_k - \textbf{D}_\vg \textbf{R}$, we can see that proving $\mB^T \vy = 0 $ is equivalent to showing 
	\begin{equation}
		\label{dkequation}
		\textbf{D}_k\vy = \textbf{R}^T \textbf{D}_\vg \vy.
	\end{equation}
	If we doubly index entries of $\textbf{D}_k \vy$ using the same indexing as the $\vy$ entries, we see that
	\begin{equation*}
		[\textbf{D}_k \vy]_{B,t} = t y_{B,t}.
	\end{equation*}
	Meanwhile, the right hand side of~\eqref{dkequation} is
	\begin{equation*}
		\textbf{R}^T \textbf{D}_\vg \vy 
		=\begin{bmatrix}
			k & k&\cdots& k & 0 &\cdots & 0  \\	
			k-1 & k-1 &\cdots& k-1 & 1 &\cdots & 1  \\	
			\vdots & \vdots & \ddots & \vdots & \vdots & \ddots & \vdots \\
			1 & 1&\cdots& 1 & k-1&\cdots & k-1  \\	
			0 & 0&\cdots& 0 & k&\cdots & k 
		\end{bmatrix} \begin{bmatrix}
			y_{B,k} \vg_k(B)\\
			y_{B,k-1} \vg_{k-1}(B)\\
			\vdots \\
			y_{B,r} \vg_r(B) \\
			y_{A,r} \vg_r(A) \\
			\vdots \\
			y_{A,k} \vg_k(A)
		\end{bmatrix}
	\end{equation*}
	After canceling $k$ from both sides, the first row of the matrix equation~\eqref{dkequation} is:
	\begin{align*}
		& y_{B,k}  = \sum_{i = r}^k y_{B,i} \vg_i(B) = \sum_{i = r}^k \left[\frac{2}{\delta} \left( \frac{k}{i} - 1\right) + \left(2 - \frac{k}{i}\right) y_{B,k}\right] \vg_i(B) \\
		\iff & y_{B,k} \left(1 - \sum_{i = r}^k (2-k/i) \vg_i(B)  \right) = \frac{2}{\delta} \sum_{i = r}^k (k/i - 1)\vg_i(B)  \\
		\iff & y_{B,k} = \frac{2}{\delta} \cdot \frac{\sum_{i = r}^k (\frac{k}{i} - 1)\vg_i(B) } {1 - \sum_{i = r}^k (2 - \frac{k}{i}) \vg_i(B) }.
	\end{align*}
	So this holds by the definition of $y_{B,k}$ in~\eqref{ybk}. Next, we show $[\textbf{D}_k \vy]_{B,t} = [\textbf{R}^T \textbf{D}_b \vy]_{B,t}$ for $t \in \{r, \hdots, k-1\}$. Let $y = y_{B,k} = \sum_{i = r}^k \vg_i(B) y_{B,i}$. Each such equation has the form
	\begin{align*}
		&\frac{2t}{\delta}(\frac{k}{t} - 1) + (2 - \frac{k}{t})y t = t\sum_{i = r}^k \vg_i(B) y_{B,i} + (k-t)\sum_{i = r}^k \vg_i(A)  y_{A,i}\\
		\iff &\frac{2}{\delta}(k - t) + (2t - k)y = ty + (k-t)\sum_{i = r}^k \vg_i(A)  y_{A,i} \\
		\iff & \frac{2}{\delta} - y = \sum_{i = r}^k \vg_i(A)  \left(\frac{2}{\delta} - (2- \frac{k}{i})y \right) \\
		\iff & \frac{2}{\delta}\left(1 - \sum_{i = r}^k \vg_i(A)  \right) = y \left( 1 - \sum_{i = r}^k \vg_i(A)  (2 - \frac{k}{i}) \right)\\
		\iff & y = \frac{2}{\delta} \frac{1 - \sum_{i = r}^k \vg_i(A) }{1 - \sum_{i = r}^k \vg_i(A)  (2 - \frac{k}{i}) }.
	\end{align*}
	Therefore, the equivalence between the first $r$ entries in the equation~\eqref{dkequation} will hold as long as we can prove that the following are both equivalent ways of writing $y = y_{B,k}$:
	\begin{equation}
		\label{twoy}
		y =y_{B,k} = \frac{2}{\delta} \frac{1 - \sum_{i = r}^k \vg_i(A) }{1 - \sum_{i = r}^k \vg_i(A)  (2 - \frac{k}{i}) } = \frac{2}{\delta} \frac{\sum_{i = r}^k (\frac{k}{i} - 1)\vg_i(B)} {1 - \sum_{i = r}^k (2 - \frac{k}{i}) \vg_i(B)} .
	\end{equation}
	We can prove this by using the fact that there is a hypergraph $G$ whose affinity scores equal the baseline scores. More specifically, for $i \in [k]$,
	\begin{align*}
		\vg_i(A) &= \frac{im_i}{D_A} \\
		\vg_i(B) &= \frac{im_{k-i}}{D_B},
	\end{align*}
	where $D_A = \sum_{i = 1}^k i m_i$ and $D_B = \sum_{i = 1}^k i m_{k-i}$.
	Re-writing the numerator and denominator of the first ratio in~\eqref{twoy} (and after scaling each expression by $\delta/2$) we get
	\begin{equation*}
		\frac{1 - \sum_{i = r}^k \vg_i(A)}{1 - \sum_{i = r}^k \vg_i(A) (2 - \frac{k}{i}) } = \frac{\frac{1}{D_A} \sum_{i = 1}^{r-1} i m_i}{1 - \frac{1}{D_A} \sum_{i = r}^k m_i (2i - k)} = \frac{\sum_{i = 1}^{r-1} i m_i}{D_A - \sum_{i=r}^k m_i (2i - k)}.
	\end{equation*}
	Similarly, we re-write the second ratio:
	\begin{equation*}
		\frac{\sum_{i = r}^k (\frac{k}{i} - 1)\vg_i(B)} {1 - \sum_{i = r}^k (2 - \frac{k}{i}) \vg_i(B)} = \frac{\frac{1}{D_B} \sum_{i = r}^k (k - i)m_{k-i}}{1 - \frac{1}{D_B} \sum_{i = r}^k (2i - k) m_{k-i} } = \frac{\sum_{i = r}^k (k-i) m_{k-i}} { D_B - \sum_{i = r}^k (2i - k) m_{k-i}}.
	\end{equation*}
	Written this way, we see that the numerators are the same since
	\begin{equation}
		\sum_{i = r}^k (k-i) m_{k-i}
		=\sum_{t = 1}^{r - 1} t m_t.
	\end{equation}
	It remains to show that the denominators are equal, which we prove by considering a sequence of equivalent statements:
	\begin{align*}
		D_B - \sum_{i = r}^k (2i - k) m_{k-i} &= D_A - \sum_{i=r}^k m_i (2i - k) \\
		\iff  \sum_{i = 1}^k i m_{k-i} - \sum_{i = r}^k (2i - k) m_{k-i} &= \sum_{i = 1}^k i m_i - \sum_{i=r}^k m_i (2i - k)  \\
		\iff \sum_{i = 0}^k i m_{k-i} - \sum_{i = 0}^k i m_i &= \sum_{i = r}^k (2i - k) m_{k-i} - \sum_{i=r}^k m_i (2i - k)  \\
		\iff \sum_{i = 0}^k [(k-i) - i ] m_i &= -\sum_{j = 0}^{r-1} (2j - k) m_{j} - \sum_{i=r}^k m_i (2i - k)\\
		\iff \sum_{i = 0}^k (k - 2i)m_i &= \sum_{i = 0}^k (k - 2i) m_i.
	\end{align*}
	At this point our proof has shown that the first $r$ rows (the rows corresponding to class B), of the matrix equation $\mD_k \vy = \mR^T \mD_{\vg} \vy$ hold. We use a similar approach to show the remaining $r$ rows also hold. First note that for $t \in \{r, r+1, \hdots, k\}$,
	\[
	[\mD_k \vy]_{A,t} = t \cdot (y_{A,t}) = t \cdot \left( \frac{2}{\delta} - \left(2 - \frac{k}{t}\right) y_{B,k}  \right).
	\]
	The last entry of the equation is $[\mD_k \vy]_{A,k} = [\mR^T \mD_{\vg} \vy]_{A,k}$, which holds by the following sequence of equivalent statements:
	\begin{align*}
		& ( 2/\delta - y_{B,k}) = \sum_{i = r}^k y_{A,i} \vg_i(A) \\
		\iff  &( 2/\delta - y_{B,k}) = \sum_{i = r}^k \vg_i(A) \left( \frac{2}{\delta} - \left(2 - \frac{k}{i}\right) y_{B,k}  \right) \\
		\iff & \frac{2}{\delta}\left(1 - \sum_{i = r}^k \vg_i(A)\right) = y_{B,k} \left( 1 - \sum_{i = r}^k (2- \frac{k}{i}) \vg_i(A) \right).
	\end{align*}
	The last equation holds by the equivalent ways of writing $y_{B,k}$ shown in~\eqref{twoy}.
	
	Finally, we confirm that $[\mD_k \vy]_{A,t} = [\mR^T \mD_{\vg} \vy]_{A,t}$ holds for $t \in \{r, \hdots, k-1\}$. Let $y = y_{B,k}$ and recall from the last step that $2/\delta - y = \sum_{i = r}^k y_{A,i} \vg_i(A) $. Each equation corresponding to one of the last $r$ rows has the form
	\begin{align*}
		&t \cdot \left( \frac{2}{\delta} - \left(2 - \frac{k}{t}\right) y \right) = t\sum_{i = r}^k \vg_i(A)y_{A,i} + (k-t)\sum_{i = r}^k \vg_i(B) y_{B,i}\\
		\iff &\left( \frac{2t}{\delta} - \left(2t - k \right) y \right)
		= \frac{2t}{\delta} - yt + (k-t)\sum_{i = r}^k \vg_i(B) y_{B,i} \\
		\iff & -(t - k) y = (k - t)\sum_{i = r}^k \vg_i(B)y_{B,i} \\
		\iff & y = \sum_{i = r}^k \vg_i(B)y_{B,i},
	\end{align*}
	which again was shown in previous steps. At this point we have shown that all entries in the equation $\mD_k \vy= \mR^T \mD_{\vg} \vy$  hold. Therefore, the dual variables are nonnegative and satisfy $\mB^T \vy = 0$, concluding the proof.
\end{proof}

\subsection{Proof for even $k$}
When $k$ is even, the definition of majority homophily does not place any restriction on type-$(k/2)$ hyperedges for either class, since neither class is strictly in the majority for these hyperedges. If the number of type-$(k/2)$ hyperedges is small enough, it is possible for both classes to exhibit majority homophily. As one example, starting with a complete hypergraph and deleting all hyperedges of type-$(k/2)$ will produce a hypergraph where both classes satisfy majority homophily. However, we can still prove an analogous impossibility result for even $k$ if we add one extra constraint. We restate and prove our result for even $k$.
\paragraph{Theorem~\ref{thm:majorityimpossible}, even $k$.} \emph{When $k$ is even, it is impossible for both classes $A$ and $B$ to exhibit majority homophily if additionally $\h_\ell(A) > \vg_{\ell}(A)$ or $\h_\ell(B) > \vg_{\ell}(B)$ for $\ell = k/2$.}

\begin{proof}
	The proof follows the same steps as the proof for odd $k$, with minor alterations to the linear program and the optimal dual variables. We highlight key changes that must be made and for brevity skip steps that are nearly identical to the proof of the previous result.
	
	Let $\ell = k/2$. Without loss of generality we prove the result is impossible if we restrict $\h_\ell(A) > \vg_{\ell}(A)$. By symmetry, the same impossibility result holds if we added the new constraint for class $B$ instead. We begin by altering the LP from~\eqref{lpmain-app} to include an additional constraint:
	\begin{equation}
		\label{extra}
		\ell\cdot x_\ell - \vg_t(A) \cdot \sum_{i = 1}^k i \cdot x_i \geq \gamma.
	\end{equation}
	For this new linear program, we can again confirm that the optimal score $\gamma^*$ will be greater than zero if and only if it is possible for both classes to exhibit majority homophily \emph{and} for constraint~\eqref{extra} to hold.
	We then can again re-write the LP and its dual in the form shown in~\eqref{primal} and~\eqref{dual}, by extending the matrix $\mB$ to include one extra row to account for the new constraint. By our assumptions about baseline scores, we know there exists a hypergraph $G$, with $M$ total hyperedges and $m_i$ hyperedges of type-$i$, such that the affinity scores of $G$ equal the baseline scores in question. A primal feasible solution with an objective score of zero can then be realized by setting $x_i = m_i /M$ and $\gamma = 0$.
	
	Next, let $r = k/2 +1$
	and construct the following set of dual variables:
	\begin{align*}
		\label{dualeven}
		y_{B,k} &= \frac{2}{\delta} \cdot \frac{\sum_{i = r}^k (\frac{k}{i} - 1)\vg_i(B) } {1 - \sum_{i = r}^k (2 - \frac{k}{i}) \vg_i(B) } \\
		y_{B,t} &= \frac{2}{\delta}\left( \frac{k}{t} - 1\right) + \left(2 - \frac{k}{t}\right) y_{B,k} \,\,\text{ for $t \in \{r, \cdots, k-1\}$} \\
		y_{A,t} &= \frac{2k}{\delta t} - y_{B,t} \,\,\text{ for $t \in \{r, \cdots, k\} $} \\
		y_{A,\ell} & = 2/\delta,
	\end{align*}
	where $\delta = 2+ 2k \sum_{t = r}^k \frac{1}{t}$. 
	The result again relies on the fact that $y_{B,k}$ can be written in two ways, using the fact that baseline scores correspond to affinity scores for some hypergraph $G$:
	\begin{equation*}
		y_{B,k} = \frac{2}{\delta}\frac{1 - \sum_{i = \ell}^k \vg_i(A)}{1 - \sum_{i = \ell }^k \vg_i(A)  (2 - \frac{k}{i}) } = \frac{2}{\delta}\frac{\sum_{i = r}^k (\frac{k}{i} - 1)\vg_i(B) } {1 - \sum_{i = r}^k (2 - \frac{k}{i}) \vg_i(B) }.
	\end{equation*}
	Using the same basic set of steps used in Lemma~\ref{dualvariables-app}, we can show that $\mB^T \vy = 0$ and $\vy \geq 0$. Scaling the variables to sum to one produces a dual feasible solution with an objective score of zero, which proves the result. 
\end{proof}

\subsection{Impossibility Results for Normalized Bias Scores}
{ One alternative approach to measuring an affinity score's deviation from baseline is to consider the normalized bias score introduced in the main text. For a class $X$, the type-$t$ normalized bias score is 
	\begin{equation}
		\label{eq:normbias-app}
		\textbf{f}_t(X) = \begin{cases}
			\frac{\h_t(X) - \bs_t(X)}{1 - \bs_t(X)} & \text{ if $\h_t(X) \geq \bs_t(X)$ } \\ \\
			\frac{\h_t(X) - \bs_t(X)}{\bs_t(X)} & \text{ if $\h_t(X) < \bs_t(X)$. } \\
		\end{cases}
	\end{equation}
	Our existing notion of strict majority homophily is equivalent to requiring $\textbf{f}_t(X) > 0$ whenever $t > k/2$. We see therefore that the same impossibility results and combinatorial limits apply to a natural notion of majority homophily for normalized bias scores.
	The natural way to define strict monotonic homophily for normalized bias scores is to require $\textbf{f}_t(X) > \textbf{f}_{t-1}(X)$ whenever $t > k/2$, which can be different from monotonicity of ratio scores. Nevertheless, in our empirical results we find that ratio scores and normalized bias scores often increase and decrease in similar patterns. Furthermore, we can prove the same type of impossibility results for normalized bias scores by adding one natural assumption regarding the balance in homophily levels exhibited by two node classes.
	\begin{theorem}
		\label{thm:normbias}
		Let $H = (V,E)$ be a two-class $k$-uniform hypergraph and $\{\bs_i(X)\colon  i\in [k], X \in \{A,B\} \}$ be realizable baseline scores. Let $k$ be odd and $r = (k+1)/2$. If $\textbf{f}_r(A)$ and $\textbf{f}_r(B)$ have the same sign, then it is impossible for both $A$ and $B$ to simultaneously exhibit strict monotonic homophily in terms of normalized bias scores.
	\end{theorem}
	\begin{proof}
		If we consider first of all the case where $\textbf{f}_r(A) > 0$ and $\textbf{f}_r(B) > 0$, assuming that normalized bias scores are strictly increasing for both classes implies that $\textbf{f}_t(X) > 0$ whenever $t > k/2$ for each $X \in \{A,B\}$. This would mean that both classes satisfy strict majority homophily, which is impossible. If on the other hand we have $\textbf{f}_r(A) \leq 0$ and $\textbf{f}_r(B) \leq 0$, then $\textbf{f}_r(X) > \textbf{f}_{r-1}(X)$ means that
		\begin{equation*}
			\frac{\h_r(X) - \bs_r(X)}{\bs_r(X)} > \frac{\h_{r-1}(X) - \bs_{r-1}(X)}{\bs_{r-1}(X)}  \implies \frac{\h_r(X)}{\bs_r(X)} > \frac{\h_{r-1}(X)}{\bs_{r-1}(X)}.
		\end{equation*}
		Assuming that this holds for both classes at once is again a contradiction, as shown in the proof of Theorem~\ref{thm:oddkmonotonic}.
	\end{proof}
	We can similarly prove impossibility results for even values of $k$ by adding an additional assumption as we did in Theorem~\ref{thm:evenkmonotonic}. The assumption that $\textbf{f}_r(A)$ and $\textbf{f}_r(B)$ share the same sign is in line with the goal of trying to understand whether two classes can satisfy the same homophily properties at the same time. We conjecture that Theorem~\ref{thm:normbias} holds without this assumption, though we leave a more in depth analysis for future work. In any case, this theorem confirms that even if there is some way for both classes to have strictly increasing normalized bias scores (which there may not be), there will still be a fundamental imbalance in their affinity scores and in the level of homophily they exhibit. This further confirms the message that natural notions of group homophily are governed by subtle combinatorial limits that must exist independent of human preferences and choices.}

\section{Derivation and Results for Alternative Affinity Scores}
\label{sec:alternative}
The hypergraph affinity scores we consider in the main text and the previous two sections of the appendix are based on ratios of typed degrees for nodes in a certain class. This directly generalizes the standard approach that has been used for measuring homophily in graphs~\cite{altenburger2018monophily}. Another natural approach is to consider affinity scores defined by ratios of hyperedge types. Formally, given the same $k$-uniform hypergraph $H = (V,E)$ where $m_t$ denotes the number of type-$t$ hyperedges, we can measure the following \emph{alternative} affinity scores:
\begin{align}
	\label{alta}
	\va_t(A)  &= \frac{m_t}{\sum_{i = 1}^k m_i} \\
	\label{altb}
	\va_t(B) &= \frac{m_{k-t}}{\sum_{i = 1}^k m_{k-i} }.
\end{align}
For each class $X \in \{A,B\}$, these ratios directly measure the proportion of hyperedges of type-$(X,t)$, among all hyperedges involving at least one class $X$ node. In this section we show that all of our main theoretical results also hold for these alternative scores. This first of all highlights that our main results on hypergraph homophily are broadly true for a wide range of notions of affinity scores. Furthermore, as we shall also see, our main impossibility results are in fact easier to show for these alternative scores, and our proof that these scores correspond to maximum likelihood estimates of a certain model parameter is more direct and does not require approximations. 

\subsection{Baseline Scores for Alternative Affinities}
Analogous to our approach for standard affinity scores, we define the baseline score for $\va_t(X)$ to be the probability that we obtain a type-$(X,t)$ hyperedge if we select a $k$-tuple uniformly at random from among all $k$-tuples involving at least one $X$ node:
\begin{align}
	\label{bltx}
	\bs_t(X)  = \frac{{|X| \choose t}{N-|X| \choose k-t}}{\sum_{i = 1}^k {|X| \choose i}{N-|X| \choose k-i}}.
\end{align}
The denominator counts all $k$-tuples involving at least one node from $X$.
For simplicity, we use the same notation as we did for standard affinity scores. 
\begin{proposition}
	Let $H_{k,n}^* = (V,E)$ be the complete $k$-uniform hypergraph on $n$ nodes with two node classes  $\{A,B\}$.
	 For $t \in [k]$, the type-$t$ alternative affinity scores for class $X \in \{A,B\}$ equals the type-$t$ alternative baseline score~\eqref{bltx}: $\va_t(X) = \bs_t(X)$.
\end{proposition}
The proof of this proposition is omitted as it follows the same steps as Proposition~\ref{lem:base-app}. We will prove homophily impossibility results for the following more general notion of baseline scores.

\begin{definition}
	Let $k$ be a fixed constant. The set of scores $\{ \vg_t(X) \colon X \in \{A,B\}, t \in [k]\}$ are generalized baseline scores for alternative affinity scores $\{\va_t(X) \colon X \in \{A,B\}, t \in [k]\}$ if the following two conditions hold:
	\begin{itemize}
		\item $\vg_t(X) > 0$ for $X \in \{A,B\}$ and $t \in [k]$.
		\item The scores $\{\vg_t(X) \}$ correspond to alternative affinity scores for some $k$-uniform hypergraph $G$ with two node classes.
	\end{itemize}
\end{definition}

\subsection{Interpreting Scores as Maximum Likelihood Estimates}
We can interpret the alternative affinity score $\va_t(A)$ as the maximum likelihood estimate for a certain affinity parameter of a binomial distribution for hyperedge data. An analogous interpretation also applies for class $B$. 
We begin by considering a slight variation of the cardinality-based HSBM. This new model still considers a set of $n$ nodes separated into two classes $\{A,B\}$, and generates typed hyperedges based on a set of probabilities $\vp = \begin{bmatrix} p_0 & p_1 & \cdots & p_k \end{bmatrix}$. Let $\mathcal{K}_t$ be the set of $k$-tuples of type-$t$. For each $e \in \mathcal{K}_t$, let $X_e$ be a Poisson random variable with parameter $p_t$, and let this represent the number of hyperedges placed at $e$:
\begin{equation}
	\label{poissonXe}
	X_e \sim \text{Poisson}(p_t).
\end{equation}
Recall that the cardinality-based HSBM differs in that we instead defined $X_e \sim \text{Bernoulli}(p_t)$. When we consider $p_t = o(1)$, which will typically be the case, this Poisson distribution will be very close to a Bernoulli with parameter $p_t$. For a hypergraph generated from this distribution, let $M_j$ be the random variable representing the number of hyperedges of type-$j$ in $H$, which as a sum of Poisson random variables will also be Poisson:
\begin{equation}
	\label{Mj}
	M_j = \sum_{e \in \mathcal{K}_j} X_e \sim \text{Poisson}(K_jp_j).
\end{equation}
Define $M_A$ to be the random variable denoting the total number of hyperedges involving at least one node in $A$, which is also Poisson distributed:
\begin{equation}
	\label{MA}
	M_A = \sum_{j = 1}^k M_j \sim \text{Poisson}\left(\sum_{j = 1}^k K_jp_j\right)
\end{equation}
The random variable representing the number of hyperedges that are \emph{not} type-$j$ is given by
\begin{equation}
	\label{Mnotj}
	\tilde{M}_j = \sum_{i \neq j} M_i \sim \text{Poisson}\left(\sum_{i \neq j} K_ip_i\right).
\end{equation}

Consider now a fixed hypergraph $H$ with $m_A$ hyperedges involving at least one class-$A$ node. If we assume this hypergraph was drawn from the random distribution given above, the random variable $M_t$, conditioned on the observed hyperedge count $m_A$, will be binomial:
\begin{equation}
	M_t \; |\; m_A \sim \text{Binom}(m_A, f_t)
\end{equation}
where $f_t$ is an affinity parameter:
\begin{equation}
	f_t = \frac{K_t p_t}{\sum_{j = 1}^k K_j p_j}.
\end{equation}
To see why, note
\begin{align*}
	\pr(M_t = c \; | \; M_A = m_A) &= \frac{\pr(M_t = m_t \text{ and } M_A = m_A) }{\pr(M_A = m_A)} \\
	&= \frac{\pr(M_t = m_t \text{ and } \tilde{M}_j = m_A - c) }{\pr(M_A = m_A)} \\
	&= \frac{
		\left[\frac{1}{c!} (p_tK_t)^c \cdot e^{-p_tK_t} \right]
		\left[\frac{1}{(m_A-c)!} (\sum_{i \neq t} p_iK_i)^{m_A-c} \cdot e^{-\sum_{i \neq t}p_iK_i}\right]
	}{
		\frac{1}{(m_A)!} (\sum_{i =1}^k p_iK_i)^{m_A} \cdot e^{-\sum_{i =1}^kp_iK_i}  } \\
	&={m_A \choose c} \left(\frac{K_t p_t}{\sum_{j = 1}^k K_j p_j} \right)^{c} \left(1 -\frac{K_t p_t}{\sum_{j = 1}^k K_j p_j} \right)^{m_A-c}.
\end{align*}

Finally, given an observed hypergraph with $m_A$ hyperedges involving at least one node in $A$, the likelihood of observing $m_t$ hyperedges of type-$t$ under this model is
\begin{equation*}
	\pr(M_t = m_t \; | \; M_A = m_A) = {m_A \choose c} f_t^{m_t} (1-f_t)^{m_A - m_t}.
\end{equation*}
Taking a derivative of the log-likelihood function with respect to the parameter $f_t$ and setting it to zero will give the maximum likelihood estimate for $f_t$. This ends up being equal to the alternative type-$t$ affinity score,
\begin{equation*}
	f_t = \frac{m_t}{m_A} = \frac{m_t}{\sum_{i = 1}^k m_i}.
\end{equation*}

\subsection{Impossibility Results}
Analogous to our results for standard affinity scores, we define two notions of hypergraph homophily based on alternative baseline scores. Let $\{\vg_t(X) \colon X \in \{A,B\}, t \in [k]\}$ denote a set of generalized baseline scores. 
\begin{definition}
	Class $X \in \{A,B\}$ exhibits majority homophily if for all $t > k-t$,
	\begin{equation}
		\va_t(X) > \vg_t(X).
	\end{equation}
\end{definition}
\begin{definition}
	Class $X \in \{A,B\}$ exhibits monotonic homophily if for all $t > k-t$,
	\begin{equation}
		\frac{\va_t(X)}{\vg_t(X)} > \frac{\va_{t-1}(X)}{\vg_{t-1}(X)} .
	\end{equation}
\end{definition}

The proof of the following impossibility result for monotonic homophily follows the same steps as the proof for Theorem~\ref{thm:oddkmonotonic}. In particular, these results can be shown by considering only two types of hyperedges.
\begin{theorem}
	Let $H$ be a two-class, $k$-uniform hypergraph and $\{\vg_t(X) \}$ be a set of generalized baseline scores for alternative affinity scores $\{\va_t(X)\}$. If $k$ is odd, it is impossible for both classes to exhibit monotonic homophily in terms of alternative affinity scores.  If $k$ is even, then both classes can exhibit monotonic homophily, but in this case $\frac{\va_\ell(X)}{\vg_\ell(X)} < \frac{\va_{\ell-1}(X)}{\vg_{\ell-1}(X)}$ for $\ell = k/2$ and $X \in \{A,B\}$.
\end{theorem}

In order to prove impossibility results for majority homophily, we use the same linear programming based proof technique that we used for Theorem~\ref{thm:majorityimpossible}.
\begin{theorem}
	\label{thm:oddkaltmaj}
	Let $H$ be a two-class, $k$-uniform hypergraph and $\{\vg_t(X) \}$ be a set of generalized baseline scores for alternative affinity scores $\{\va_t(X)\}$. If $k$ is odd, it is impossible for both classes to exhibit majority homophily in terms of alternative affinity scores. 
\end{theorem}
\begin{proof}
	Let $r = (k+1)/2$. The following linear program will have a strictly positive solution if and only if it is possible for both classes to exhibit majority homophily, with respect to the new alternative affinity scores:
	\begin{equation}
		\label{lpmain2}
		\begin{array}{lll}
			\text{maximize} & \gamma \\
			\text{subject to}&x_t - \vg_t(A) \cdot \sum_{i = 1}^k x_i \geq \gamma &\text{  for $t \in [k], t > k-t$} \\
			&x_{k-t} - \vg_t(B) \cdot \sum_{i = 1}^k  x_{k-i} \geq \gamma  &\text{  for $t \in [k], t > k-t$} \\
			& \sum_{i = 0}^k x_i = 1 &\\
			& x_i \geq 0 &\text{ for all $i \in \{0\} \cup [k]$}
		\end{array}
	\end{equation}
	The variable $x_t$ again represents the proportion of hyperedges where $t$ of the nodes are from class $A$. We can express the linear program as well as its dual in the same general form
	\begin{equation*}
		\begin{minipage}{.4\textwidth}
			\centering
			\textbf{Primal Linear Program}
			\begin{equation}
				\label{primal2}
				\begin{array}{ll}
					\max \,\, & \gamma \\
					\text{s.t. }  	&\begin{bmatrix}
						-\textbf{B} &\textbf{e} \\
					\end{bmatrix}
					\begin{bmatrix}
						\textbf{x} \\ \gamma
					\end{bmatrix}
					\leq 0\\
					&\begin{bmatrix}
						\textbf{e}^T & 0
					\end{bmatrix}
					\begin{bmatrix}
						\textbf{x} \\ \gamma
					\end{bmatrix}
					= 	1 \\
					&\textbf{x} \geq 0, \gamma \geq 0
				\end{array}
			\end{equation}
		\end{minipage}
		\begin{minipage}{.6\textwidth}
			\centering
			\textbf{Dual Linear Program}
			\begin{equation}
				\label{dual2}
				\begin{array}{ll}
					\min \,\, &\alpha \\
					\text{s.t. }  	&\begin{bmatrix}
						-\textbf{B}^T &\textbf{e} \\
						\textbf{e} & 0
					\end{bmatrix}
					\begin{bmatrix}
						\textbf{y} \\ \alpha
					\end{bmatrix}
					\geq 
					\begin{bmatrix}
						0 \\ 1
					\end{bmatrix}\\
					& \textbf{y} \geq 0\\
					& \alpha \text{ unrestricted}.
				\end{array}
			\end{equation}
		\end{minipage}
	\end{equation*}
	Above, $\ve$ is the all ones vector and the matrix $\mB$ encodes constraints of the form
	\begin{align}
		\label{constraintsb2}
		x_{k-t} - \vg_t(B) \cdot \sum_{i = 1}^k x_{k-i} &\geq \gamma  \text{   for $t = k, k-1, k-2, \hdots, r$} \\
		\label{constraintsa2}
		x_t - \vg_t(A) \cdot \sum_{i = 1}^k x_i &\geq \gamma \text{  for $t = r, r+1, \hdots, k$.} 
	\end{align}
	The matrix $\mB$ can be decomposed into the following form:
	\begin{equation}
		\label{bdecomp}
		\mB = \mI - \mD_\vg \mE,
	\end{equation}
	where $\mI$ is the $2r \times 2r$ identity matrix, $\mD_\vg$ is a diagonal matrix with diagonal entries $$[\vg_k(B), \vg_{k-1}(B), \cdots \vg_{r}(B), \vg_r(A), \cdots , \vg_{k-1}(A), \vg_k(A)],$$ and $\mE$ is a matrix that is all ones except for zeros in the last column of the first $r$ rows, and the first column in the last $r$ rows. Formally,
	\begin{align*}
		\mE_{ij} =
		\begin{cases}
			0 & \text { if $j = k$ and $i \in \{1,2, \hdots r\}$ }\\
			0 & \text { if  $j = 1$ and $i \in \{r+1,r+2, \hdots 2r\}$ }\\
			1 & \text{ otherwise.}
		\end{cases}
	\end{align*}
	The zero entries reflect the fact that hyperedges of type zero do not affect the affinity scores for class $A$, and type-$k$ hyperedges do not affect affinity scores for class $B$.
	
	A primal solution with an objective score of 0 can be obtained by setting $x_j = m_j/(\sum_{i = 0}^k m_j)$, where $m_i$ is the number of hyperedges of type $i$ in the hypergraph whose affinity scores are equal to the baseline scores $\{\vg_i(X)\}$. In order to find a set of dual variables with an objective score of zero, it suffices to find a vector $\vy$ that is strictly positive and satisfies $\mB^T \vy$. 	 Equivalently, given the decomposition of $\mB$ in~\eqref{bdecomp}, we want $\vy$ to satisfy
	\begin{equation}
		\label{todo}
		\vy = \mE^T \mD_\vg \vy.
	\end{equation}
	Each row and column of $\mB$ can be associated with a class $X$ and hyperedge type $t$, so we doubly index the dual variables as follows
	\begin{equation*}
		\vy^T  = \begin{bmatrix} y_{B,k} & y_{B,k-1} & \hdots & y_{B,r} & y_{A,r} & \hdots & y_{A,k-1} & y_{A,k} \end{bmatrix}.
	\end{equation*}
	Rows $2$ through $2r -1$ of $\mE^T$ have the value 1 in every column, so for this equation to hold we must have
	\begin{equation}
		\label{yxi}
		y_{X,i} = \sum_{i = r}^k \vg_i(A) y_{A,i} + \sum_{i = r}^k \vg_i(B) y_{B,i} .
	\end{equation}
	for  $i \in \{r, r+1, \cdots , k-1\}$ and $X \in \{A,B\}$. Thus, a necessary condition for satisfying~\eqref{todo} is that entries $2$ through $2r-1$ of $\vy$ are all equal. With this in mind, let $z > 0$ be a fixed positive value and construct a set of dual variables as follows:
	\begin{align*}
		y_{X,i} &= z  \;\text{ for $i \in \{r, r+1, \cdots , k-1\}$, $X \in \{A,B\}$} \\
		y_{B,k} &= \frac{z \cdot \sum_{i = r}^{k-1} \vg_i(B)}{1-\vg_k(B)} \\
		y_{A,k} &= \frac{z \cdot \sum_{i = r}^{k-1} \vg_i(A)}{1-\vg_k(A)}. 
	\end{align*}
	As long as $z > 0$, we can see that $y_{B,k}$ and $y_{A,k}$ will also be strictly positive. We can also set $z$ so that the entries of $\vy$ sum to one. As long as we can show $\mB^T \vy = 0$, this means we have a set of dual variables with an objective score of zero, confirming that majority homophily cannot hold for both classes simultaneously. 
	
	Variables $y_{B,k}$ and $y_{A,k}$ are explicitly chosen so that the first and last entries in equation~\eqref{todo} hold when all other variables are equal to a positive constant $z$. To check that $\mB^T \vy = 0$, we just need to confirm that~\eqref{yxi} holds. Recall that the baseline scores can be written in terms of hyperedge counts for some hypergraph $G$, i.e., for $i \in [k]$
	\begin{align*}
		\vg_i(A) &= \frac{m_i}{D_A} \\
		\vg_i(B) &= \frac{m_{k-i}}{D_B}
	\end{align*}
	where $D_A = \sum_{i = 1}^k m_i$ and $D_B = \sum_{i = 1}^k m_{k-i} = \sum_{i = 0}^{k-1} m_i$. Notice that $D_A - m_k = D_B-  m_0 = \sum_{i = 1}^{k-1} m_i$.
	Substituting our choice of variables into the right hand side of~\eqref{yxi}, we confirm that the equation holds:
	\begin{align*}
		&\sum_{i = r}^k  \vg_i(A) y_{A,i} + \sum_{i = r}^k \vg_i(B) y_{B,i}  \\
		&= z \cdot \sum_{i = r}^{k-1} \vg_i(A) + z \cdot\sum_{i = r}^{k-1} \vg_i(B) + \vg_k(A) y_{A,k} + \vg_k(B) y_{B,k} \\
		&= z \cdot \sum_{i = r}^{k-1} \vg_i(A) + z \cdot\sum_{i = r}^{k-1} \vg_i(B)	
		+ \vg_k(A)\frac{z \cdot \sum_{i = r}^{k-1} \vg_i(A)}{1-\vg_k(A)} 
		+ \vg_k(B)\frac{z \cdot \sum_{i = r}^{k-1} \vg_i(B)}{1-\vg_k(B)}  \\
		&=z \cdot \sum_{i = r}^{k-1} \vg_i(A) \left( 1 + \frac{\vg_k(A)}{1- \vg_k(A)}\right) + 
		z \cdot \sum_{i = r}^{k-1} \vg_i(B) \left( 1 + \frac{\vg_k(B)}{1- \vg_k(B)}\right) \\
		&=z \cdot \left(\sum_{i = r}^{k-1} \frac{m_i}{D_A} \left( \frac{1}{1- \vg_k(A)}\right) + 
		\sum_{i = r}^{k-1} \frac{m_{k-i}}{D_B} \left( \frac{1}{1- \vg_k(B)}\right) 
		\right) \\
		&= z \cdot \left(\sum_{i = r}^{k-1} \frac{m_i}{D_A} \left( \frac{D_A}{D_A - m_k}\right) + 
		\sum_{i = r}^{k-1} \frac{m_{k-i}}{D_B} \left( \frac{D_B}{D_B- m_0}\right) 
		\right) \\
		&= z \cdot \left(\sum_{i = r}^{k-1} \frac{m_i + m_{k-i}}{\sum_{i = 1}^{k-1} m_i} \right) \\
		&= z \cdot 1 = z.
	\end{align*}
\end{proof}

An analogous impossibility result also holds for even $k$ when we use alternative affinity scores.
\begin{theorem}
	If $k$ is even, it is impossible for both classes to exhibit monotonic homophily if additionally $\va_\ell(X) > \vg_\ell(X)$ for one class $X \in \{A,B\}$ when $\ell = k/2$. 
\end{theorem}
\begin{proof}
	The proof is nearly identical to the proof of Theorem~\ref{thm:oddkaltmaj}. For even $k$, let $r = \frac{k}{2} + 1$. We use the same linear program encoding the maximum possible amount of majority homophily, with an additional constraint for class $A$:
	\begin{equation*}
		x_{\ell} - g_\ell(A) \cdot \sum_{i = 1}^k x_{i} \geq \gamma.
	\end{equation*}
	The LP and its dual can again be written in the form~\eqref{primal2} and~\eqref{dual2}. The matrix $\mB$ is given by
	\begin{equation*}
		\mB = \mI - \mD_\vg \mE,
	\end{equation*}
	where $\mI$ is the $(k+1)\times (k+1)$ identity matrix, $\mD_\vg$ is a diagonal matrix with diagonal entries
	$$[\vg_k(B), \vg_{k-1}(B), \cdots, \vg_{r}(B), \vg_{\ell}(A), \vg_r(A), \cdots , \vg_{k-1}(A), \vg_k(A)],$$ and $\mE$ is a matrix that is all ones except for zeros in the last column of the first $k/2$ rows, and zeros in the first column in the last $k/2+1$ rows. The dual variables of the linear program can be indexed as follows
	\begin{equation*}
		\vy^T  = \begin{bmatrix} y_{B,k} & y_{B,k-1} & \hdots & y_{B,r} &y_{A,\ell} & y_{A,r} & \hdots & y_{A,k-1} & y_{A,k} \end{bmatrix}.
	\end{equation*}
	The proof again follows as long as we can find a nonnegative vector $\vy$ satisfying $\mB^T \vy = 0$, or equivalently $\vy = \mE^T \mD_\vg \vy$. In order to accomplish this, all but the first and last dual variable can be set equal to some positive value $z$, and then we can solve for $y_{A,k}$ and $y_{B,k}$:
	\begin{align*}
		y_{B,k} &= \frac{z \cdot \sum_{i = r}^{k-1} \vg_i(B)}{1-\vg_k(B)} \\
		y_{A,k} &= \frac{z \cdot \sum_{i = \ell}^{k-1} \vg_i(A)}{1-\vg_k(A)}. 
	\end{align*}
	The only difference from the dual variables in Theorem~\ref{thm:oddkaltmaj}
	is that the summation in the numerator of $y_{A,k}$ starts from $\ell$ rather than from $r$. The rest of the result follows by showing algebraically that $\mB^T \vy = 0$. 
\end{proof}

\section{Dataset Information and Additional Experimental Results}
\begin{figure}[t!]
	\includegraphics[width=\linewidth]{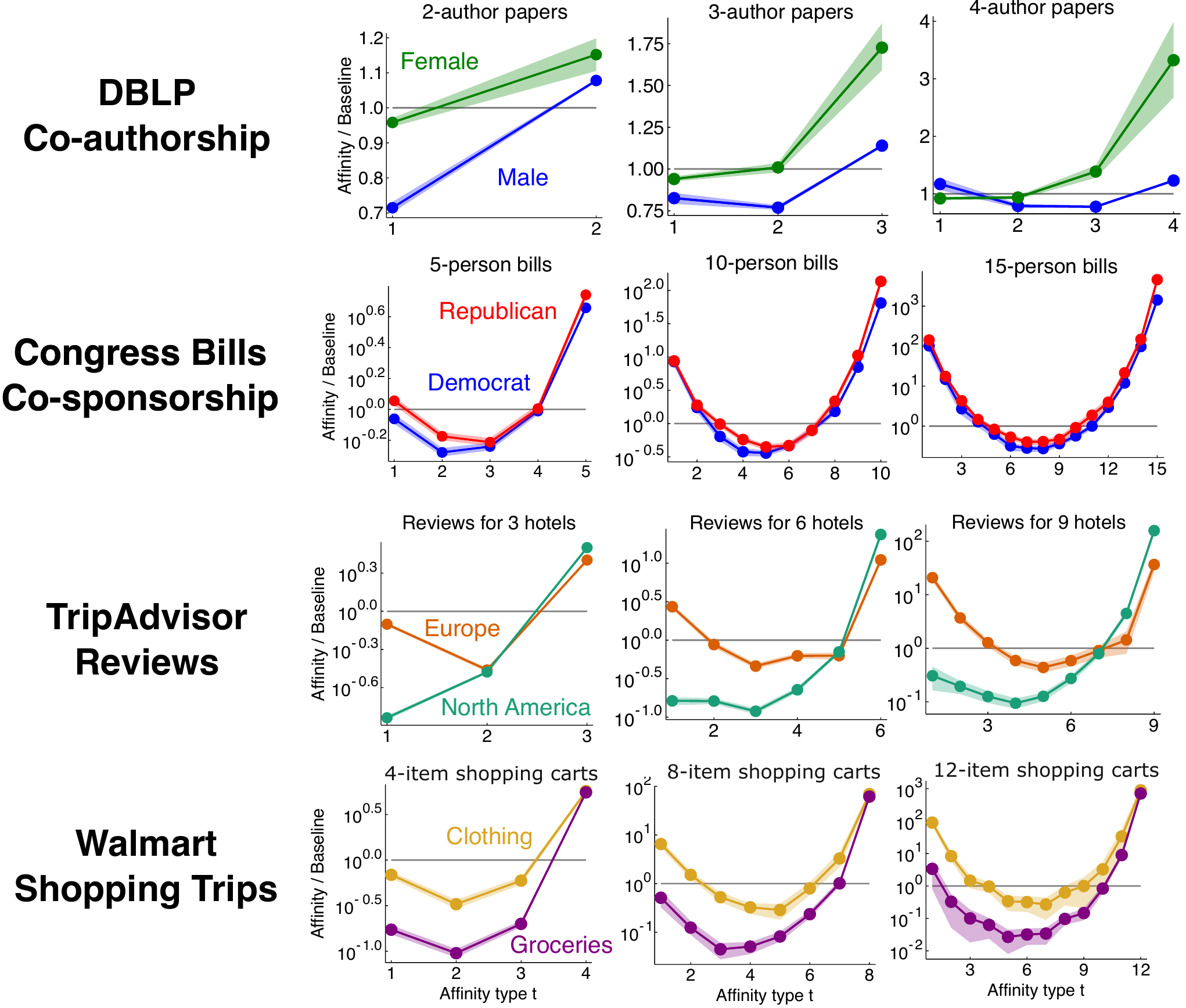} 
	\caption{\small 
		\textbf{Results from running a bootstrapping procedure to check the robustness of empirical results.} Ratio scores are shown for the DBLP co-authorship hypergraph (top row), the congress bills hypergraph (second row), the TripAdvisor hypergraph (third row), and the walmart hypergraph (last row).
		These results indicate that affinity and ratio scores in each case are robust to perturbations in the data.  
		For each hypergraph and a range of values of $k$, the procedure samples from the original set of size-$k$ hyperedges with replacement, and computes affinity scores and ratio scores each time. Solid lines in the plots indicate the true affinity scores computed on the entire dataset, which are nearly identical to the mean affinity score obtained from 100 runs of this procedure. Lighter colored regions show two standard errors above and below the mean. In many plots, the error region is too small to be perceptible. }
	\label{fig:main3}
\end{figure} 
\begin{figure}[t!]
	\centering
	\includegraphics[width=\linewidth]{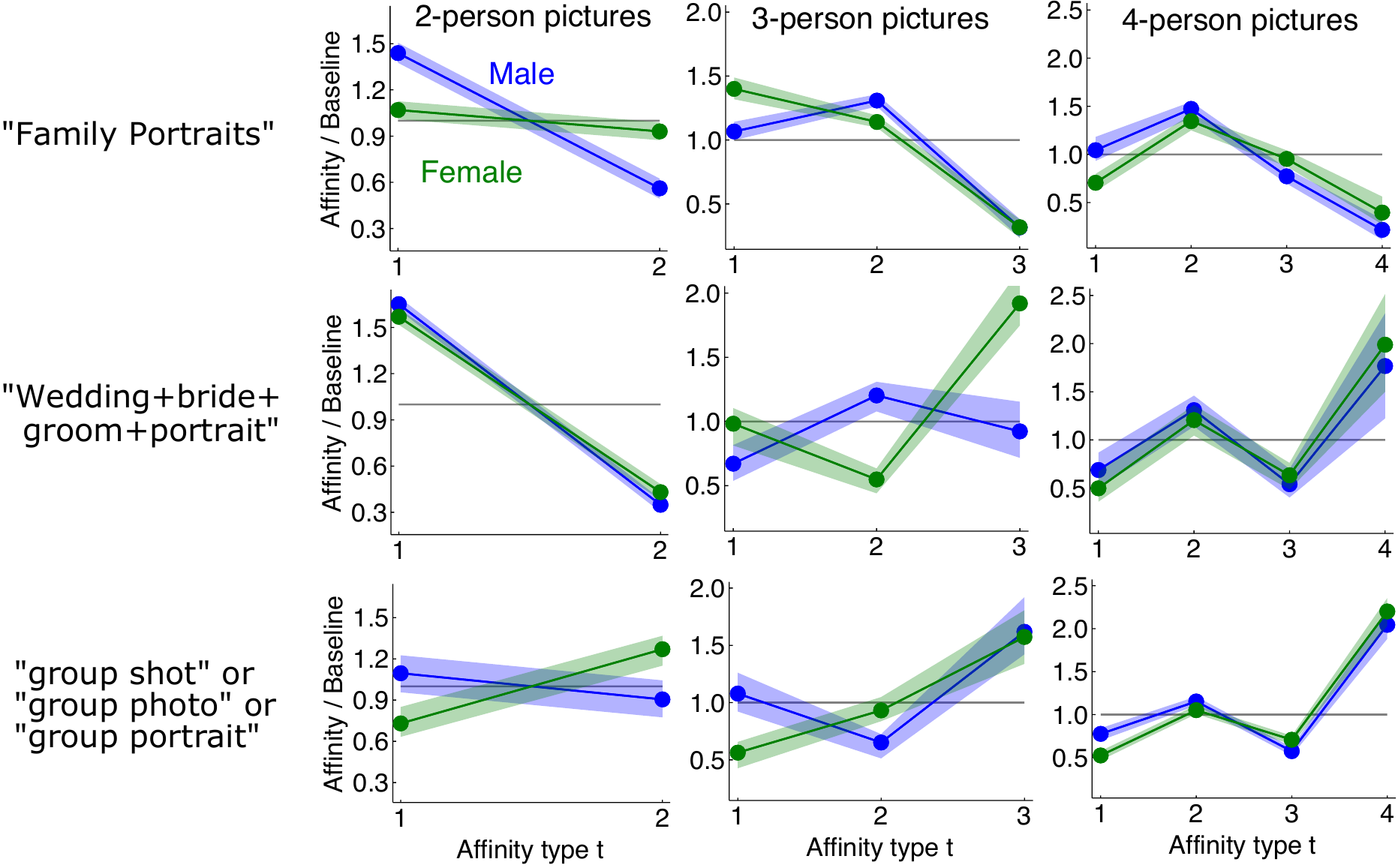} 
	\caption{\small 
		\textbf{Ratio scores obtained from a bootstrapping procedure on the group pictures dataset, indicating the robustness of our results for each group picture type.} The solid lines indicate true affinity scores, which are very close to the mean affinity score from taking 100 different samples from the data. Lighter colored regions show one standard error above and below the mean from this bootstrapping procedure. On the subsampled datasets, we continue to see all the same major trends and differences between pictures types as we do when using all of the data.}
	\label{fig:photos2}
\end{figure} 
In this section we provide details regarding the datasets used in the main text and additional experimental results. Further information about each dataset is available from the original source. For all experiments in the main text and the appendix, we used asymptotic baseline scores when forming the ratio plots { (see the Materials and Methods section). Code and data sufficient to reproduce all experimental results in the main text and in the appendix is available on Zenodo (\url{https://doi.org/10.5281/zenodo.7086798}) as well as on GitHub (\url{https://github.com/nveldt/HypergraphHomophily}).
	
	In order to provide an additional comparison against graph-based approaches, we compute and report graph homophily indices obtained by projecting each hypergraph into a graph based on co-participation in group interactions. These graph homophily indices do provide some information regarding the level of homophily exhibited by each class of nodes in each dataset, but they do not allow us to capture any notion of majority or monotonic homophily. These graph scores also often hide the way in which group size affects the level of homophily that is exhibited by a class.}

\paragraph{Co-authorship and gender} The co-authorship dataset we considered is the DBLP Records and Entries for Key Computer Science Conferences~\cite{dblp_data}, originally used in a study on women in computer science research~\cite{agarwal2016women} and available online at~\url{https://data.mendeley.com/datasets/3p9w84t5mr/1}. The data constitutes 16 years of publications at top computer science conferences, from 2000 to 2015, listed on the DBLP bibliography database. We consider only authors in the dataset whose gender is known with high confidence, and discard all papers including an author of unknown gender. The resulting hypergraph has 105256 nodes (82620 male, 22636 female), with a maximum hyperedge size of 21. The number of hyperedges between size 2 and 4 is 74134.

{
	We can project the hypergraph into a graph by including an edge between author $i$ and $j$ if they ever publish a paper together. If we perform this projection using all group interactions (i.e., all papers), the resulting graph has a homophily index of 0.270 for women and a homophily index of 0.821 for men. If we project only group interactions with 2-4 authors, the scores are similar: 0.261 and 0.828. In both cases, the scores are higher than the relative class proportions of 0.215 and 0.785. 
}

\paragraph{Congress bills and political parties}
The congress bills dataset is made up of legislative bills co-sponsored by US politicians in the Senate and House of Representatives. The original data was collected by James Fowler~\cite{Fowler-2006-connecting,Fowler-2006-cosponsorship}. For our experiments we consider a derivation of the dataset presented as timestamped hyperedges~\cite{Benson-2018-simplicial}, available online at~\url{https://www.cs.cornell.edu/~arb/data/congress-bills}. The hypergraph has 1718 nodes (810 Republican and 908 Democrat) and 83105 hyperedges ranging from size 2 to 25. 

{
	If we project all group interactions to a graph, the graph homophily indices for Republicans and Democrats are 0.499 and 0.591, respectively, and their class proportions are 0.471 and 0.529. In other words, graph homophily indices are only slightly above baseline.
	If we only consider group interactions of size 20 or smaller (the largest group size we considered for our hypergraph affinity scores in the main text), the homophily indices increase only slightly to 0.516 and 0.606. These scores are obtained if we use an unweighted graph projection, where nodes $i$ and $j$ have a unit weight edge if they {ever} co-participate in a hyperedge. We can also perform a weighted projection where the weight of an edge between nodes $i$ and $j$ equals the \emph{number} of hyperedges they both participate in. If we perform this weighted projection for all groups of size up to 20, graph homophily indices for Republicans and Democrats are 0.586 and 0.724 respectively. These higher scores indicate that it is common to see \emph{repeated} collaborations between the same two individuals from the same political party, which is also a valid notion of homophily. This is arguably a more accurate measure of the high levels of homophily present in the dataset, though this still does not capture the extreme notions of majority and monotonic group homophily that we reveal using our hypergraph measures.
}

\paragraph{TripAdvisor hotels and locations}
The TripAdvisor hypergraph is derived from a review dataset originally used for research on opinion mining from online reviews~\cite{wang2011latent}, available online at~\url{https://www.cs.virginia.edu/~hw5x/dataset.html}.
We associate each hotel in the dataset from North America or Europe as a labeled node in a hypergraph. We discard reviews from the original data that are cross-listed reviews from other travel sites, and reviews where the reviewer name is simply ``A TripAdvisor Member''. For each remaining unique reviewer, we construct a hyperedge joining all hotels they reviewed. The resulting hypergraph has 8956 nodes and 130570 hyperedges of size 2 to 87. 
{
	The main text shows MaHI and MoHI scores for group sizes up to $k = 13$. If we project groups of size 2 to 13 to an unweighted graph, the graph homophily indices for North American and Europe are 0.829 and 0.519, respectively. If we use a weighted graph projection, graph homophily indices increase to 0.893 and 0.613. In both cases, graph homophily indices are well above baseline class proportions of 0.502 and 0.498.
}

\paragraph{Walmart trips dataset}
The Walmart products datasets is available online at~\url{https://www.kaggle.com/c/walmart-recruiting-trip-type-classification}. This dataset was first made available as part of a Kaggle competition. Amburg et al.~\cite{Amburg-2020-categorical} derived a large hypergraph of co-purchased products based on the data, and identified department labels for each product. This derived hypergraph is available at~\url{https://www.cs.cornell.edu/~arb/data/walmart-trips/}. In our work, we restrict to the subset of grocery and clothes products, resulting in a hypergraph with 48480 nodes (26178 grocery products, and 22302 clothes products) and 47034 hyperedges of size 2 to 25.

{
	The main text shows MaHI and MoHI scores for group sizes up to $k = 14$. If we project groups of size 2 to 14 to an unweighted graph, the graph homophily indices for Groceries and Clothing are 0.942 and 0.618, respectively. If we use a weighted graph projection, these indices increase only slightly to 0.948 and 0.621. Baseline scores (i.e., class proportions) are 0.540 and 0.460 for Groceries and Clothing, respectively. 
}

\paragraph{Group pictures and gender} The three different group picture datasets~\cite{gallagher2009} were downloaded from
~\url{http://chenlab.ece.cornell.edu/people/Andy/ImagesOfGroups.html}. 
Images in the original dataset were obtained via three different Flickr search queries, producing sets of family pictures, wedding pictures, and general group pictures. We parse and store data for all group pictures with up to $k = 10$ people. The number of pictures with between two and four people for each type of group picture is 1051, 662, and 963 for family, wedding, and general group pictures respectively. We used even baseline scores for the experiments, which assumes an equal number of men and women. Although we do not have unique identifiers for people in the pictures, we checked that the data is indeed roughly gender balanced. Treating each person in each picture as unique, 51-52\% of the subjects are female in each of the three group picture datasets.

{
	In Figure~\ref{fig:photos} of the main text, the graph homophily indices we report are based on projecting all group pictures with up to 10 people to a graph based on co-appearance. For example, we reported graph homophily indices of 0.57 and 0.55 (for women and men, respectively) for wedding pictures.
	If we only project hyperedges of size 2 to 4 for wedding pictures, then the affinity scores are instead 0.45 for women and 0.40 for men, which are just below baseline values.
	Whether we project using all groups pictures or only small group pictures ($k = 2,3,4$), reducing group interactions to pairwise co-appearances overlooks meaningful information about the way in which homophily depends on the group size. One way to partially remedy this issue is to separately project different group sizes to different graphs, and then compute a graph homophily index for each separate graph. In Table~\ref{tab:ghi}, we report the graph homophily indices obtained by projecting size 2, 3, and 4 group pictures to a graph, as well as scores obtained from projecting all groups of size 2-4, and all groups of size 2-10. Columns $\alpha_W$ and $\alpha_M$ represent the class proportion (i.e., baseline scores) for women and men, respectively.
	\begin{table}
		\caption{
			{
				\textbf{Graph homophily indices obtained by projecting three different types of group picture datasets to graphs}. We report indices obtained from projecting groups of size 2, 3, and 4, separately. We also provide graph homophily indices obtained when projecting multiple groups sizes at once to a graph.
			}
		}
		\label{tab:ghi}
		\begin{tabular}{c |  lll | lll | lll }
			\toprule
			& \multicolumn{3}{c}{Family Pictures} & \multicolumn{3}{c}{Wedding Pictures} & \multicolumn{3}{c}{General Group Pictures} \\
			\toprule
			$k$ & women & men & $\alpha_W$ / $\alpha_M$ & women & men & $\alpha_W$ / $\alpha_M$ & women & men & $\alpha_W$ / $\alpha_M$\\
			\midrule 
			2 & 0.465 & 0.281 & 0.57 / 0.43 & 0.215 & 0.174 & 0.51 / 0.49 & 0.635 & 0.452 & 0.6 / 0.4 \\
			\midrule 
			3 & 0.365 & 0.406 & 0.48 / 0.52 & 0.617 & 0.531 & 0.55 / 0.45 & 0.626 & 0.567 & 0.54 / 0.46 \\
			\midrule 
			4 & 0.456 & 0.405 & 0.52 / 0.48 & 0.557 & 0.519 & 0.52 / 0.48 & 0.583 & 0.543 & 0.52 / 0.48 \\
			\midrule 
			2-4 & 0.426 & 0.395 & 0.52 / 0.48 & 0.452 & 0.394 & 0.52 / 0.48 & 0.588 & 0.543 & 0.53 / 0.47 \\
			\midrule 
			2-10 & 0.433 & 0.411 & 0.51 / 0.49 & 0.571 & 0.546 & 0.52 / 0.48 & 0.599 & 0.583 & 0.51 / 0.49 \\
			\bottomrule
		\end{tabular}
	\end{table}
	This approach of separately projecting different group sizes is already a departure from standard graph-based techniques for measuring homophily, which often project all group sizes at once. These separate graph projections do provide one way to observe the way homophily depends on group size. In wedding pictures, for example, these separated scores accurately capture the fact that homophily is not present for groups of size $k = 2$, but starts to become more present for groups of size $k = 3$ and $k = 4$. However, separately projecting different group sizes still does not capture the fact that ratio scores meaningfully differ within each group size $k$ depending on affinity type $t \leq k$. For example, our hypergraph scores capture the fact that in wedding group pictures with $k = 4$ people, type-$3$ affinity scores are below baseline while groups that are perfectly gender homogeneous or perfectly gender balanced are above baseline. 
	Similarly, in family pictures with 4 people, our hypergraph measures capture the fact that type-2 affinities are above baseline, which intuitively reflects a high proportion of 4-person pictures of a husband, wife, one boy, and one girl. Graph homophily indices do not capture this observation.}

%


\subsection{Checking Robustness of Affinity Scores}
\label{sec:bootstrap}
By design, the affinity scores and the impossibility results we have considered apply to a specific hypergraph with a fixed set of hyperedges. In the main text, we therefore used all available hyperedge information when plotting results for each dataset we considered.  
Building a hypergraph from real data can be a noisy and imperfect process, and affinity scores will therefore vary depending on the quality of data and availability of group information in each context. We use a simple bootstrapping procedure to show that the basic patterns in affinity, ratio, and normalized bias scores for all of the hypergraphs we consider are stable to perturbations in the data. Given a hypergraph $H$ with $m_k$ hyperedges of size $k$, we sample $m_k$ of these hyperedges with replacement and compute typed affinity, ratio, and normalized bias scores from the set of sampled hyperedges in each case. We repeat the process 100 times for each dataset, and then compute the mean and standard error for the resulting ratio scores. { Figures~\ref{fig:main3} and~\ref{fig:photos2} display bootstrapping results for ratio scores. The same observations about robustness apply to affinity scores and normalized bias scores as well.}  Results on paper co-authorship, congress bill co-sponsorship, TripAdvisor reviews, and Walmart shopping trips are particularly robust to perturbations in the data. Average scores when bootstrapping are nearly identical to scores obtain when using the entire hypergraph, and there is a very small standard error (Fig.~\ref{fig:main3}). Scores for group pictures vary slightly more depending on the sample, but still preserve the same shape and pattern. We observe the same basic differences among pictures with regard to group type (wedding, family, or general group picture) and size (Fig.~\ref{fig:photos2}).

\subsection{Experimental Results on Contact Data}
In addition to the empirical results shown in the main text, we compute affinity scores with respect to gender for group gatherings among primary and high schools students~\cite{Mastrandrea-2015-contact,Stehl-2011-contact,Benson-2018-simplicial}. Each hypergraph consists of groups of students (in primary school~\cite{Stehl-2011-contact} and high school~\cite{Mastrandrea-2015-contact} respectively) that gathered together in close proximity at some point during the day, as measured by wearable sensors. In high school student interactions of size two through four (Fig.~\ref{fig:high}), students have high tendencies for gathering in groups where all members are of the same gender. All other affinity scores are below baseline. Results for primary school interactions are nearly the same (Fig.~\ref{fig:primary}), except in the case of groups of size two, where groups with two male students are just slightly below baseline. We applied bootstrapping to both datasets (see Section~\ref{sec:bootstrap} for details), to confirm that our results are robust to perturbations in the data. For both hypergraphs we observe a small standard error around the mean affinity scores obtained across different subsamples of the hyperedges. 
\begin{figure}[t]
	\centering
	\includegraphics[width=.9\linewidth]{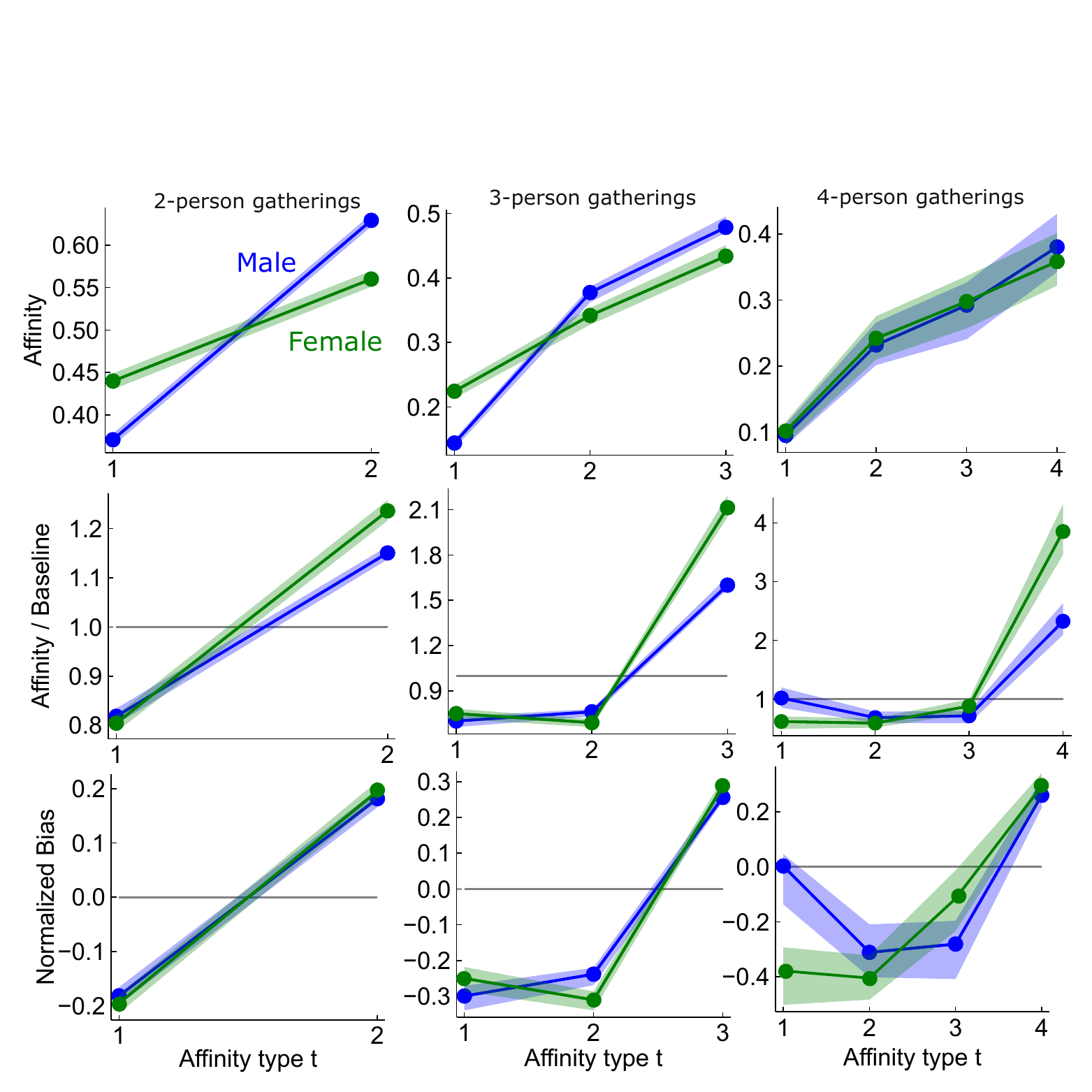} 
	\caption{\small 
		\textbf{Affinity, ratio, and normalized bias scores with respect to gender for a contact hypergraph at a high school}. Nodes in the hypergraph are students, and hyperedges indicate sets of students who were gathered in close proximity at some point, as measured by wearable sensors. Solid lines indicate affinity scores for the entire hypergraph, and lighter colored regions show one standard error around the mean affinity score obtained from a bootstrapping procedure on the hyperedges, indicating our results are robust to perturbations in the data. Both genders exhibit strong simple homophily, meaning a high tendency for gathering in groups where everyone is of the same gender. All other scores are below baseline. For gatherings of size three and four, ratio scores almost increase monotonically, though perfect monotonic increase is impossible, as shown by our theoretical results. However, monotonic increase in raw affinity scores is possible, as seen in plots from the first row. This results from having a roughly balanced number of gatherings of each type.}
	\label{fig:high}
\end{figure} 
\begin{figure}[t]
	\centering
	\includegraphics[width=.9\linewidth]{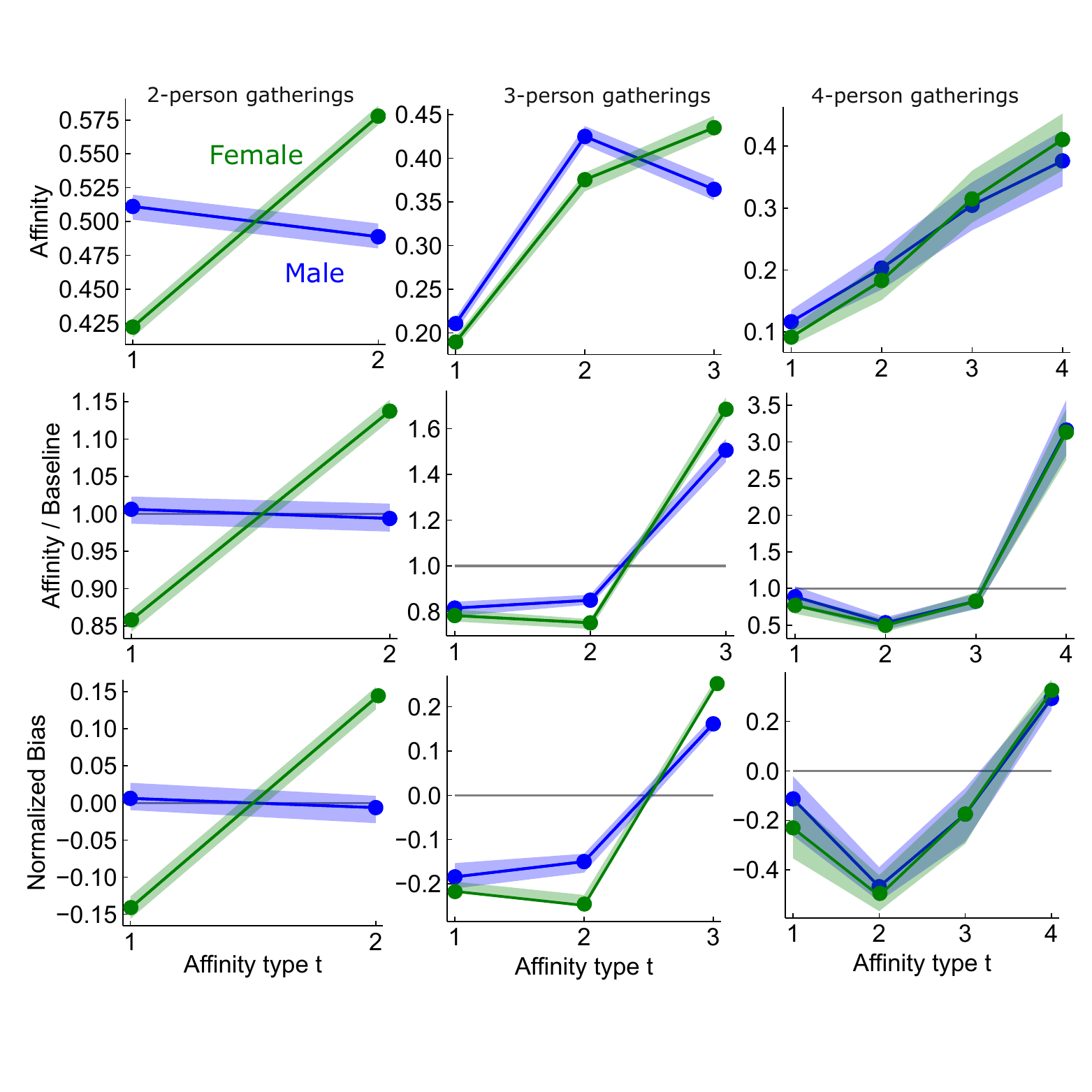} 
	\caption{\small 
		\textbf{Affinity, ratio, and normalized bias scores with respect to gender for a contact hypergraph at a primary school}. Hyperedges indicate a set of students (nodes) who were gathered in close proximity at some point, as measured by wearable sensors. 
		Solid lines again indicate affinity scores for the entire dataset, and lighter colored regions show one standard error around the mean affinity score from a bootstrapping.
		For groups of size three and four, both genders exhibit strong simple homophily, indicating a high tendency towards groups where everyone is of the same gender. For groups of size two, male students have affinity scores that are almost exactly equal to baseline. }
	\label{fig:primary}
\end{figure} 
\end{document}